\documentclass[11pt,reqno]{amsart}
\usepackage{amsmath,amsxtra,amssymb,amsthm,amsfonts,eufrak,bm}
\usepackage{hyperref}
\usepackage{cleveref}

\usepackage{tikz-cd}
\usepackage{float}
\usepackage{mathrsfs}
\usetikzlibrary{matrix,arrows,decorations.pathmorphing}
\usepackage{soul}

\newtheorem{lemma}{Lemma}[section]
\newtheorem{prop}[lemma]{Proposition}
\newtheorem{theorem}[lemma]{Theorem}
\newtheorem{cor}[lemma]{Corollary}

\newtheorem{prob}[lemma]{Problem}
\newtheorem{rem}[lemma]{Remark}

\newcommand{\re}{\begin{rem}\rm}
\newcommand{\mar}{\end{rem}}
\newtheorem{exam}[lemma]{Example}

\newtheorem{defi}[lemma]{Definition}

\newcommand{\qd}{\end{proof}\vspace{0.5ex}}
\newcommand{\prf}{\begin{proof}[\bf Proof:]}

\newcommand{\pl}{\hspace{.1cm}}

\newcommand{\kl}{\pl \le \pl}
\newcommand{\gl}{\pl \ge \pl}

\newcommand{\lel}{\pl = \pl}

\newcommand{\A}{{\mathcal A}}

\newcommand{\E}{{\mathcal E}}

\newcommand{\M}{{\mathcal M}}
\newcommand{\N}{{\mathcal N}}

\newcommand{\Q}{{\mathcal Q}}

\newcommand{\cM}{{\mathcal M}}
\newcommand{\cN}{{\mathcal N}}

\newcommand{\cQ}{{\mathcal Q}}

\newcommand{\bC}{{\mathbb C}}

\newcommand{\bE}{{\mathbb E}}

\newcommand{\bM}{{\mathbb M}}
\newcommand{\bN}{{\mathbb N}}

\newcommand{\bR}{{\mathbb R}}

\newcommand{\bT}{\mathbb{T}}

\newcommand{\bZ}{\mathbb{Z}}

\newcommand{\al}{\alpha}
\newcommand{\si}{\sigma}

\newcommand{\la}{\lambda}
\newcommand{\eps}{\varepsilon}

\newcommand{\norm}[2]{\parallel \! #1 \! \parallel_{#2}}
\newcommand{\ran}{\rangle}
\newcommand{\lan}{\langle}
\newcommand{\bra}[1]{\langle{#1}|}
\newcommand{\ket}[1]{|{#1}\rangle}
\newcommand{\ketbra}[1]{|{#1}\rangle\langle{#1}|}

\newcommand{\ten}{\otimes}

\DeclareMathOperator{\LSI}{LSI}

\DeclareMathOperator{\id}{\operatorname{id}}
\newcommand{\tr}{{\text{tr}}}

\DeclareMathOperator{\dom}{\operatorname{dom}}

\renewcommand{\div}{\text{div}}

\newcommand{\zz}{\mathbb{Z}}

\newcommand{\Mz}{\mathbb{M}}

\newcommand{\om}{\omega}
\newcommand{\MLSI}{\text{MLSI}}
\newcommand{\Log}{\text{Log}}
\newcommand{\ssubset}{\subset}
\newcommand{\nen}{n\in \bN}
\newcommand{\Prob}{\text{Prob}}
\newcommand{\Lip}{\text{Lip}}

\allowdisplaybreaks

\numberwithin{equation}{section}
\renewcommand{\log}{\ln}

\begin{document}
\title{Relative entropy decay and complete positivity mixing time}
\author{Li Gao}
\address{Department of Mathematics\\ University of Houston, Houston, TX 77204, USA
} \email[Li Gao]{lgao12@uh.edu}
\author{Marius Junge}
\address{Department of Mathematics\\
University of Illinois, Urbana, IL 61801, USA} \email[Marius Junge]{mjunge@illinois.edu}
\author{Nicholas LaRacuente}
\address{Department of Computer Science \\ Indiana University Bloomington, Bloomington, IN 47408, USA} \email[Nicholas LaRacuente]{nick.laracuente@gmail.com}
\author{Haojian Li}
\address{Zentrum Mathematik\\
Technische Universit\"at M\"unchen, Garching, 85748, Germany} \email[Haojian Li]{lihaojianmath@gmail.com}


\begin{abstract}We prove that the complete modified logarithmic Sobolev constant of a quantum Markov semigroup is bounded by the inverse of its complete positivity mixing time. For classical Markov semigroups, this implies that every sub-Laplacian given by a H\"ormander system on a compact manifold satisfies a uniform modified log-Sobolev inequality for matrix-valued functions. For quantum Markov semigroups, we obtain that the complete modified logarithmic Sobolev constant is comparable to spectral gap up to a constant as logarithm of dimension constant. This estimate is asymptotically tight for a quantum birth-death process. Our results and the consequence of concentration inequalities apply to GNS-symmetric semigroups on general von Neumann algebras.
\end{abstract}
\maketitle

\tableofcontents



\section{Introduction}
Over the past decades, functional inequalities have stimulated fruitful exchanges between analysis, geometry, and probability. Central to these interactions is the logarithmic Sobolev inequality, which connects the curvature conditions with the analytical properties like Sobolev inequality and spectral gap. Furthermore, it leads to probabilistic consequences such as the concentration phenomenon. In recent years, logarithmic Sobolev inequalities of quantum systems are studied for applications in quantum information theory and quantum many-body systems. In this work, we introduce a novel information-theoretic approach to modified logarithmic Sobolev inequalities and offer some new insights within both classical and quantum regimes.

Since the seminal works of Gross \cite{gross1,gross2}, logarithmic Sobolev inequalities have been widely studied on manifolds, graphs, as well as noncommutative spaces (see the surveys \cite{Ledoux,Gross14}). A natural framework of  logarithmic Sobolev inequalities is given by the theory of Markov semigroups, i.e. semigroups of positive unital maps on functions of some probability space $(\Omega,\mu)$. Recall that a Markov semigroup $(P_t)_{t\ge 0}:L_\infty(\Omega,\mu)\to L_\infty(\Omega,\mu)$ satisfies the \emph{logarithmic Sobolev inequality} ($\LSI$) if there is a positive constant $\al>0$ such that,
\begin{align}\label{eq:introLSI}
\int f^2 \ln f^2 d\mu - \left( \int f^2d\mu \right)  \ln  \left( \int f^2d\mu \right) \leq  \frac{2}{\al}\int (L f)f d\mu\pl ,\pl \forall \pl  f \in \text{dom}(L)\cap L_{2}(\Omega,\mu) \tag{$\LSI$}
\end{align}
where $\text{dom}(L)$ is the domain of the generator $L$ of the semigroup $P_t=e^{-tL}$. LSI was first introduced by Gross \cite{gross1} as a reformulation of Hypercontractivity (HC) \cite{nelson1966quartic,nelson1973construction} (see also \cite{bonami1968ensembles,bonami1970etude,rudin1960trigonometric,stam1959some} for earlier related results). That is, \eqref{eq:introLSI} is equivalent to
\begin{align}\label{eq:introHC}
\norm{P_t:L_2(\Omega,\mu)\to L_p(\Omega,\mu)}{}\le 1\pl  \text{ if } \pl p\le 1+e^{2\al t} \pl .  \tag{HC}
\end{align}
LSI is often called $L_2$-logarithmic Sobolev inequality, and it admits a $L_1$-variant called \emph{modified logarithmic Sobolev inequality} ($\MLSI$), which states that for any  positive $g \in \text{dom}(L)$,
\begin{align}\label{eq:introMLSI}
\int g\ln  g d\mu - \left( \int g d\mu\right)  \ln   \left( \int g d\mu \right) \leq \frac{1}{2\al} \int \left( L g \right)\ln g d\mu.  \tag{$\MLSI$}
\end{align}
These two forms of logarithmic Sobolev inequalities are closely related and in general \eqref{eq:introLSI} implies \eqref{eq:introMLSI}.
\eqref{eq:introMLSI} is equivalent to the exponential entropy decay that for any probability density $g$,
\begin{align}\label{eq:ed}
H(P_t g) \le e^{-2\al t} H(g)\pl,
\end{align}
where $H(g)=\int g\ln g \pl d\mu$ is the entropy functional. The exponential entropy decay \eqref{eq:ed} is a useful tool to derive the convergence time ($L_1$-mixing time) of the semigroup evolution. 
Conversely, Diaconis and Saloff-Coste in \cite{diaconis1996logarithmic} observed that the optimal LSI constant $\al_2$ is lower bounded by the inverse of $L_\infty$-mixing time. Later, Otto and Villani \cite{otto2000generalization} found that MLSI implies the transport cost inequalities, which were used by Talagrand \cite{Talagrand95} in deriving his concentration inequality. These implications later have also been extended to finite Markov chains on discrete spaces in \cite{maas2011gradient,erbar2012ricci}. To summarize, one has the following chain of functional inequalities:
\begin{align}\label{eq:diagram}
\boxed{L_\infty\text{-mixing time}}\Longrightarrow \boxed{\text{HC} \Leftrightarrow \text{LSI}}\Longrightarrow  \boxed{\text{MLSI}} \Longrightarrow \boxed{\text{Concentration}}
\end{align}

In a series of works \cite{olkiewicz1999hypercontractivity,kastoryano2013quantum,temme2014hypercontractivity,
rouze2019concentration,datta1,carlen2020non,gao2020fisher,wirth2018noncommutative}, the above chain of functional inequalities has been extended for {\bf ergodic} quantum Markov semigroups (QMS). Here, the ergodicity means the semigroup admits a unique invariant state, also called {\bf primitive} in mathematical physics literature.
While the study of classical Markov chains often focuses on  the ergodic cases, there are structural reasons to study non-ergodic semigroup on quantum systems. From physics perspective, the invariant states of a general quantum decoherence process may forms a non-trivial classical system, hence non-ergodic. Mathematically, one major motivation is the {\bf tensorization property}: given two (classical) Markov semigroups $(P_t)$ and $(Q_t)$, the tensor product semigroup satisfies
\begin{align}\al( P_t \ten Q_t)=\min \{ \al(P_t),\al(Q_t)\}\pl, \label{eq:tensor}\end{align}
where $\al(P_t)$ can be optimal constant $\al$ in either \eqref{eq:introLSI} or  \eqref{eq:introMLSI}. However,
while being crucial in the study of many-body systems, the tensorization property is false for MLSI constant of QMS \cite{brannan2021complete}, and in general open for LSI constant (see \cite{beigi2020quantum} for some limited examples).
The tensorization property motivates the following definition of the {\bf complete modified logarithmic Sobolev constant} (in short, {\bf CMLSI constant}),
\begin{align*} \al_c(P_t):= \inf_{n} \al_1(P_t\ten \id_{\mathbb{M}_n} )\pl, \end{align*}
where the infimum is over the optimal MLSI constant $\al_1(P_t \ten \id_{\mathbb{M}_n})$ of the semigroup $P_t$ tensoring with the identity map on arbitrary matrix algebra $\mathbb{M}_n$\footnotemark. \footnotetext{ This definition is sufficient for tensorization property between QMS over matrix algebras. In our discussion, we will consider a stronger definition that the infimum is over all ($\sigma$-)finite von Neumann algebras to ensure tensorization property over the general case. Whether these two definitions are equivalent remains open.} This notion of ``complete" MLSI is motivated from the completely bounded (CB) norm in operator space theory, and the study of it inevitably encounters non-ergodic semigroups, as $P_t\ten \id_{\mathbb{M}_n}$ is always invariant on the amplified part $\mathbb{M}_n$.

Using various methods, the positivity of CMLSI constant $\al_c$ has been recently established for both classical and quantum examples, including Heat semigroup on compact manifolds and finite graph \cite{li2020graph,brannan2022complete}, as well as QMS on matrix algebras \cite{gao2022complete} and group von Neumann algebras \cite{brannan2021complete,wirth2021complete}. Despite the progress, a direct approach to CMLSI via mixing time similar to the diagram \eqref{eq:diagram} is still open. In this work we fill this gap by establishing the following implications for noncommutative, non-ergodic QMS in the general setting of von Neumann algebras.
\begin{align}\label{eq:diagram2}
\boxed{\text{CB mixing time}}\Longrightarrow  \boxed{\text{CMLSI}} \Longrightarrow \boxed{\text{Concentration inequality}}
\end{align}
Here, we left out the LSI and HC because both fail for any non-ergodic QMS on matrix algebras \cite{bardet2022hypercontractivity}. We introduce a completely different approach to achieve \eqref{eq:diagram2}, which is illustrated below with applications in both classical and quantum settings.

\subsection{CMLSI via complete positivity }
Quantum Markov semigroups are noncommutative generalizations of Markov semigroups, where the underlying probability spaces are replaced by matrix algebras or operator algebras. 
We start with the relatively simple case of tracial von Neumann algebras. Let $\M$ be a finite von Neumann algebra equipped with a normal faithful tracial state $\tau$. A quantum Markov semigroup $(T_t)_{t\ge 0}:\M\to \M$ is a continuous semigroup of completely positive (CP) unital maps,  $(T_t)$ is called symmetric if
\[\tau(T_t(x)y)=\tau(yT_t(x)),\quad \forall t\geq 0 \pl, \pl x,y \in \M.\]
 Under the symmetric assumption, the fixed point space
\[\N=\{x\in \M \pl |\pl T_t(x)=x,t\geq 0 \}\]
forms a subalgebra, which admits a unique trace preserving conditional expectation $E:\M\to\N$ onto $\N$. Our main theorem is the following CMLSI via a CP mixing time. For two maps $\Phi$ and $\Psi$, we write $\Phi\le_{cp} \Psi$ if $\Psi-\Phi$ is completely positive.
\begin{theorem}\label{thm:main1}
Let $T_t=e^{-tL}:\M\to \M$ be a symmetric QMS with the generator $L$ and $E:\M\to \N$ be the conditional expectation onto its fixed point subalgebra $\N$. Define 
\[t_{cb}(\eps):=\inf\{t>0 \pl |\pl (1-\eps) E\le_{cp} T_t\le_{cp} (1+\eps)E \}. \]
Then the modified logarithmic Sobolev inequality
\begin{align}\label{eq:symmetricMLSI}
 \pl \tau(\rho \ln \rho-E(\rho)\ln E(\rho) )\le \frac{1}{2\al }\tau(L (\rho)\ln \rho)\pl, \pl  \forall \rho\ge 0, \rho\in \dom(L) \pl,
\end{align}
is satisfied for \[\al\ge \frac{1}{2t_{cb}(0.1)}.\]
Moreover, the same inequality holds for the semigroup $T_t\ten \id_{\cQ}$ with any finite von Neumann algebra $\cQ$.
\end{theorem}
Denote $\al_1$ as the optimal (largest) constant satisfying \eqref{eq:symmetricMLSI}, and $\al_c$ as the optimal CMLSI constant such that $L\ten \id_{\cQ}$ satisfying \eqref{eq:symmetricMLSI} for all $\cQ$. The above theorem states that
\begin{align}\label{eq:inversetcb} \al_1\ge \al_c\ge \frac{1}{2t_{cb}(0.1)}\pl.\end{align}
Here, $t_{cb}$ is the mixing time in terms of complete positivity, which was first introduced in \cite{gao2020fisher} under the name {\bf CB return time}, where the terminology was motivated comes from an equivalent definition via CB norm.

The form \eqref{eq:symmetricMLSI} of MLSI is equivalent to the exponential decay of relative entropy  that for any density operator $\rho$ (positive and trace 1)
\[ D(T_t(\rho)\|E(\rho))\le e^{-2\al t} D(\rho||E(\rho)) \pl. \]
where $D(\rho||\sigma)=\tau(\rho \ln \rho-\rho\ln \sigma)$ denotes the quantum relative entropy. The quantum relative entropy is a fundamental quantity that measures how well one can distinguish the state $\rho$ from $\sigma$ (see the survey \cite{vedral2002role} for the ubiquitous role of relative entropy in quantum information).
The key property behind the wide applications of relative entropy is the data processing inequality that for any quantum channel $\Phi$ (completely positive and trace preserving map)
\begin{align*} D(\Phi(\rho)||\Phi(\sigma))\le D(\rho||\sigma)\pl. \end{align*}
 From this perspective, MLSI is an improved data processing inequality that the relative entropy $D(\rho||E(\rho))$ from a state $\rho$ to its equilibrium $E(\rho)$ is not only non-increasing, but actually decays exponentially over time. Indeed, the proof of Theorem \ref{thm:main1} can be reduced from the following entropy contraction of symmetric quantum Markov maps.

\begin{theorem}\label{thm:main5} Let $\Phi$ be a symmetric completely positive unital map and let $E:\M\to\N$ be the conditional expectation onto the multiplicative domain of $\Phi$. Define $$k_{cb}(\eps):=\{\pl k\in \mathbb{N}\pl | \pl (1-\eps)E \le_{cp} \Phi^{2k}\le_{cp} (1+\eps)E \}.$$
 Then for any quantum state $\rho$,
 \begin{align} D(\Phi(\rho)||\Phi\circ E(\rho))\le (1-\frac{1}{2k_{cb}(0.1)}) D(\rho||E(\rho))\label{eq:introSDPI}\end{align}
 Furthermore, the same inequality holds for $\Phi\ten \id_{\cQ}$  for any finite von Neumann algebra $\mathcal{Q}$.
\end{theorem}
 In the ergodic cases ($\N=\mathbb{C}1$), such relative entropy contraction was considered in both classical \cite{del2003contraction} and quantum \cite{hirche2020contraction,muller2016entropy} settings. They are closely related to the strong data processing inequality studied in \cite{polyanskiy2017strong,du2017strong,lesniewski1999monotone,hiai2016contraction}. In particular, the existence of strict contraction in finite dimensions was obtained in \cite{gao2022complete}. The above theorem gives the first explicit estimate of  contraction coefficient for non-ergodic cases as well as the complete version (as the assertion extends to $\Phi\ten \id_{\cQ}$). This theorem is the discrete time analog of Theorem \ref{thm:main1}, which follows easily from applying \eqref{eq:introSDPI} to a semigroup map $T_t$ and taking the continuous-time limit. It is interesting that the proof of Theorem \ref{thm:main5} is purely information-theoretical and use only entropic quantities.

\subsection{CMLSI for matrix valued functions}
We first present the consequence of Theorem \ref{thm:main1} in classical cases. For a classical Markov semigroup $P_t:L_\infty(\Omega,\mu)\to L_\infty(\Omega,\mu)$, the notion of CMLSI is simply a uniform MLSI for all positive matrix-valued functions $g:\Omega\to \bM_n$ and $n\ge 1$,
\begin{align}\label{eq:classicCMLSI}
\mu\circ \tr (g\ln  g  - E_\mu(g) \ln   E_\mu(g))  \leq \frac{1}{2\al} \mu \circ \tr( (L g) \ln g )
\end{align}
where  $\mu(f)=\int fd\mu$ is the scalar valued mean, $E_\mu(g)=\int gd\mu\in \bM_n$ is the matrix valued mean, and $\tr$ is the standard matrix trace. In this setting, the CB return time $t_{cb}$ coincides with the {\bf $L_\infty$-mixing time}  
\[ t_{b}(\eps)=\{t >0| \norm{T_t-E_\mu:L_1(\Omega)\to L_\infty(\Omega)}{}\le \eps\}\pl,\]
which is accessible by standard kernel estimates. Indeed, $t_{b}$ is always finite if $P_t$ satisfies both ultra-contractivity and Poincar\'e inequality (spectral gap).
\begin{theorem}\label{thm:main9}
Let $P_t=e^{-Lt}:L_\infty(\Omega,\mu)\to L_\infty(\Omega,\mu)$ be an ergodic Markov semigroup symmetric to the probability measure $\mu$. Suppose
\begin{enumerate}
\item[i)] $P_t$ satisfies $\lambda$-Poincar\'e inequality for some $\lambda>0$: for $f\in \dom(L^{1/2})$,
\begin{align} \lambda\mu(|f-E_\mu(f)|^2)\le \int f(Lf)d\mu \pl.\label{eq:label}\end{align}
\item[ii)] There exists $t_0>0$ such that \begin{align}\label{eq:ultracon}\norm{P_{t_0}: L_1(\Omega,\mu)\to L_\infty(\Omega,\mu)}{}\le C_0\end{align}
\end{enumerate}Then
\begin{align}\al_1\ge \al_{c}\ge \frac{\lambda}{2(\lambda t_0+\ln C_0+2 )}\pl. \label{eq:CMLSIclassical}\end{align}
\end{theorem}
For finite Markov chains, a similar bound for LSI constant was proved by Diaconis and Saloff-Coste \cite{diaconis1996logarithmic} that $\al_1\ge \al_{2}\ge \frac{\lambda}{2(\lambda t_{0}+\ln C_{0}+\ln 10)}$.
In particular, this implies $\al_1\ge \al_2\ge \frac{2}{t_{b}(e^{-2})}$, which can be compared to our bound \eqref{eq:inversetcb}.
We remark that such $\Omega(t_{b}^{-1})$ lower bounds are asymptotically tight for $\al_1,\al_2$ and $\al_c$ of graph Laplacians on the odd cycles. It is quite interesting that Diaconis and Saloff-Coste's argument relies on the assumption i) and ii) above (which always holds for finite Markov chains), whereas our information-theoretic approach does not need such assumptions and applies to infinite dimensional, noncommutative systems. Combining these two results, we have the following comparison between $\al_c$ and $\al_2$ for ergodic finite Markov chains (Corollary \ref{eq:loglog}),
\[  \al_c\ge \frac{2\al_2}{3(4+\log\log \norm{\mu^{-1}}{\infty})}\pl.\]

Beyond finite state spaces, an immediate consequence of Theorem \ref{thm:main9} is the CMLSI for sub-Laplacian generators of H\"ormander systems.
\begin{cor}\label{cor:main3} Let $(M,g)$ be a compact connected Riemannian manifold without boundary and let $\omega d\operatorname{vol}$ be a probability measure with a smooth density $\omega$ with respect to the volume form $d\operatorname{vol}$. Suppose  $H=\{X_i\}_{i=1}^k\subset TM$ is a family of vectors fields satisfies the \emph{H\"{o}rmander's condition} that at every point $x\in M$,
\[
T_xM=\operatorname{span}\{[X_{i_1},[X_{i_2},\cdots, [X_{i_{n-1}}, X_{i_n}]]] \pl  | \pl 1\leqslant i_1,i_2\cdots i_n\leqslant k  , n\ge 1\}. \tag{\text{H\"{o}rmander condition}}
\]
Then the horizontal heat semigroup $P_t=e^{-\Delta_H t}$ generated by  the sub-Laplacian
\[
\Delta_H=\sum_{i} X_i^*X_i
\]
has CMLSI constant $\al_c(\Delta_H)>0$. Here $X_i^*$ is the adjoint operator with respect to $L_2(M, \omega d\operatorname{vol})$.
\end{cor}
Before this work, the only known example of sub-Laplacian satisfying CMLSI is the canonical sub-Laplacian on $SU(2)$ . The difficulty of CMLSI for sub-Laplacians was the lack of curvature. Indeed, the CMLSI for Heat semigroup was obtained using a non-commutative Bakry-\'Emery curvature dimension condition \cite{li2020graph,brannan2022complete}. Nevertheless, in the sub-elliptic case the Ricci curvature in the degenerate direction of the vector field $H=\{X_i\}_{i=1}^k$ can be interpreted as $-\infty$. For $SU(2)$, the CMLSI of sub-Laplacian was obtained in \cite{gao2022complete2} using certain gradient estimate as a quasi-curvature condition. Such gradient estimate was first introduced by Driver and Melcher \cite{driver2005hypoelliptic} for Heisenberg group, later obtained for nilpotent Lie
groups \cite{melcher2008hypoelliptic} and $SU(2)$ \cite{baudoin2009subelliptic}. Our Corollary \ref{cor:main3} asserts CMLSI for all sub-Laplacian of H\"ormander systems, without any curvature condition.

For scalar-valued functions, the positivity of $\al_2(\Delta_H)$ was proved by \L ugiewicz and Zegarli\'nski \cite{lugiewicz2007coercive}, using a similar hyper-contractive argument from \cite{diaconis1996logarithmic}. In particular, both \cite{diaconis1996logarithmic} and \cite{lugiewicz2007coercive} rely on the Rothaus Lemma that
\begin{align} \mu(f^2\ln f^2)\le \mu(\hat{f}^2\ln \hat{f}^2)+\norm{\hat{f}}{2}^2\pl, \label{eq:rothaus1}\end{align}
where $\hat{f}=f-E_\mu(f)$ is the mean zero part of $f$. 
We remark that both approaches fail for matrix valued functions.
\begin{prop}\label{prop:main4}
Let $P_t=e^{-tL}:L_\infty(\Omega,\mu)\to L_\infty(\Omega,\mu)$ be an ergodic symmetric Markov semigroup. Suppose either of the following inequalities holds for all $f\in L_\infty(\Omega,\bM_2)\cap \text{dom} (L)$:
\begin{align}
&\al\pl  \mu\circ\emph{\tr} \Big(|f|^2\ln |f|^2- E_\mu(|f|^2)\ln E_\mu(|f|^2)\Big)\pl\le\pl  \mu\circ \emph{\tr} (f^* (L\ten \id_{\bM_2 })f)\pl, \tag{LSI}\\
&\norm{(P_t\ten \id_{\bM_2 }) f}{L_2(\bM_2,L_{p(t)}(\Omega))}\pl \le \pl \norm{f}{L_2(\bM_2,L_2(\Omega))} \emph{\text{ for }} \pl p(t)=1+e^{2\al t}\pl, \tag{HC}\\
&\al \mu\circ\emph{\tr} \Big(|f|^2\ln |f|^2- E_\mu(|f|^2)\ln E_\mu(|f|^2)\Big)\le \mu\circ\emph{\tr} \Big(|\hat{f}|^2\ln |\hat{f}|^2- E_\mu(|\hat{f}|^2)\ln E_\mu(|\hat{f}|^2)\Big)+\norm{\hat{f}}{2}^2\pl, \tag{\text{Rothaus}}
\end{align}
then $\al=0$.
\end{prop}
Here, $L_2(\bM_2,L_{p(t)}(\Omega))$ is noncommutative vector valued $L_p$ space introduced by Pisier \cite{pisier1998non}, which was used in \cite{beigi2016hypercontractivity} in defining the CB hypercontractivity. Here, the failure of LSI, HC and Rothaus lemma indicates that there is no hope to obtain the diagram \eqref{eq:diagram} for matrix valued functions. In contrast, our Theorem \ref{thm:main9} shows that CMLSI is the only appropriate notion for logarithmic Sobolev inequalities for matrix valued functions.

\subsection{GNS-symmetric quantum Markov semigroup}
Although the trace symmetric case is broad enough to include all classic Markov chains under detail balance condition,
most physical models of quantum systems have the equilibrium as a non-tracial state. For instance, the quantum Gibbs state at finite temperature is never a trace. The tracial state corresponds to the infinite temperature only in the theoretical limit.
From the mathematical perspective, a general von Neumann algebra may not even admit a non-trivial trace, called Type III von Neumann algebra. These von Neumann algebras models the system in quantum field theory and the only symmetric condition we can hope there is for a state.

It turns out that relative entropy decay formulation of MLSI seamlessly extends to Type III setting under a state symmetric condition. We note that such difference between trace symmetric and state symmetric does not exist for classic systems as every measure gives a trace.
\begin{theorem}\label{thm:main6}
Let $\M$ be a $\sigma$-finite von Neumann algebra equipped with a normal faithful state $\phi$. Suppose $T_t=e^{-tL}:\M\to \M$ is a quantum Markov semigroup that is GNS-symmetric to $\phi$,
\begin{align}\label{eq:introGNSsymmetry}
\phi(xT_t(y))=\phi(T_t(x)y)\pl, \pl  \forall x, y\in \M , \ t\ge 0\pl.
\end{align}
Denote $E:\M\to\N$ as the $\phi$-preserving conditional expectation onto the fixed point subalgebra $\N$ and the CB return time $t_{cb}:=t_{cb}(0.1)$ as before. Then for any state $\rho$ and $t\ge 0$,
\begin{align}D(T_{t,*}(\rho)||E_*(\rho)) \kl e^{-\frac{t}{t_{cb}}}D(\rho||E_*(\rho))\pl .\label{eq:introexpdecayintro}\end{align}
The same inequality holds for $T_t\ten \id_{\cQ}$ for any $\sigma$-finite von Neumann algebra $\cQ$.
\end{theorem}

The exponential decay of relative entropy \eqref{eq:introexpdecayintro} can be viewed as an equivalent definition of MLSI in Type III cases\footnotemark.\footnotetext{The equivalent logarithmic Sobolev form in Type III setting is studied in a paper under preparation by Vernooij, Wirth and Zhang.} The technical tool to extend Theorem \ref{thm:main1} to Theorem \ref{thm:main6} is the Haagerup reduction \cite{haagerup2010reduction}, which is a method to derive results for type {III} von Neumann algebras by reducing to the cases of tracial von Neumann algebras. We remark that a similar result to Theorem \ref{thm:main5} also holds for GNS symmetric quantum channel (see Theorem \ref{thm:GNSsymmetric}).

In finite dimensions, the positivity of CMLSI constant was proved \cite{gao2022complete} that
\begin{align}  \al_c\ge \frac{\lambda}{2C_{cb}(E)}\label{eq:GRbound1}\end{align}
where $\lambda$ is the spectral gap of the generator and $C_{cb}(E)$ is the completely Pimsner-Popa index
\[C(E)=\inf \{c>0\pl |\pl  x\le c\:E(x)\pl,  \forall \pl x\in \M_+\}\pl,  \pl C_{cb}(E):=\sup_{\cQ}C(E\ten \id_{\cQ})\pl. \]
As an application of Theorem \ref{thm:main6}, we have the following improvement to \eqref{eq:GRbound1}.
\begin{cor}\label{cor:main7}
For GNS-symmetric quantum Markov semigroups,
\begin{align}\al_1\ge \al_{c}\ge \frac{\lambda}{2\ln  (10 C_{cb}(E))}\label{eq:introfdQMS}.\end{align}
\end{cor}
We construct an example of non-commutative birth-death process showing that the above estimate is tight up to some absolute constants.

\subsection{Concentration inequalities}One important application of MLSI is to derive concentration inequalities.
Such implication is discovered by Otto and Villani \cite{otto2000generalization} for classical Markov semigroups (see also the work of Erbar and Maas \cite{erbar2012ricci} for the discrete case), and recently extended to the noncommutative setting in \cite{rouze2019concentration,gao2020fisher,carlen2020non}. As an application of our MLSI for GNS-symmetric semigroups, we obtain concentration inequalities for a general invariant state $\phi$. Recall that the Lipschitz semi-norm
 \[ \|x\|_{\Lip} \lel: \max\{ \norm{\Gamma_L(x,x)}{}^{\frac{1}{2}}\pl, \pl \norm{\Gamma_{L}(x^*,x^*)}{}^{\frac{1}{2}}\} \]
is defined through the \emph{gradient form} (or Carr\'e du Champ operator)
 \[ \Gamma_L(x,y) \lel \frac{1}{2}\Big(L(x^*)y+x^*L(y)-L(x^*y)\Big)\pl, \pl \forall x,y\in \text{dom}(L)\pl. \]
\begin{theorem}\label{thm:main8} Let $\M$ be a $\sigma$-finite von Neumann algebra and $T_t=e^{-tL}$ be a GNS-$\phi$-symmetric quantum Markov semigroup. Suppose $T_t$ satisfies MLSI with parameter $\al>0$. There exists a universal constant $c$ such that for $2\le p <\infty$
 \[ \al\|x-E(x)\|_{L_p(\M,\phi)}\kl c\sqrt{p}\norm{x}{\emph{Lip}} \pl .\]
 Moreover, for any $t>0$, there exist a projection $e$ such that
 \[ \|e(x-E(x))e\|_{\infty} \le t \quad \mbox{and} \quad \phi(1-e)\le 2 \exp( -\frac{\al^2t^2}{16ec^2\norm{x}{\emph{Lip}}^2}) \pl .\]
\end{theorem}
One example is the following matrix concentration inequality that can be compared to the work of Tropp \cite{tropp2015introduction}.
\begin{exam}\label{exam:main11}{\rm Let $S_1,\cdots,S_n$ be an independent sequence of random $d\times d$-matrices $S_1,\cdots,S_n$ such that
$\norm{S_i-\mathbb{E}S_i}{\infty}\le M \pl , \pl a.e.$. We have the matrix Bernstein inequality that for the sum $Z=\sum_{k=1}^nS_k$
\[ \bE\norm{Z-\bE Z}{\infty}\le 2ce^{-1/2} \sqrt{ (v(Z)+M^2)\log d}\]
and  the matrix Chernoff bound
\[ P(|Z-\mathbb{E}Z|>t)\le 2 d\exp\Big( -\frac{t^2}{64ec^2(v(Z)+M^2)}\Big)\pl.\]
where
\[v(Z)=\max \{ \norm{\bE((Z-\mathbb{E}Z)^*(Z-\mathbb{E}Z))}\pl, \norm{\bE((Z-\mathbb{E}Z)^*(Z-\mathbb{E}Z))}{}\}\pl.\]}
\end{exam}



\subsection{Outline of the paper} We organize our paper as follows to make it accessible for readers from different backgrounds.  In Section 2, we briefly review some quantum information basics in setting of tracial von Neumann algebras. We prove our key entropy difference lemma (Lemma \ref{lemma:difference}) and the improved data processing inequality (Theorem \ref{thm:unital}, a slightly stronger version of Theorem \ref{thm:main5}).  Based on that, we discuss the functional inequalities of symmetric quantum Markov semigroups in Section 3,  prove the main Theorem \ref{thm:main1} and its consequences Theorem \ref{thm:main9} and Corollary \ref{cor:main3} for classical Markov semigroups.  We also illustrate the failure of matrix-valued logarithmic Sobolev ineuqality  Propositon \ref{prop:main4} (restated as Proposition \ref{prop:rothaus}). The discussion to this point does not involves much technicality beyond the basic concepts of finite von Neumann algebras, for which the readers are welcome to think the matrix algebra $\bM_n$ and matric valued functions $L_\infty(\Omega, \bM_n)$ as examples.

After that in Section 4 we dive into the GNS-symmetric cases, for which we discuss the Haagerup reduction for channels and entropic quantities, and derive Theorem \ref{thm:main6} (restated as Theorem \ref{thm:GNSsemigroup}) and Corollary \ref{cor:main7} (see Section \ref{sec:symmetricQMS}). Section 5 collects applications Theorem \ref{thm:main8} (restated as Theorem \ref{thm:concentration}) and Example \ref{exam:main11} (restated as Example \ref{exam:randommatrix}) of our general results. We conclude the paper in Section 6 with some discussions on questions left in open.\\

\noindent {\bf Notations.} We use calligraphic letters $\M,\N$ for von Neumann algebras, and denote $\bM_n$ as the algebra of $n\times n$ as complex matrices. We use $\tau$ as the trace on von Neumann algebra,  and $\tr$ as the standard matrix trace. The identity operator is denoted by $1$, and the identity map between spaces is denoted as $\id$, sometimes specified with subscript like $1_\M$ and $\id_{\bM_n}$. We write $ a^*$ as the adjoint element of $a$ and $\Phi_*$ for a pre-adjoint map of $\Phi$.\\

\noindent {\bf Acknowledgement.} The research of LG is partially supported by NSF grant DMS-2154903.
LG is also grateful to the support of AMS-Simons Travel Grants. Nicholas LaRacuente was supported as an IBM Postdoc at The University of Chicago. MJ was partially supported by NSF Grant DMS 1800872 and NSF RAISE-TAQS 1839177. HL acknowledges support by the DFG cluster of excellence 2111 (Munich Center for Quantum Science and Technology).


\section{Entropy Contraction of Symmetric Markov maps}\label{sec:unital}
\subsection{States, channels and entropies}
We briefly review some basic information-theoretic concepts in the noncommutative setting.
 Recall that a von Neumann algebra $\M$ is a unital $*$-subalgebra of $B(H)$ closed under weak$^*$-topology. A linear functional $\phi:\M\to\mathbb{C}$ is called a state if it is positive $\phi(x^*x)\ge 0$ for any $x\in \M$ and in addition  $\phi(1)=1$.
 We say $\phi$ is normal if $\phi$ is weak$^*$-continuous. Throughout the paper, we will only consider normal states and denote $S(\M)$ as the normal state space of $\M$. We write $s(\phi)$ as the support projection of a state $\phi$, which is the minimal projection $e$ such that $\phi(x)=\phi(exe)\pl, \forall \pl x\in \M$. A normal state $\phi$ is faithful if $s(\phi)=1$. For two normal states $\rho$ and $\sigma$, the relative entropy
is defined as
\begin{align}D(\rho||\sigma)=\begin{cases}
                    \lan \xi_\rho| \log\Delta(\rho/\sigma) |\xi_\rho \ran\pl, & \mbox{if } s(\rho)\le  s(\sigma)\\
                    +\infty, & \mbox{otherwise}.
                  \end{cases}\pl,\label{eq:araki}\end{align}
where $\xi_\rho$ is a  vector implementing the state $\rho$ and $\Delta(\rho/\sigma)$ is the relative modular operator.  
This form of definition \eqref{eq:araki} was introduced by Araki \cite{araki1976relative} for general von Neumann algebras.

In this section, we will focus on the case that $\M$ is a finite von Neumann algebra. Namely, $\M$ is equipped with a normal faithful tracial state $\tau$. The tracial noncommutative $L_p$-space $L_p(\M,\tau)$ is defined as the completion of $\M$ with respect to the $p$-norm $\norm{a}{p}=\tau(|a|^{p})^{1/p}$. We identify $L_\infty(\M)\cong \M$, and also $ L_1(\M)\cong \M_*$ via the trace duality
\begin{align*}d_\phi\in L_1(\M)\longleftrightarrow \phi\in \M_*,\pl  \phi(x)=\tau(d_\phi x)\pl.\end{align*}
We say $\rho \in L_1(\M)$ is a density operator if $\rho\ge 0$ and $\tau(\rho)=1$, which corresponds to a normal state in the above identification. For example, the identity element $1$ is the density operator corresponding to the trace $\tau$ itself.
We will often identify normal states with their density operators if no ambiguity. Via this identification, relative entropy reduces to the original definition of Umegaki \cite{umegaki1962conditional},
\[D(\rho||\sigma)=\tau(\rho\log \rho-\rho\log \sigma )\]
provided this trace is well-defined. For example, for $\rho$ and $\sigma$ in the bounded state space
\begin{align*}
S_{b}(\M)=\{\rho\in S(\M)\pl |\pl  \mu_1 1\le \rho \le \mu_2 1\text{ for some } \mu_1,\mu_2>0 \}.
\end{align*}
the Umegaki's formula is always well-defined and finite. For this reason, we will mostly work with bounded states from $S_{b}(\M)$ and derive results for general case $S(\M)$ by approximation.
 When the second state $\sigma=1$, this gives the entropy functional
\[H(\rho):=D(\rho|| 1)=\tau(\rho\log\rho)\pl.\] Note that the standard convention of von Neumann entropy in quantum information literature is often with an additional negative sign .

We say a linear map $T:\M\to \M$ is a quantum Markov map if $T$ is normal, unital and completely positive. The pre-adjoint map $T_*:\M_*\to \M_*$ is called a quantum channel, which sends normal states to normal states. In the tracial setting, $T_*:L_1(\M)\to L_1(\M)$ given by \[\tau(T_*(\rho)y)=\tau(\rho T(y))\pl, \pl \forall\pl  y\in \M, \rho\in L_1(\M)\] is completely positive and trace preserving (in short, CPTP). One fundamental inequality about quantum channel is the data processing inequality (also called monotonicity of relative entropy)
\begin{align} D(\rho||\sigma)\ge D(T_*(\rho)||T_*(\sigma))\pl, \pl \forall \rho,\sigma\in S(\M)\pl. \label{eq:DPI} \end{align}
The data processing inequality states that two quantum states cannot become more distinguishable under a quantum channel. We note that the above inequality remains valid for $T$ being positive but not necessarily completely positive \cite{MR17,frenkel2022integral}. The main technical result of this work is an improved data processing inequality for quantum channels under symmetric conditions.

\subsection{Entropy contraction for unital quantum channels}\label{sec:unital}
We start our discussion on entropy contraction of unital quantum channels. A quantum channel $\Phi:L_1(\M)\to L_1(\M)$ is unital if $\Phi( 1)=1$. The restriction of $\Phi$ on $\M$ is bounded and normal, thus $\Phi$ can be viewed as the $L_1$-norm extension of its restriction $\Phi:\M\to\M$. By duality, its adjoint $\Phi^*:\M\to \M$
 is a trace preserving quantum Markov map hence also extends to a unital quantum channel.

 For a state $\rho$ with $H(\rho)<\infty$, we define the entropy difference of $\Phi$,
\[D_\Phi(\rho):=H(\rho)-H(\Phi(\rho))\pl.\]
Note that $D_\Phi(\rho)\ge 0$ because of data processing inequality \eqref{eq:DPI} and $\Phi(1)=1$,
\[H(\rho)=D(\rho||1)\geq D(\Phi(\rho)||\Phi(1))= H(\Phi(\rho))\pl.\]
We start with a simple but key lemma in our argument.
\begin{lemma}[Entropy difference lemma]\label{lemma:difference} Let $\Phi:L_1(\M)\to L_1(\M)$ be a unital quantum channel and $\Phi^*$ be its adjoint. Then for two bounded states $\rho,\omega\in {S_B}(\M)$,
 \[ D(\rho\|\Phi^*\Phi(\omega))
 \le D_\Phi(\rho)  + D(\rho\|\omega) \le \tau((\id-\Phi^*\Phi)(\rho)\ln \rho ) + D(\rho\|\omega)  \pl .\]
\end{lemma}
\begin{proof} By duality, $\Phi^*$ is also completely positive unital.
Then,
\begin{align*}
D(\rho\|\Phi^*\Phi(\omega))=&\tau(\rho\ln \rho -\rho \ln \Phi^*\Phi(\omega))
\\ =\pl&\tau(\rho\ln \rho-\Phi(\rho)\log \Phi(\rho))+\tau( \Phi(\rho)\log \Phi(\rho)-\rho \ln \Phi^*\Phi(\omega))
\\ =\pl&D_\Phi(\rho)+\tau\big( \Phi(\rho)\log \Phi(\rho)-\rho \ln \Phi^*\Phi(\omega)\big)
\\ \overset{(1)}{\le}\pl &D_\Phi(\rho)+\tau\Big( \Phi(\rho)\log \Phi(\rho)-\rho \Phi^*
\big(\ln \Phi(\omega)\big)\Big)
\\ = \pl&D_\Phi(\rho)+\tau\Big( \Phi(\rho)\log \Phi(\rho)-\Phi(\rho)\ln \Phi(\omega)\Big)
\\ =\pl &D_\Phi(\rho)+D(\Phi(\rho)\|\Phi(\omega))
\\ \overset{(2)}{\le}\pl &D_\Phi(\rho)+D(\rho\|\omega)\pl,
\end{align*}
where (2) follows from the monotonicity of relative entropy.
The inequality (1) uses the operator concavity \cite{choi74} of logarithm function $t\mapsto \ln t$ that for any positive operator $x\ge 0$,
\[\Phi^*
(\ln x)\le \ln \Phi^*(x)\pl .\]
 This proves the first inequality in the assertion. For the second part, it suffices to note that
\begin{align*}
D_\Phi(\rho)=&\tau(\rho\log \rho-\Phi(\rho)\log \Phi(\rho))
\le \tau(\rho\log \rho-\Phi(\rho)\Phi(\log \rho))
=\tau(\rho\log \rho-\Phi^*\Phi(\rho)\log \rho),
\end{align*}
where we use the operator concavity $\Phi(\ln x)\le \ln \Phi(x)$ again.
\end{proof}
Our idea is to iterate the above lemma as follows,
\begin{align*}D(\rho|| (\Phi^* \Phi)^n(\rho))\le\pl D_\Phi(\rho)+D(\rho|| (\Phi^* \Phi)^{n-1}(\rho))
 \le\pl n D_\Phi(\rho)+D(\rho|| \rho)=n D_\Phi(\rho)
\end{align*}
Then a relevant question is what would be the limit of $(\Phi^* \Phi)^n(\rho)$ as $n \to\infty$. This leads to the multiplicative domain of $\Phi$. Recall that the multiplicative domain of a unital completely positive map $\Phi$ is
\[\N_\Phi=\{ x\in\M \ |\ \Phi(y) \Phi(x)=\Phi(yx)\pl,\pl  \Phi(x) \Phi(y)= \Phi(xy),\forall y\in\mathcal{M}\}\pl.\]
When $\Phi$ is normal, $\N_\Phi\subset\mathcal{M}$ is a von Neumann subalgebra. A linear map $E:\M\to \M$ is called a conditional expectation if $E$ is a unital completely positive map and idempotent $E\circ E=E$. When $\M$ is a finite von Neumann algebra, for any subalgebra $\N\subset \M$, there always exists a (unique) trace preserving conditional expectation $E$ onto $\N$ such that
\begin{align}\tau(xy)=\tau(xE(y))\pl, \pl x\in \N, y\in \M\pl. \label{eq:condexp}\end{align}
In particular, $E$ is also a unital quantum channel. \begin{prop}\label{prop:cd}Let $\Phi:L_1(\M)\to L_1(\M)$ be a unital quantum channel and let $E:\M\to \N$ be the trace preserving conditional expectation onto the multiplicative domain $\N:=\N_\Phi$. Then
 \begin{enumerate}
 \item[i)] $\Phi:\N\to\Phi(\N)$ is a $*$-isomorphism with inverse $\Phi^*:\Phi(\N)\to \N$. Moreover, $\Phi(\N)$ is the multiplicative domain for $\Phi^*$, and
\begin{align}\label{eq:cd2} (\Phi^* \Phi)\circ E=E\circ (\Phi^* \Phi)=E\pl, \pl   E_0\circ \Phi=\Phi\circ E\pl , \end{align}
 where $E_0:\M\to \Phi(\N)$ is the trace preserving conditional expectation onto $\Phi(\N)$.
\item[ii)] $\Phi$ is an isometry on $L_2(\N)$. If in addition $\norm{\Phi(\id-E):L_2(\M)\to L_2(\M)}{2}<1$, then $E=\lim_{n}(\Phi^* \Phi)^n$ as a map from $L_2(\M)$ to $L_2(\M)$ .
    \end{enumerate}
 \end{prop}
\begin{proof} It is clear that $\Phi$ is a $*$-homomorphism on $\N$.
For any $x,y\in L_2(\N)\subset L_2(\M)$,
\[\tau( y(\Phi^*\circ \Phi)(x))=\tau( \Phi(y) \Phi(x))=\tau(\Phi(xy))=\tau(xy)\]
Thus $\Phi^*\circ \Phi|_{\N}=\id_\N$ is the identity map. This verifies $(\Phi^* \Phi)\circ E=E$. Since $E^*=E$, $E\circ (\Phi^* \Phi)=E$ follows from taking adjoint. Thus $\Phi:\N\to\Phi(\N)$ is a $*$-isomorphism with inverse $\Phi^*$. Denote $\N_0$ as the multiplicative domain for $\Phi^*$, we have $\Phi(\N)\subset \N_0$. Conversely, we also have $\Phi^*(N_0)\subset \N$ by
switching the role of $\Phi=(\Phi^*)_*$. Then $\Phi(\N)=\N_0$ since $\Phi$ is bijective on $\N$.
For ii), we note that by \eqref{eq:cd2}
\begin{align*}(\id-E)\Phi^* \Phi(\id-E)=(\id-E)(\Phi^* \Phi-E)=\Phi^* \Phi-E, \pl
&\pl (\Phi^* \Phi-E)^n=\Phi^* \Phi^n-E.
\end{align*}
Therefore,
\begin{align*}&\norm{\Phi^* \Phi-E:L_2(\M)\to L_2(\M)}{}=\norm{\Phi(\id-E)}{2}^2<1,
\\
&\norm{(\Phi^* \Phi)^n-E:L_2(\M)\to L_2(\M)}{}=\norm{(\Phi^* \Phi-E)^n:L_2(\M)\to L_2(\M)}{}=\norm{\Phi(\id-E)}{2}^{2n},
\end{align*}
which goes to $0$ as $n\to \infty$.
\end{proof}

In order to estimate entropic quantities, we will use the approximation in terms of complete positivity.
Recall that for a density operator $\sigma\in S(\M)$ with full support, the Bogoliubov-Kubo-Mori (BKM) metric for $X\in\M$ is defined by
\begin{align*}\gamma_{\sigma}(X)=\int_{0}^\infty \tau(X^*(\sigma+s)^{-1}X(\sigma+s)^{-1})ds\pl.\end{align*}
The BKM metric is a Riemannian metric on the state space $S(\M)$ that is monotone under any quantum channel $\Psi$,
\[ \gamma_{\Psi(\sigma)}(\Psi(X))\le \gamma_{\sigma}(X)\pl, \forall X\in \M\pl. \]
It connects to the relative entropy as follows
\begin{align} D(\rho||\sigma)=\int_0^1 (1-t)  \gamma_{\rho_t}(\rho-\sigma)dt\pl, \label{eq:st}\end{align}
where $\rho_t=t\rho+(1-t) \sigma$ for $t\in [0,1]$.
It is proved in \cite[Lemma 2.1 \& 2.2]{gao2022complete} that if $\rho\le c\sigma$,
\begin{align} &c\gamma_{\rho}(X)\le \gamma_{\sigma}(X)\pl, \pl \forall X\in \M \nonumber\\    &k(c) \gamma_{\sigma}(\rho-\sigma)  \le D(\rho||\sigma) \le \gamma_{\sigma}(\rho-\sigma) \label{eq:keylemma}
\end{align}
where $k(c)=\frac{c\ln c-c+1}{(c-1)^2}$. The above discussion remains valid if $s(\rho)\le s(\sigma)$ and $X\in s(\sigma)\M s(\sigma)$. For two  positive maps $\Psi$ and $\Phi$, we write $\Phi\le \Psi $ if $\Psi-\Phi$ is positive.
\begin{lemma}\label{lemma:approximate}
Let $E$ be a conditional expectation (not necessarily trace preserving) and $\Psi$ be a quantum Markov map such that
\[ (1-\eps)E\le \Psi \le (1+\eps)E\pl.\]
 Assume that $E\circ \Psi=E$. Then for any $\rho\in S(\M)$, 
\[D(\rho||\Psi_*(\rho))\ge \Big(\frac{1-\eps}{1+\eps}-\frac{\eps}{(1-\eps)k(2)}\Big) D(\rho||E_*(\rho))\]
In particular, for $\eps=\frac{1}{10}$, \[D(\rho||\Psi_*(\rho))\ge \frac{1}{2} D(\rho||E_*(\rho))\]
\end{lemma}
\begin{proof}By assumption, $\Psi_*=(1-\eps)E_*+\eps\Psi_0$ for some unital positive map $\Psi_0\le 2E_*$. 
We denote $\sigma=E_*(\rho), \tilde{\sigma}=\Phi_*(\rho)$ and $\omega=\Psi_0(\rho)$. Then $\tilde{\sigma}=(1-\eps)\sigma+\eps \omega$.
Note that for any bounded state $\sigma\in S_B(\M)$, $X\mapsto \sqrt{\gamma_\sigma(X)}$ is a Hilbert space norm. Then by the triangle inequality,
\begin{align*}\sqrt{\gamma(\rho-\tilde{\sigma})}=&\sqrt{\gamma(\rho-(1-\eps)\sigma-\eps \omega)}\\
=&\sqrt{\gamma((\rho-\sigma)+\eps (\sigma-\omega))}
\\
\ge &\sqrt{\gamma(\rho-\sigma)}-\eps\sqrt{\gamma(\sigma-\omega)}
\end{align*}
where $\gamma$ can be $\gamma_\phi$ for any bounded state $\phi \in S_B(\M)$. Then
\begin{align*}\gamma(\rho-\tilde{\sigma})
\ge &\gamma(\rho-\sigma)-2\eps\sqrt{\gamma(\rho-\sigma)}\sqrt{\gamma(\sigma-\omega)}+\eps^2\gamma(\sigma-\omega)\\
\ge &\gamma(\rho-\sigma)-2\eps\sqrt{\gamma(\rho-\sigma)}\sqrt{\gamma(\sigma-\omega)}\\
\ge &\gamma(\rho-\sigma)-\eps\gamma(\rho-\sigma)-\eps\gamma(\sigma-\omega)\\
= &(1-\eps)\gamma(\rho-\sigma)-\eps\gamma(\sigma-\omega)\pl.
\end{align*}
Now take $\rho_t=t\rho+(1-t)\sigma$ and $\tilde{\rho}_t=t\rho+(1-t)\tilde{\sigma}$,
\begin{align*}
D(\rho||\tilde{\sigma})=&\int_0^1 (1-t)\gamma_{\tilde{\rho}_t}(\rho-\tilde{\sigma})dt \\
\ge & (1-\eps)\int_0^1 (1-t) \gamma_{\tilde{\rho}_t}(\rho-\sigma)dt-\eps\int_0^1(1-t) \gamma_{\tilde{\rho}_t}(\sigma-\omega)dt\pl.
\end{align*}
For the first term, because $\tilde{\rho}_t\le  (1+\eps)\rho_t$,
\begin{align*}
\int_0^1 (1-t) \gamma_{\tilde{\rho}_t}(\rho-\sigma)dt
\ge  (1+\eps)^{-1}\int_0^1 (1-t)  \gamma_{\rho_t}(\rho-\sigma)dtds=(1+\eps)^{-1}D(\rho||\sigma)\pl.
\end{align*}
For the second term, consider that $\tilde{\rho}_t\ge (1-\eps)(1-t)\sigma$
\begin{align*}
\int_0^1 (1-t)\gamma_{\tilde{\rho}_t}(\sigma-\omega)dt
\le & \frac{1}{(1-\eps)}\int_0^1  \gamma_{\sigma}(\omega-\sigma)dt\\ =&\frac{1}{(1-\eps)}\gamma_{\sigma}(\omega-\sigma)\\ \overset{(1)}{\le}&  \frac{1}{(1-\eps)k(2)}D(\omega||\sigma)
\\ =& \frac{1}{(1-\eps)k(2)}D(\Psi_*(\rho)||\sigma)\overset{(2)}{\le} \frac{1}{(1-\eps)k(2)}D(\rho||\sigma)\pl.
\end{align*}
Here, the inequality (1) above uses $\omega\le 2\sigma$ and \eqref{eq:keylemma}. The inequality (2) above follows from the monotonicity of relative entropy 
and the fact $\Psi_*(\sigma)=\sigma$.
Combined the estimated above, we obtained
\[ D(\rho||\tilde{\sigma})\ge \frac{1-\eps}{1+\eps} D(\rho||\sigma)-\eps k(2)^{-1}D(\rho||\sigma)=\Big(\frac{1-\eps}{1+\eps}-\frac{\eps}{(1-\eps)k(2)}\Big)D(\rho||\sigma) \pl,\]
where $k(2)=2\ln 2-1$. The above inequality is non-trivial for $\eps$ such that
\[\frac{1-\eps}{1+\eps}-\frac{\eps}{(1-\eps)k(2)}>0\pl. \]
Taking $\eps=0.1$, the above expression is approximately $0.53>\frac{1}{2}$.
\end{proof}
\begin{rem}{\rm
This lemma is related to \cite[Corollary 2.15]{laracuente2022quasi-factorization} and is a variant of \cite[Theorem 5.3]{gao2022complete}, which proves for GNS symmetric $\Phi$,
\begin{align}D(\rho||(\Phi_*)^2(\rho))\ge (1-\eps^2 k(2)^{-1}) D(\rho||E_*(\rho)).\label{eq:lemma}\end{align}
 The above lemma improves \eqref{eq:lemma} from two points: 1) does not need any symmetric assumption; 2) remove the square in $\Phi_*^2$. When $\Psi_*=\Phi_*^2$ being a square, \eqref{eq:lemma} could yield better bound that for $\eps=0.4$
\[ (1-(0.4)^2 k(2)^{-1})>\frac{1}{2}>\frac{1}{4}> \frac{1-0.4^2}{1+0.4^2}-\frac{0.4^2}{(1-0.4^2)k(2)}\pl.\]}
\end{rem}
Putting the above lemmas together, we obtain  Theorem \ref{thm:main5} in the following slightly more general form.
\begin{theorem}\label{thm:unital} Let $\Phi$ be a unital quantum channel and let $E:\M\to\N$ be the trace preserving conditional expectation onto the multiplicative domain $\mathcal{N}$ of $\Phi$. Define the CB return time
\begin{align}\label{eq:cbreturn} k_{cb}(\Phi):=\inf \{k\in \mathbb{N}^+\pl |\pl  0.9 E \le_{cp} (\Phi^*\Phi)^{k}\le_{cp} 1.1 E \} \ . \end{align}
 Then for any state $\rho\in S(\M)$,
 \begin{align} D(\Phi(\rho)||\Phi\circ E(\rho))\le \Big (1-\frac{1}{2k_{cb}(\Phi)} \Big ) D(\rho||E(\rho))\label{eq:SDPI}\ .\end{align}
 Furthermore,  for any finite von Neumann algebra $\mathcal{Q}$ and state $\rho\in S(\M \overline{\ten}\mathcal{Q})$
 \begin{align}\label{eq:CSDPI} D( \Phi\ten \id (\rho)||(\Phi\circ E)\ten \id (\rho))\le \Big (1-\frac{1}{2k_{cb}(\Phi)} \Big ) D(\rho|| E\ten \id (\rho))\pl.\end{align}
\end{theorem}
\begin{proof} It suffices to consider a bounded state $\rho \in S_B(\M)$. Note that by the conditional expectation property
\eqref{eq:condexp},
\begin{align*}&D(\rho||E(\rho))=\tau(\rho\log \rho-\rho\log E(\rho))=\tau(\rho\log \rho)-\tau(E(\rho)\log E(\rho))=H(\rho)-H(E(\rho))\pl,\\
&D(\Phi(\rho)||\Phi\circ E(\rho))=D(\Phi(\rho)||E_0\circ \Phi(\rho))=H(\Phi(\rho))-H(E_0 \circ \Phi(\rho)) =H(\Phi(\rho))-H(\Phi\circ E(\rho))\pl,
\end{align*}
where we used the property $\Phi\circ E=E_0 \circ \Phi$ from Proposition \ref{prop:cd}. Moreover, $H(E(\rho))=H(\Phi\circ E(\rho))$ as $\Phi$ is a trace preserving $*$-isomorphism on $\N$. Thus we have
\[D_\Phi(\rho)=H(\rho)-H(\Phi(\rho))=D(\rho||E(\rho))-D(\Phi(\rho)||\Phi \circ E(\rho))\pl.\]
 Iterating the entropy difference Lemma \ref{lemma:difference}, we have
\begin{align*}D(\rho\|(\Phi^*\Phi)^{k}(\rho))
\pl \le\pl  & D_\Phi(\rho)  + D(\rho\|(\Phi^*\Phi)^{k-1}(\rho))\\
  \le\pl & k D_\Phi(\rho)  + D(\rho\|\rho)\\
 = \pl& k (D(\rho||E(\rho))-D(\Phi(\rho)||\Phi \circ E(\rho))
 \end{align*}
 Now using Lemma \ref{lemma:approximate}, for $k=k_{cb}(\Phi)$,
 \[ D(\rho||E(\rho))\le 2 D(\rho\|(\Phi^*\Phi)^{k}\rho))\le 2k\big (D(\rho||E(\rho))-D(\Phi(\rho)||\Phi \circ E(\rho))\big)\pl. \]
Rearranging the terms gives the assertion. The general case $\rho\in S(\M)$ can be obtained via approximation $\rho_\eps=(1-\eps)\rho+\eps 1$ as \cite[Lemma A.2]{brannan2022complete}. The same argument applies to $\id_\cQ\ten \Phi$, because the CB return time $k_{cb}$ of $\id_\Q\ten \Phi$ is same as of $\Phi$ by the definition.
\qd


The above theorem is an improved data processing inequality for the relative entropy between a state $\rho$ and its conditional expectation $E(\rho)$.  Here $\N$ is the ``decoherence free'' subalgebra. Indeed,
for any two states $\sigma_1,\sigma_2\in\N$
\[D(\sigma_1||\sigma_2)\ge D(\Phi(\sigma_1)||\Phi(\sigma_2))\ge D(\Phi^*\Phi(\sigma_1)||\Phi^*\Phi(\sigma_2))=D(\sigma_1||\sigma_2)\]
does not decay.
Outside the ``decoherence free'' subalgebra $\N$, the relative entropy from a state $\rho$ to its projection $E(\rho)$ on $\N$ is strictly contractive under every use of the channel $\Phi$.

For $\Phi$ being a symmetric quantum Markov map, we have $\Phi=\Phi_*$ and this gives Theorem \ref{thm:main5}. Moreover, the Proposition \ref{prop:cd} reduces to
  \[\Phi\circ E=E\circ \Phi\pl , \pl  \Phi^2\circ E= E\circ \Phi^2=E\pl,\]
and one can iterate entropy contraction to obtain the discrete time entropy decay,
\begin{align*}
 D(\Phi^{2n}(\rho)|| E(\rho))\le (1-\frac{1}{2k_{cb}(\Phi)})^{2n} D(\rho||E(\rho)).
 \end{align*}

\section{Complete Modified Log-Sobolev Inequality for Symmetric Markov Semigroups}\label{sec:symmetricCMLSI}
\subsection{Functional inequalities}\label{sec:CMLSI}
In this section, we discuss a continuous time relative entropy decay for symmetric quantum Markov semigroups. We first review some basics of quantum Markov semigroups.
A quantum Markov semigroup $(T_t)_{t\ge 0}:\M\to\M$ is a family of maps satisfying
\begin{enumerate}
\item[i)] for each $t\ge 0$, $T_t$ is a quantum Markov map (i.e. normal, completely positive, and unital)
\item[ii)] $T_0=\id_\M$ and $T_s\circ T_t=T_{s+t}$ for any $s,t \ge 0$.
\item[iii)] for $x\in \M$, $t\mapsto T_t(x)$ is weak$^*$-continuous.
\end{enumerate}
The generator of the semigroup is defined as
\[\pl Lx=w^*\text{-}\lim_{t\to 0} \frac{1}{t}(x-T_t(x))\pl \]
on the domain of $L$ that the limit exists. 
In this section we still consider $\M$ as a finite von Neumann algebra equipped with a normal faithful tracial state $\tau$. Given $(T_t)$  is \emph{symmetric} (or more specifically, \emph{trace-symmetric}) i.e.
\[\tau(x^*T_t(y))=\tau(T_t(x)^*y)\pl , \pl \forall  x,y\in \M, t\ge 0,\]
the generator $L$ is a positive, symmetric operator, densely defined on $L_2(\M)$. Its kernel is
the fixed-point subspace $\N:=\ker(L)=\{x\in \M\pl | \pl T_t(x)=x, \forall t\ge 0\}$, which coincides with the common multiplicative domain of all $T_t$, hence a von Neumann subalgebra. Moreover, each $T_t$ is an $\N$-bimodule map
\[T_t(axb)=aT_t(x)b\pl, \pl  \forall\pl a,b\in \N ,x\in \M\]
In particular, we have
\[T_t\circ E= E\circ T_t=E\pl.\]
where $E:\M\to \N$ is the trace preserving conditional expectation onto the fixpoint algebra $\N$. We say $(T_t)$ is \emph{ergodic} if $\N=\bC 1$ is trivial.
Note that in the mathematical physics literature it is common to use \emph{primitive} instead of ergodic.
In this case,  the semigroup admits a unique invariant state, namely the trace $\tau$. 
We will consider symmetric quantum Markov semigroups that are not necessarily  ergodic.

Recall that a semigroup is equivalently determined by its
{\it  Dirichlet form} \[\E:L_2(\M)\to [0,\infty]\pl ,\pl  \E(x,x)=\tau(x^*L x)\pl.\]
We write $\dom (L)$ for the domain of $L$ and $\dom (\E)$ for the domain of $\E$. The Dirichlet subalgebra is defined as $\A_\E:=\dom (\E)\cap \M$.
It  was proved \cite{davies1992non} that  $\A_\E$ is a dense $*$-subalgebra of $\M$ and a core of $L^{1/2}$.
We denote by
\[S(\A_\E)=S(\M)\cap \A_\E \pl, \pl S_B(\A_\E)=S_B(\M)\cap \A_\E\]
the set of bounded density operators from $\A_\E$. We now introduce the formal definitions of functional inequalities for quantum Markov semigroups.
 \begin{defi}\label{defi:FI}
 Let $T_t=e^{-Lt}:\M\to\M$ be a symmetric quantum Markov semigroup and $E:\M\to \N$ be the trace preserving conditional expectation onto its fixed point space. We say $T_t$ satisfies
 \begin{enumerate}
 \item[i)] the Poincar\'e inequality (PI) for $\la>0$ if
 \begin{align}\label{eq:PI}
 \la\norm{x-E(x)}{2}^2\le \E(x,x)\pl, \pl  \forall x\in \A_\E\pl.
 \end{align}
 \item[ii)] the log-Sobolev inequality (LSI) for $\al>0$ if
 \begin{align}\label{eq:LSI}
 \al \tau\big(|x|^2\ln |x^2|-E(|x^2|)\ln E(|x^2|)\big)\le 2\E(x,x)\pl, \pl  \forall x\in \A_\E\pl.
 \end{align}
 \item[iii)]the modified log-Sobolev inequality (MLSI) for $\al>0$ if
 \begin{align}\label{eq:MLSI1}
 2\al D(\rho||E(\rho))\le \E(\rho,\ln\rho)\pl, \pl  \forall \rho\in S_B(\A_\E)\pl,
 \end{align}
\item[iv)]the complete modified log-Sobolev inequality (CMLSI) for $\al>0$ if  $\id_\cQ\ten T_t$ satisfies $\al$-MLSI inequality for any finite von Neumann algebra $\cQ$.
 \end{enumerate}
 The  optimal (largest possible) constants for PI, LSI, MLSI, and CMLSI will be denoted respectively as $\la(L), \al_2(L),\al_{1}(L),$ and $\al_{c}(L)$ (or $\la, \al_2,\al_1$ and $\al_{c}$ in short if the generator is clear).
 \end{defi}

 The Poincar\'e inequality \eqref{eq:PI} is equivalent to the spectral gap of $L$ as a positive operator. LSI \eqref{eq:LSI} is equivalent to hypercontractivity \cite{olkiewicz1999hypercontractivity}
 \begin{align} \norm{T_t:L_2(\M)\to L_p(\M)}{}\le 1\pl  \text{ if } \pl p\le 1+e^{2\al t}  \label{eq:HC}.\end{align}
MLSI \eqref{eq:MLSI1} is known to be equivalent to the exponential decay of relative entropy (\cite[Theorem 3.2]{bardet2017estimating} and \cite[Proposition A.3]{brannan2022complete}) that
\begin{align} D(T_t(\rho)||E(\rho))\le e^{-2\al t}D(\rho||E(\rho))\pl, \pl \forall \rho\in S(\M)\pl.\label{eq:entropy}\end{align}
The equivalence of \eqref{eq:MLSI1} and \eqref{eq:entropy} is obtained by differentiating the relative entropy for $T_t$ at $0$, which leads to the  entropy production on the R.H.S of MLSI
\[I_L(\rho):=\E(\rho,\ln\rho)=-\frac{d}{dt}\vert_{t=0} D(T_t(\rho)||E(\rho))=\tau(L(\rho) \ln \rho) \pl,\]
It is well-known that
\[ \al_{2}\leq\al_{1}\leq \la \pl. \]
The main motivation to consider CMLSI over MLSI and LSI is the tensorization property
\begin{align}\label{eq:tensorization} \al_{c}(L_1\ten \id+\id\ten L_2)=\min\{\al_{c}(L_1), \pl \al_{c}(L_2)\}\pl,\end{align}
which  in the quantum cases fails for $\al_1$ \cite[Section 4.4]{brannan2022complete}, and only known to hold for $\al_2$ for limited examples in small dimensions.
The main result of this section is Theorem \ref{thm:main1}, which asserts a lower bound
\[ \al_c(L)\ge \frac{1}{2t_{cb}(L)}\]
by the inverse of CB return time
\begin{align} t_{cb}(L)=\inf\{t>0 \pl |\pl (1-0.1) E\le_{cp} T_t\le_{cp} (1+0.1)E \}. \label{eq:cbreturn3}\end{align}
Here we set $\eps=0.1$ for the notation $t_{cb}(\eps)$ in Theorem \ref{thm:main1} because of Lemma \ref{lemma:approximate}.


\begin{proof}[Proof of Theorem \ref{thm:main1}]
Let $t_m=t_{cb}(L)/2m$ for some $m\in \mathbb{N}_+$. Since $T_t$ is symmetric,  $T_{t_m}^*T_{t_m}=T_{t_m}T_{t_m}=T_{2t_m}$. Hence $T_{t_m}$ has discrete return time $k_{cb}(T_{t_m})\le m$. By the Lemma \ref{lemma:difference}, for any $\rho\in S_B(\M)$,
\begin{align*}D(T_{t_m}(\rho)||E(\rho))\le & (1-\frac{1}{2m}) D(\rho||E(\rho))
\end{align*}
Write $t_{cb}=t_{cb}(L)$. Now assume further $\rho\in \cup_{t>0} T_t(\M)\subset \dom (L)$.
We have by Theorem \eqref{thm:unital},
\begin{align*}I(\rho)&=\lim_{t\to 0}\frac{ D(\rho||E(\rho))-D(T_{t}(\rho)||E(\rho))}{t}\\
&=\lim_{m\to \infty}\frac{ D(\rho||E(\rho))-D(T_{\frac{t_{cb}}{2m}}(\rho)||E(\rho))}{\frac{t_{cb}}{2m}}\\
&\ge \lim_{m\to \infty}\frac{\frac{1}{2m}D(\rho||E(\rho))}{\frac{t_{cb}}{2m}}
=\frac{1}{t_{cb}} D(\rho||E(\rho))\pl .
\end{align*}
The entropy decay for general $\rho\in S(\M)$ can be obtained by approximation as in the Appendix \cite[Appendix]{brannan2022complete}). This proves $\al_{1}(L)\ge \frac{1}{2t_{cb}(L)}$. The same argument applies to $\id_\cQ\ten L$ yields the assertion $\al_{c}(L)\ge \frac{1}{2t_{cb}(L)}$.
\qd


\begin{rem}{\rm
a) For LSI constant $\al_2$, the $\Omega(\frac{1}{t_{cb}})$ lower bounds were obtained for  ergodic semigroups  in both classical \cite{diaconis1996logarithmic} and quantum setting \cite{temme2014hypercontractivity}. These bounds as well as our bound for $\al_c$ are asymptotic tight (See Example \ref{ex:cyclic2} and Section \ref{sec:birth}).

b)  In \cite{brannan2022complete}, a similar estimate $\al_c\ge \Omega(\frac{1}{t_{cb}})$ was obtained for semigroups that admits non-negative entropic Ricci curvature lower bound.
The entropy Ricci curvature lower bound for quantum Markov semigroup was introduced by Carlen and Mass \cite{carlen2017gradient} using $\lambda$-displacement convexity of entropy functionals $H$ w.r.t to certain noncommutative Wasserstein distance, inspired from Lott and Villani \cite{lott2009ricci}, and  Sturm's \cite{sturm2006geometry} work on metric measure spaces. For heat semigroups on Riemmannian manifold, the entropy Ricci curvature lower bound follows from a lower bound of the Ricci curvature tensor. Nevertheless, in the noncommutative case,  these entropy Ricci curvature lower bounds for quantum Markov semigroup are in general hard to verify. So far most examples rely on certain interwining  relation $\nabla T_t= e^{-\lambda t}\tilde{T}_t \nabla$ between the semigroup $T_t$ and a gradient operator $\nabla$ (see
\cite{carlen2017gradient,brannan2021complete,wirth2021complete}).

c) Our Theorem \ref{thm:main1} here does not rely on any curvature conditions, which uses only information theoretic tools such as entropic quantities and inequalities. To the best of our knowledge, this direct proof is even novel in  the classical setting. It is worth point out that the definition of relative entropy as well as its exponential decay of relative entropy is independent of the choice of the trace, which also shows the naturalness of our approach and the extension to non-tracial von Neumann algebras in Section \ref{sec:GNS}.
}\end{rem}

\subsection{CB return time}
We now consider a common scenario where the CB return time $t_{cb}$ is finite. The original motivation for the notion, despite defining using CP (complete possitive) order \eqref{eq:cbreturn3}, is the following characterization using CB (completely bounded) norm. Recall that  a linear map $\Psi:\M\to\M$ is called a $\N$-bimodule map if
\[\Psi(axb)=a\Psi(x)b\pl, \pl \forall\pl  a,b\in \N, x\in \M\]
\begin{prop}\label{prop:equivalence}
Let $\N\subset\M$ be a subalgebra and $E:\M\to\N$ be the trace preserving conditional expectation. Let $\Psi:\M\to\M$ be a $\N$-bimodule $*$-preserving map. For any $\eps>0$, the following two conditions are equivalent:
\begin{enumerate}\item[i)]$(1-\eps) E\le_{cp} \Psi \le_{cp} (1+\eps) E\pl $;
\item[ii)] $\norm{\Psi-E:L_\infty^1(\N\subset\M)\to L_\infty(\M)}{cb}\le \eps$.
\end{enumerate}
\end{prop}
The condition ii) above is the completely bounded norm from the space $L_\infty^1(\N\subset\M)$  to $\M$. $L_\infty^1(\N\subset\M)$ is called a conditional $L_\infty$ space, defined as the completion of $\M$ with respect to the norm
\[ \norm{x}{L_\infty^1(\N\subset\M)}=\sup_{a,b\in\N \pl ,\pl \|a\|_2=\|b\|_2=1}\norm{axb}{1}\pl,  \]
where the supremum takes over all $a,b\in L_2(\N)$ with $\|a\|_2=\|b\|_2=1$.
The operator space structure of $L_\infty^1(\N\subset\M)$ is given by the identification
\[\bM_n(L_\infty^1(\N\subset\M))=L_\infty^1(\bM_n(\N)\subset \bM_n(\M))\pl.\]
(see \cite{junge2010mixed} and \cite[Appendix]{gao2020relative}).
Proposition \ref{prop:equivalence} is relatively self-evident in the ergodic case $\N=\mathbb{C}1$, $L_\infty^1(\N\subset\M)\cong L_1(\M)$, which we illustrate below.

\begin{exam}[Classical case]{\rm \label{exam:CMS} Let $(\Omega,\mu)$ be a probability space. Consider $P:L_\infty(\Omega)\to L_\infty(\Omega)$ be a linear map with kernel $P(f)(x)=\int_{\Omega}k(x,y)f(y)d\mu(y)$. It is clear that $P$ is $*$-preserving i.e. $P(\bar{f})=\overline{P(f)}$ if $k$ is real; $P$ is a positive map if and only if the kernel function $k\ge 0$. Recall the expectation map $$E_\mu:L_\infty(\Omega)\to \mathbb{C}{\bf 1}\pl , \pl  E(f)=(\int_{\Omega}f\mu){\bf 1},$$ where ${\bf 1}$ is the unit constant function. The kernel of $E_\mu$ is the constant function ${\bf 1}$ on the product space $\Omega\times \Omega$. The following equivalence is self-evident
\begin{align} (1-\eps)E\le P \le (1+\eps)E \nonumber
\Longleftrightarrow\pl &\eps E\le P-E \le \eps E \nonumber
\\ \Longleftrightarrow\pl &\eps {\bf 1}\le k-{\bf 1} \le \eps {\bf 1}\nonumber
\\ \Longleftrightarrow\pl &\norm{k-{\bf 1}}{L_\infty(\Omega\otimes \Omega)}\le \eps\nonumber
\\ \Longleftrightarrow\pl &\norm{P-E:L_1(\Omega)\to L_\infty(\Omega)}{}\le \eps\label{eq:equivalence}
\end{align}
To see the equivalence in terms of complete positivity and completely bounded norm in Proposition \ref{prop:equivalence}, it suffices to notice that every positive (resp. bounded) map  to $L_\infty(\Omega)$ is automatically completely positive (resp. completely bounded with same norm \cite{smith1983completely}).
}
\end{exam}

\begin{exam}[Quantum case]{\rm
The above argument also applies to the noncommutative ergodic case $\N=\mathbb{C}1\subset\M$. The correspondence between the map $P$ and its kernel $k$ generalizes to the isomorphism between the map $T$ and its Choi operator $C_T\in \M^{op}\overline{\ten}\M$
\[T(x)=\tau \ten \id (C_T(x\ten 1))\pl,\pl x\in L_1(\M)\cong (\M^{op})_*,\]
where $\M^{op}$ is the opposite algebra of $\M$. The isomorphism $T\mapsto C_T$ is not only positivity preserving by Choi's Theorem ( $T$ is CP iff $C_T\ge 0$), but also isometric by Effros-Ruan Theorem (see \cite{effros2000operator,blecher1991tensor}),
\[\norm{T:L_1(\M)\to L_\infty(\M)}{cb}=\norm{C_T}{\M^{op}\overline{\ten}\M}\pl.\]
Then the equivalence in Proposition \ref{prop:equivalence} follows as \eqref{eq:equivalence}.
}
\end{exam}

For the general case of a $\N$-bimodule map $T$ with a nontrivial $\N$, the above isomorphism holds with more involved module Choi operator, which
we refer to the discussion in Section \ref{sec:GNSentropydecay} and also \cite[Lemma 5.1]{bardet2021entropy} and \cite[Lemma 3.15]{gao2020fisher} for the complete proof of Proposition \ref{prop:equivalence}.

With Proposition \ref{prop:equivalence}, the CB-return time can be equivalently defined as
\begin{align}\label{eq:cbreturn2} t_{cb}(L):=\inf \{\pl t>0\pl |\pl  \norm{T_{t}-E:L_\infty^1(\N\subset\M)\to L_\infty(\M)}{cb}\le 0.1  \} \end{align}
It is known that $t_{cb}$ is finite whenever $T_t$ satisfies the Poincar\'e inequality and one-point ultra-contractive estimate.
\begin{prop}\label{prop:finitetcb}Let $T_t:\M\to \M$ be a symmetric quantum Markov semigroup and $E:\M\to\N$ be the trace preserving conditional expectation onto the fixed point space. Suppose \begin{enumerate} \item[i)] $T_t$ satisfies Poincar\'e inequality of $\lambda>0$ that $\norm{T_{t}-E:  L_2(\M)\to L_2(\M)}{}\le e^{-\lambda t}\pl ,\pl \forall t\ge 0$;
\item[ii)] There exists $t_0\ge 0$ such that $\norm{T_{t_0}:  L_\infty^1(\N\subset \M)\to L_\infty(\M)}{cb}\le C_0$.
  \end{enumerate}
  Then
  $ t_{cb}\kl \frac{1}{\lambda}\ln (10 C_0)+t_0$.
\end{prop}
\begin{proof}This is now a standard argument similar to \cite[Proposition 3.8]{brannan2022complete} and \cite[Lemma B.1]{gao2022complete}.
\end{proof}
\begin{rem}{\rm For the special case of $T_0=\id:\M\to \M$, Proposition \ref{prop:equivalence} gives
\[ \norm{\id:  L_\infty^1(\N\subset \M)\to L_\infty(\M)}{cb}=\inf \{\pl \mu>0\pl |\pl \id \le_{cp} \mu E\}:=C_{cb}(E).\]
$C_{cb}(E)$ was introduced in \cite{gao2020relative} as the complete bounded version of Popa and Pimsner's subalgebra index \cite{pimsner1986entropy},
 \[  C(E):=\inf\{\mu>0 \pl |\pl  \rho\le \mu E\rho\pl, \pl \forall \rho\in \M_{+}\}\pl, \pl C_{cb}(E):=\sup_{n}C(E\ten \id_{\mathbb{M}_n} )\pl.\]
 When $\M$ is finite dimensional, both the index $C(E)$ and $C_{cb}(E)$ are finite and admits explicit formula \cite[Theorem 6.1]{pimsner1986entropy}. In this case, one can take $t_0=0$ in above Proposition \ref{prop:finitetcb} and yields,
 \[ t_{cb}\kl \frac{\ln (10 C_{cb}(E))}{\lambda}\pl.\]
 }
\end{rem}

\subsection{Classical Markov semigroups} In the remainder of this section, we focus on applications towards classical Markov map. We postpone the discussion of truly noncommutative semigroups to Section \ref{sec:GNSsymmetric}. 
Let $T_t=e^{-Lt}:L_\infty(\Omega,\mu)\to L_\infty(\Omega,\mu)$ be an ergodic Markov semigroup symmetric to the probability measure $\mu$. Note that in the ergodic case $L_\infty^1(\mathbb{C}1\ssubset L_\infty(\Omega) )=L_1(\Omega,\mu)$,  and by Smith's lemma \cite{smith1983completely}, any bounded map $T:L_1(\Omega,\mu)\to L_\infty(\Omega,\mu)$ is automatic completely bounded
\[\norm{T: L_1(\Omega,\mu)\to L_\infty(\Omega,\mu)}{}=\norm{T: L_1(\Omega,\mu)\to L_\infty(\Omega,\mu)}{cb}\pl. \]
Then the CB return time $t_{cb}$ reduces to the standard $L_\infty$-mixing time
\begin{align*} t_{b}=&\inf \{ \pl t>0\pl  |\pl  \norm{T_t-E:L_1(\Omega,\mu)\to L_\infty(\Omega,\mu)}{}\le \frac{1}{10} \}.
\end{align*}
Then by a combination of Theorem \ref{thm:main1} and Proposition \ref{prop:finitetcb}, we obtain Theorem \ref{thm:main9},
\begin{align}\label{eq:CMSbound1}\al_1\ge \al_{c}\ge  \frac{\lambda}{2(\lambda t_0+\ln C_0+\ln 10)}\end{align}
where $\lambda$ is the spectral gap of $L$ and $\norm{T_{t_0}: L_1(\Omega)\to L_\infty(\Omega)}{}\le C_0$.
This result can be compared to the bound of Diaconis and Saloff-Coste \cite[Theoem 3.10]{diaconis1996logarithmic}, which states\footnotemark 
\begin{align} \label{eq:diaconis}\al_1\ge \al_{2}\ge \frac{\lambda}{\lambda t_0+\ln (C_0)+1}\pl.\end{align}
\footnotetext{Note that the LSI constant in \cite{diaconis1996logarithmic} is defined as half of our $\al_2$ here. }

In particular, $\al_1\ge \al_{2}\ge \frac{2}{t_b(e^{-2})}$ for the
the alternative $L_\infty$-mixing time
\begin{align*} t_b(e^{-2})=\inf \{ t>0\pl  |\pl  \norm{T_t-E_\mu:L_1(\Omega,\mu)\to L_\infty(\Omega,\mu)}{}\le \frac{1}{e^2}\}
\end{align*}
By the comparison $e^{-3}<0.1<e^{-2}$, we have $t_b(e^{-2})\le t_{cb}(0.1)\le \frac{3}{2}t_b(e^{-2})$. Hence, in terms of lower bound for $\al_1$, \eqref{eq:diaconis} and \eqref{eq:CMSbound1} are equivalent up to absolute constants. The difference is that \eqref{eq:diaconis} lower bounds the LSI constant $\al_2$ and our estimate \eqref{eq:CMSbound1} bounds the CMLSI constant $\al_{c}$.

For finite Markov chains with $|\Omega|<\infty$,  we have finite index
\[C_{cb}(E_\mu)=C(E_\mu)=\inf\{C>0\pl |\pl  f\le C\mu(f) \pl \forall f\ge 0\}= \norm{\mu^{-1}}{\infty}\pl,\]
where  $\mu$ is a strictly positive probability density function.
It was proved in \cite{diaconis1996logarithmic} that
\begin{align}&\frac{1}{t_b(e^{-2})}\le \lambda \le \frac{2+\log\norm{\mu^{-1}}{\infty}}{2t_b(e^{-2})}\pl ,\\
&\frac{1}{t_b(e^{-2})}\le \al_2 \le \frac{4+\log\log \norm{\mu^{-1}}{\infty}}{2t_b(e^{-2})}\pl. \label{eq:re}
\end{align}
Combined with our Theorem \ref{thm:main1}, we obtain
\begin{cor}\label{eq:loglog}
For a finite Markov chain $T_t:l_\infty(\Omega,\mu)\to l_\infty(\Omega,\mu)$ symmetric to $\mu$,
\[  \min \Big \{ \frac{4\al_2}{3(4+\log\log \norm{\mu^{-1}}{\infty})}, \frac{\lambda}{2\log (10\norm{\mu^{-1}}{\infty})} \Big \} \le \al_c\le \al_1\le \lambda\]
\end{cor}
\begin{proof}  Note that $t_b(e^{-2})\le t_{cb}\le \frac{3}{2}t_b(e^{-2})$. Then by Theorem \ref{thm:main1} and \eqref{eq:re}
\[ \al_c\ge \frac{1}{2t_{cb}(0.1)}\ge \frac{1}{3t_{b}(e^{-2})} \ge \frac{2\al_2}{3(4+\log\log \norm{\mu^{-1}}{\infty})}\pl.\]
The other lower bound $\al_c\ge \frac{\lambda}{2\log (10\norm{\mu^{-1}}{\infty})}$ follows from Corollary \ref{thm:main9} by choosing $t_0=0$.
\end{proof}



\begin{exam}{\rm
When $\Omega$ is not finite, here  is a simple example with $\al_c(L)>0$ but the ultra-contractivity \eqref{eq:ultracon} is never satisfied for finite $t_0$. Take $L=I-E_\mu$. It generates the so-called depolarizing semigroup
\[ T_t=e^{-t}\id+(1-e^{-t})E_\mu\pl, \pl T_t(f)= e^{-t}f+(1-e^{-t})\mu(f){\bf 1}\pl,\]
where ${\bf 1}$ is the unit constant function.
Then for any $t<\infty$, \[\norm{T_{t}-E_\mu: L_1(\Omega,\mu)\to L_\infty(\Omega,\mu)}{}=\norm{e^{-t}\id: L_1(\Omega,\mu)\to L_\infty(\Omega,\mu)}{}=e^{-t}C(E_\mu)\pl.\]
which is infinite whenever $L_{\infty}(\Omega,\mu)$ is infinite dimensional. One the other hand, it follows from direct calculation that $\al_{c}(I-E_\mu)\geq \frac{1}{2}$. }
\end{exam}

\subsection{H\"ormander system}\label{sec:hormander}
We now discuss the application to Markov semigroups on smooth manifolds generated by sub-Laplacians.
Let $(M,g)$ be a $d$-dimensional compact connected Riemannian manifold without boundary and let $d\mu=\omega d\operatorname{vol}$ be a probability measure with smooth density $\omega$ w.r.t to the volume form $d\operatorname{vol}$. A family of vectors fields  $H=\{X_i\}_{i=1}^k\subset TM$ with $k\leqslant d$ is called a \emph{H\"{o}rmander system} if at every point $x\in M$, the tangent space at $x$ can be spanned by the iterated Lie brackets of $X_i$'s
\[
T_xM=\operatorname{span}\{[X_{i_1},[X_{i_2},\cdots, [X_{i_{n-1}}, X_{i_n}]]] \pl  | \pl 1\leqslant i_1,i_2\cdots i_n\leqslant k  \}. \tag{\text{H\"{o}rmander condition}}
\]
By compactness we can assume there is a global constant $l_H$ such that for every point $x\in M$, we need at most $l_H$th iterated Lie bracket in above expression (also called strong H\"{o}rmander condition). Denote $\nabla=(X_1, \cdots, X_k)$ and by $X_i^*$ the adjoint of $X_i$ on $L^{2}(M,d\mu)$. Under the  H\"{o}rmander condition, the sub-Laplacian
\[
\Delta_H=\nabla^*\nabla= \sum_{i} X_i^*X_i=\sum_{i}X_i^2+\div_{\mu}(X_i)X_i
\]
is a symmetric operator on $L^{2}(M,\mu)$ which generates an ergodic Markov semigroup $P_t=e^{-\Delta_H t}$, often called the horizontal heat semigroup. Here $\operatorname{div}_{\mu}(X)$ is the divergence of $X$ w.r.t to $\mu$. When $H=\{X_i\}_{i=1}^d$ forms an orthonormal frame to the Riemannian metric, $\Delta_H=\Delta$ recovers the (weighted) Laplace-Beltrami operator and $P_t=e^{-\Delta t}$ is the (weighted) heat semigroup on $M$.

The \emph{gradient form} (Carr\'e du Champ operator) of $\Delta_H$ is given by
\[
\Gamma(f,g):=\frac{1}{2}(f\Delta_H(g)+\Delta_H(f)g-\Delta_H(fg))=
\sum_{i}\lan X_i f,X_ig\ran\pl. \]
It follows from the product rule of derivatives that $\Gamma$ is diffusive, i.e. $\Gamma(fg,h)=f\Gamma(g,h)+g\Gamma(f,h)$. For diffusion semigroups , it is known \cite[Theorem 5.2.1]{bakry2014analysis} that the MLSI constant $\al_1$ and the LSI constant $\al_2$ coincides, i.e. $\al:=\al_1=\al_2$. The positivity \[ \al(\Delta_H)>0 \pl. \]
for any H\"{o}rmander's system $H=\{X_i\}_{i=1}^k$ on a compact connected Riemannian manifold without boundary was proved in \cite[Theorem 3.1]{lugiewicz2007coercive}.

Our Corollary \ref{cor:main3} improve this to $\al_c(\Delta_H)>0$.
\begin{proof}[Proof of Corollary \ref{cor:main3}]
Recall the following Sobolev-type inequality (see e.g. \cite[Lemma 2.1]{lugiewicz2007coercive})
\begin{align}\label{eq:sobolev}
\|f\|_{q} \kl C \big( \lan \Delta_H f,f\ran + \norm{f}{2}^2\big)^{1/2} \pl,
\end{align}
where $q=\frac{2dl_H}{dl_H-2}>2$ and $l_H$ is globoal Lie bracket length needed in the strong H\"ormander condition. By Varopoulos' Theorem (see \cite[Chapter 2]{varopoulos1991analysis}) on the dimension of semigroups, this implies the following ultra-contractive estimate
\begin{align}\label{eq:ultra}
\norm{e^{-\Delta_Ht}:L_{1}(M,\mu)\to L_{\infty}(M,\mu)}{}\leqslant C^{\prime} t^{-m/2} \text{ for } 0 \le t \le 1 \text{ and some }C'>0\pl,
\end{align}
where $m=dl_{H}$. Also, it was proved in \cite[Theorem 2.3]{lugiewicz2007coercive} that $\Delta_H$ satisfies Poincar\'e inequality $\la(\Delta_H)>0$. Combining these with our Theorem \ref{thm:main9} yields the assertion.
\end{proof}

The Sobolev-type inequality \eqref{eq:sobolev} is also used in \cite{lugiewicz2007coercive} by Lugiewicz and Zegarl\'inski to prove that $\al_2(\Delta_H)>0$. Their proof relies on the Rothaus lemma, so does the discrete case by Diaconis and Saloff-Coste  \cite{diaconis1996logarithmic}. However, we will see in Section \ref{sec:failureLSI} that this approach is out of scope for showing the CMLSI constant $\al_c(\Delta_H)>0$.

\begin{exam}{\rm \label{exam:su2}The special unitary group $\operatorname{SU}\left( 2 \right)$ is
\[
\operatorname{SU}\left( 2 \right)=\{ cI+xX+yY+zZ:  c^2+x^2+y^2+z^2=1, x,y,x, c \in \mathbb{R} \}.
\]
where $X,Y,Z$ are the skew-Hermitian Pauli unitary
\[
X=\left[\begin{array}{cc} 0& 1\\ -1& 0
\end{array}\right],Y=\left[\begin{array}{cc} 0& i\\ i& 0
\end{array}\right], Z=\left[\begin{array}{cc} i& 0\\ 0& -i
\end{array}\right]\pl.
\]
 The Lie algebra is $\mathfrak{su}(2)=\operatorname{span}_{\mathbb{R}}\{X,Y,Z\}$ with Lie bracket rules as
\begin{align}[X,Y]=2Z\ , [Y,Z]=2X\ , [Z,X]=2Y\pl. \label{eq:lie}\end{align}
The canonical sub-Riemannian structure is given by $H=\{X,Y\}$,  which is a generating set of $\mathfrak{g}$ because $[X,Y]=2Z$.  The associated sub-Laplacian is
\begin{align}\label{eq:su2}
\Delta_H=-(X^2+Y^2).
\end{align}
The semigroup $P_t=e^{-\Delta_H t}$ on $\operatorname{SU}\left( 2 \right)$ has been studied as a prototype of horizontal heat semigroup. In particular, Baudoin and Bonnefont in \cite{baudoin2009subelliptic} proved that
\begin{align}\label{eq:baudoin}
\Gamma(P_t f,P_t f)\leqslant C e^{-4 t}P_t(\Gamma(f, f)),
\end{align}
for some constant $C>0$. In \cite{gao2022complete}, Gao and Gordina based on \eqref{eq:baudoin}  proved the CMLSI constant that
\[ \al_c(\Delta_H)\ge (2\int_0^\infty C e^{-4 t}dt )^{-1}=\frac{2}{C}\pl.\]
 }
\end{exam}

The gradient estimate \eqref{eq:baudoin}, as a weaker variant of Bakry-Emery curvature dimension condition, have been found useful to derive CMLSI in \cite{gao2022complete}. Nevertheless, this weaker gradient estimate is only known for only a limited number of examples in the sub-Riemannian setting \cite{driver2005hypoelliptic,melcher2008hypoelliptic}. Our result avoid this condition and obtains CMLSI for general H\"ormander systems.

\begin{exam}{\rm Let $n\ge 3$. The special unitary group $\operatorname{SU}\left( n \right)$ is
\[
\operatorname{SU}\left( n \right)=\{ u\in \bM_n\pl|\pl u^*u=1\pl, \pl \det(u)=1\pl.\}.
\]
The Lie algebra $\mathfrak{su}(n)$ is the space of all the skew-Hermitian matrices, and
a natural basis $\mathfrak{su}(n)$ is given by $\{X_{j,k},Y_{j,k},Z_{j,k} \pl | \pl  1\le j< k\le n\}$ where
\[
X=e_{jk}-e_{kj}\pl ,\pl Y=i(e_{jk}+e_{kj})\pl ,\pl  Z_k=i(e_{11}-e_{kk})
\]
which is $n^2-1$ dimensional. Let $V=\{1,\cdots, n\}$ be a vertex set and $E\subset V\times V$ as an edge set. The set \[H_E=\{X_{j,k},Y_{j,k} \pl |\pl (j,k)\in E\}\] is a generating set if and only if $(V,E)$ is a connected graph. The associated sub-Laplacian
\[\Delta_E=-\sum_{(j,k)\in E} X_{j,k}^2+Y_{j,k}^2\pl, \]
is a generalization of \eqref{eq:su2}.
Corollary \eqref{cor:main3} implies that $\al_c(\Delta_E)>0$ for all connected $(V,E)$, despite the gradient estimate \eqref{eq:baudoin} is not known for this type of generators.
}
\end{exam}

\subsection{Transference semigroups}
Let us discuss an immediate application of $\al_c(\Delta_H)>0$ to symmetric Quantum Markov semigroups. Let $G$ be a compact Lie group and $H=\{X_1,\cdots,X_k\}$ be a generating set of its Lie algebra $\mathfrak{g}$. Then $\{X_1,\cdots,X_k\}$ satisfies H\"ormander condition and its sub-Laplacian
$\Delta_H=-\sum_{k}X_i^2$
generates a Markov semigroup $P_t=e^{-\Delta_H t}$ symmetric to the Haar measure. Let $u:G\to \bM_n$ be a finite dimensional unitary representation and $d_u:\mathfrak{g}\to i(\bM_n)_{s.a.}$ be the corresponding Lie algebra morphism. $P_t=e^{-\Delta_H t}$ induces a quantum Markov semigroup $T_t=e^{-L_Ht}:\bM_m\to\bM_m$ with generator in the Lindbladian form \cite{lindblad1976generators},
\[L_H(\rho)=-\sum_{i=1}^k [d_u(X_i),[d_u(X_i),\rho]]\pl.\]
$T_t$ is called a transference semigroup of $P_t$ by the following commuting diagram
\begin{equation}\label{dia:transference}
\begin{tikzcd}
L_{\infty}(G, \bM_m)  \arrow{r}{P_t\ten \id_{\bM_m}}  & L_{\infty}(G, \bM_m)
 \\
\bM_m  \arrow{u}{\pi_u} \arrow{r}{T_{t}} & \bM_m  \arrow{u}{\pi_u},
\end{tikzcd}
\end{equation}
where the transference map $\pi_u$ is a $*$-endomorphism given by
\begin{align*}
\pi_u: \bM_m\to L_{\infty}(G, \bM_m)\pl,\pl \pi_u(\rho)(g)=u(g)^*(\rho)u(g)\pl.
\end{align*}
which embeds $\bM_m$ into $L^{\infty}(G, \bM_m)$. Then the quantum semigroup $T_t$ is the restriction of the matrix-valued extension of classical semigroup $P_t\ten \id_{\bM_m}$ on the image of $\pi(\bM_m)$. Such a transference relation holds for any (projective) unitary representation, which yields the following dimension-free estimate for both spectral gap and CMLSI constant (see \cite[Section 4]{gao2020fisher}),
\[\al_{c}(\Delta_H)\le \al_{c}(L_H)\pl , \la(\Delta_H)\le \la(L_H)\pl.\]
Then Corollary \ref{cor:main3} implies the following estimate of $\al_{c}(L_H)$ independent of the unitary representation $u$.
\begin{cor}
Let $G$ be a compact Lie group and $H=\{X_1,\cdots,X_k\}$ be a generating set of its Lie algebra $\mathfrak{g}$. There exists a constant $\al_{c}(\Delta_H)>0$ such that for all unitary representation $u$, the induced quantum Markov semigroup generated by
\[L_H(\rho)=-\sum_{i=1}^k [d_u(X_i),[d_u(X_i),\rho]]\]
satisfies
$\al_{c}(L_H)\ge \al_{c}(\Delta_H)>0$.
\end{cor}

\subsection{Failure of matrix valued log-Sobolev inequality}\label{sec:failureLSI}
As mentioned above, a standard analysis approach to MLSI through Hypercontractivity or LSI relies the Rothaus Lemma (see e.g. \cite{rothaus1985analytic,bakry1994hypercontractivite})
\[H(|f|^2)\le H(|f-E_\mu(f)|^2)+\norm{f-E_\mu(f)}{2}^2\pl.\]
Here we show that Rothaus Lemma, LSI and Hypercontractivity all fail for matrix valued functions for any classical Markov semigroups. This is a strong indicator that the approach by Diaconis and Saloff-Coste's hypercontractive \cite{diaconis1996logarithmic} estimate (also used in \cite{lugiewicz2007coercive}) can not be used in proving lower bounds for the CMLSI constants.

The following lemma calculates the derivatives of the entropy functional $H(\rho)=\tau(\rho\log \rho)$. Recall the BKM metric of a operator $X\in \M$ at a base state $\rho$ is
\[\gamma_\rho(X)=\int_{0}^\infty\tau(X^*(\rho+s)^{-1}X (\rho+s)^{-1})\pl.\]
\begin{lemma}\label{eq:2nd}
Let $t\mapsto \rho_t \in S_B(\M), t\in (a,b)$ be a smooth family of bounded density operator.
Define the function $F(t)=H(\rho_t)=\tau(\rho_t\log \rho_t)$. Then
\begin{align*}F'(t)=\tau(\rho_t'(\log \rho_t+1))\pl,\pl F''(t)=\tau(\rho_t''(\log \rho_t+1))+\gamma_{\rho_t}(\rho'_t)
\end{align*}
 where $\rho_t'$ and $\rho_t''$ are the first and second order derivative of $\rho_t$.
\end{lemma}
\begin{proof}
The formula for $F'$ follows from \cite[Lemma 5.8]{wirth2018noncommutative}. For the second derivative, recall the noncommutative chain rule
\[\frac{d}{dt}(\log \rho_t)=\int_{0}^\infty(\rho_t+s)^{-1}\rho_t'(\rho_t+s)^{-1}ds\pl.\]
By calculating the second derivative, we obtain the second assertion
\begin{align*}
F''(t)=&\tau(\rho_t''(\log \rho_t+1))+\int_{0}^\infty \tau(\rho_t'(\rho_t+s)\rho_t'(\rho_t+s))ds
 \lel \tau(\rho_t''(\log \rho_t+1))+\gamma_{\rho_t}(\rho'_t) \qedhere
\end{align*}
\end{proof}

We now restate and prove Proposition \ref{prop:main4}.
\begin{prop}\label{prop:rothaus}
Let  $T_t=e^{-tL}:L_\infty(\Omega,\mu)\to L_\infty(\Omega,\mu)$ be an ergodic symmetric Markov semigroup. Let $\al_R,\al_2,\al_h$ be the optimal (largest) constant such that the following inequalities holds for any $f\in L_\infty(\Omega,\bM_2)\cap \A_\E$,
\begin{align}&\al_R\pl  D\Big(|f|^2 \left| \right|E_\mu(|f|^2)\Big)\le D\Big(|\hat{f}|^2 \left| \right| E_\mu(|\hat{f}|^2)\Big)+\norm{\hat{f}}{2}^2\pl, \tag{\text{Rothaus}} \\
&\al_2\pl  D(f^2||E_\mu(f^2))\le 2\E(f,f) \tag{\text{LSI}}\\
&\norm{T_tf}{L_2(\bM_2,L_{p(t)}(\Omega))}\le \norm{f}{L_2(\bM_2,L_2(\Omega))} \text{ for } p(t)=1+e^{2\al_h t} \tag{\text{Hypercontractivity}}
\end{align}
where $E_\mu(f)=(\int f d\mu)1_\Omega$ is the expectation map and $\hat{f}=f-E_\mu(f)$ is the mean zero part of $f$.
Then $\al_R=\al_2=\al_h=0$.
\end{prop}
\begin{proof}
We write $\tau(f)=\frac{1}{2}\int \tr(f)d\mu$ for the normalized trace on $L_\infty(\Omega,\bM_2)$. We start with constant $\al_R$ in the Rothaus lemma. Without loss of generality, we may assume there is a measurable set $X\subset \Omega$ such that $\mu(X)=r$ for some $0<r<1$. Let $\eta\in(0,1)$. Then $h_0=(1-r)1_X-r 1_{X^c}$ is a real mean zero function.
Consider the matrix valued function
$f_\eps=f+\eps h$ where
\[f=\left[\begin{array}{cc}1+\eta &0 \\ 0& 1-\eta
\end{array}\right]{\bf 1}  \pl, \pl h=\left[\begin{array}{cc}0 & h_0 \\ h_0& 0
\end{array}\right]\pl,
\]
where $f$ is a constant matrix valued function.
Then $E_\mu f_\eps=f, \hat{f}_\eps=\eps h$ and
\begin{align*}&f_\eps^2=(f+\eps h)^2=f^2+\eps(fh+hf)+\eps^2 h^2=f^2+2\eps h+\eps^2 h^2.
\end{align*}
Then $E_\mu(|f_\eps|^2)=f^2+\eps^2 h^2$ and $E_\mu(|\hat{f}_\eps|^2)=E_\mu( h^2)\eps^2$. Using Lemma \ref{eq:2nd}, the Taylor expansion of the left hand side of (LSI) is
\begin{align*}
D\Big(|f|^2 \left| \right|E_\mu(|f|^2)\Big)=&D(f^2+2\eps h+\eps^2 h^2\left| \right|f^2+\eps^2 h^2)\\
=&H(f^2+2\eps h+\eps^2 h^2)-H(f^2+\eps^2 h^2)\\
=&2\tau(h\log f)\eps+\big(2\tau(h^2(\log f+1))+\gamma_{f}(2h)\big)\eps^2+O(\eps^3)\\&-\tau(2E(h^2)(\log f+1))\eps^2+O(\eps^3)\\
=&\gamma_{f}(2h)\eps^2+O(\eps^3) \pl ,
\end{align*}
where we used the fact $\tau(h\log f)=0$ and $\tau(h^2\log f-E_\mu(h^2)\log f)=0$. For the right hand side of Rothaus lemma, we find
\begin{align*}
D(|\hat{f}_\eps|^2||E_\mu(|\hat{f}_\eps|^2))=D(h^2||E_\mu( h^2))\eps^2\pl, \pl
\norm{\hat{f}_\eps}{2}^2=\norm{h}{2}^2\eps^2\pl.
\end{align*}
While both $D(h^2||E_\mu( h^2))$ and $\norm{h}{2}^2$ are finite, we have
\begin{align*}\gamma_{f}(2h)=&4\int_{0}^\infty\tau(h(f+s)^{-1}h (f+s)^{-1})ds\\
=&4\int_{0}^\infty \int_\Omega \tr(\left[\begin{array}{cc}\frac{1}{(1-\eta+s)(1+\eta+s)}h^2 &0 \\ 0& \frac{1}{(1-\eta+s)(1+\eta+s)}h^2
\end{array}\right]{\bf 1}_\Omega )d\mu ds\\
=&4\Big(\int_{0}^\infty \frac{1}{(1-\eta+s)(1+\eta+s)} ds\Big)\norm{h}{2}^2
\\ =& \frac{2}{\eta}\ln\frac{1+\eta}{1-\eta}\norm{h}{2}^2
\end{align*}
Note that we can choose $\eta\to 1$ and $\frac{1}{2\eta}\ln(\frac{1+\eta}{1-\eta})\to +\infty$, which implies $\al_R=0$. The same example applies to LSI by choosing a mean zero function $h_0$ such that $\E(h_0,h_0)<\infty$. For the hypercontractivity, for $p\ge 2$ we recall the norms
\begin{align*} &\norm{f}{L_2(\bM_2, L_2(\Omega))}=\norm{f}{L_2(\bM_2)\ten L_2(\Omega)}=(\int \tr(f^*f)d\mu)^{1/2}\pl. \\
 &\norm{f}{L_2(\bM_2, L_p(\Omega))}=\inf_{x,y\in (\bM_2)_+\pl,\pl  \norm{\pl x\pl }{2r}=\norm{\pl y\pl}{2r}=1}\norm{x^{-1}fy^{-1}}{L_2(\bM_2, L_2(\Omega))}\end{align*}
 where the infimum takes over all positive invertible $x,y\in \bM_2$ with unit $2r$-norm for $\frac{1}{r}=\frac{1}{2}-\frac{1}{p}$. Since $T_t$ is a bimodule map for  $\mathbb{C}1\ten \bM_2\subset L_\infty(\Omega,\bM_2)$, we can equivalently consider the  norm
 \[\norm{T_t: L_2(\bM_2, L_2(\Omega))\to L_2(\bM_2, L_p(\Omega))}{}=\norm{T_t: L_2(\bM_2, L_2(\Omega))\to L_2^a(\bM_2, L_p(\Omega))}{}\pl,\]
 where the asymmetric amalgamated $L_2^a(\bM_2, L_p(\Omega))$ space is equipped with norm
\begin{align*}\norm{f}{L_2^a(\bM_2, L_p(\Omega))}=\inf_{a\in (\bM_2)_+\pl, \norm{\pl a\pl }{r}=1}\norm{fa^{-1}}{L_2(\bM_2, L_2(\Omega))} \end{align*}
In particular,
\[ \norm{f}{L_2^a(\bM_2, L_p(\Omega))}^2=\norm{f^*f}{L_1(\bM_2, L_{\frac{p}{2}}(\Omega))}\]
and we have
\[D(|f|^2||E(|f|^2))=\lim_{q\to 1^+}\frac{\norm{|f|^2}{L_1(\bM_2, L_q(\Omega))} -\norm{|f|^2}{1}}{q-1 }\]
Now define $p(t)=2q(t)=1+e^{2\al_h t}$
\[G(t)=\norm{T_t f}{L_2^a(\bM_2, L_p(t)(\Omega))}^2=\norm{|T_tf|^2}{L_1(\bM_2, L_{q(t)}(\Omega))}\pl \pl, \]
By assumption $G(t)\le 1$, we have
\[G'(0)=-2\E(f,f)+\al_hD(|f|^2||E(|f|^2))\le 0,\]
which implies $\al_h\le \al_2=0$. Note however, that $\al_h\ge 0$ because $T_t$ is always contractive on $L_2(\bM_2, L_2(\Omega))$. Hence $\al_h=0$ and the proof is complete.
\end{proof}

\begin{rem}{\rm Similar to \cite[Corollary 5.2]{bardet2022hypercontractivity}, the above proposition implies that for $p\neq 2$, neither $L_2^a(\bM_2, L_p(\Omega))$ or $L_2(\bM_2, L_p(\Omega))$ are uniformly convex.}
\end{rem} 

\section{Entropy Contraction for GNS Symmetric Quantum Channels}\label{sec:GNS}
\subsection{State symmetric quantum channels}\label{sec:GNSsymmetric}
Let $\M$ be a von Neumann algebra and $\phi$ a normal faithful state. We have the GNS cyclic representation $\{\pi_\phi, H_\phi, \eta_\phi\}$, which is a $*$-isomorphism $\pi_\phi:\M\to H_\phi$ with a cyclic and separating vector $\eta_\phi$ such that
\[ \phi(x)=\lan \eta_\phi, \pi_\phi(x)\eta_\phi \ran \pl, \pl x\in \M\pl.\]
By identifying $\M\cong \pi_\phi(\M)$, the modular automorphism group $\al^\phi_t$ for $t\in \mathbb{R}$ is defined as
\[\al^\phi_t:\M\to \M\pl ,\pl  \al^\phi_t(x)=\Delta^{it}x \Delta^{-it} \pl, \pl  x\in \M\pl,\]
where $\Delta$ is the modular operator of $\phi$, defined as follows
\[\Delta=S^*\bar{S}\pl, \pl S(\pi_\phi(x)\eta_\phi)=\pi_\phi(x^*)\eta_\phi \pl.\]
We consider the following two symmetric conditions with respect to a state $\phi$.

\begin{defi}\label{defi:GNS}
We say a quantum Markov map $\Phi:\M \to \M$ is GNS-symmetric with respect to $\phi$ (in short, GNS-$\phi$-symmetric) if \[\phi(\Phi(x)y)=\phi(x\Phi(y)), \pl \forall \pl x,y\in \M\pl; \]
ii) We say $\Phi$ is KMS-symmetric with respect to $\phi$ (in short, KMS-$\phi$-symmetric) if
\[\lan \Delta^{\frac{1}{4}} x\eta_\phi, \Delta^{\frac{1}{4}}\Phi(y)\eta_\phi \ran=\lan \Delta^{\frac{1}{4}} \Phi(x)\eta_\phi, \Delta^{\frac{1}{4}}y\eta_\phi \ran, \pl \forall \pl x,y\in \M\pl. \]
Correspondingly, we call the pre-adjoint $\Phi_*:\M_*\to \M_*$ a GNS- or KMS-$\phi$-symmetric quantum channel.
\end{defi}
Both definitions are generalizations of the detailed balance condition for classical Markov chains. It is proven in \cite{wirth2022christensen} that GNS-$\phi$-symmetric quantum Markov map is equivalent to KMS-$\phi$-symmetric plus that $\Phi$ commutes with the modular group
\[\al_t^{\phi}\circ\Phi\lel \Phi\circ\al_t^{\phi}\pl, \pl t\in \bR\pl.\]
Both KMS and GNS implies $\phi=\phi\circ\Phi=\Phi_*(\phi)$ is an invariant state of $\Phi$.

For simplicity, we will consider a semifinite von Neumann algebra $\M$ equipped with a normal faithful semi-finite trace $\tau$, but our discussion applies to general von Neumann algebras with proper interpretation of notations. In the tracial setting, we can write $\phi(x)=\tau(d_\phi x)$ using the density operator $d_\phi$ of $\phi$. Then the modular automorphism group is given by
 \[ \al_t^{\phi}(x) \lel d_{\phi}^{-it}xd_{\phi}^{it}\pl, \pl x\in \M \pl, t\in \mathbb{R}\pl.\]
Let $\Phi_*:L_1(\M)\to L_1(\M)$ be the pre-adjoint quantum channel via trace duality. The KMS-$\phi$-symmetry is equivalent to
\begin{equation}\label{eq:pre}
 \Phi_*(d_\phi^{1/2}xd_\phi^{1/2})
 \lel d_\phi^{1/2}\Phi(x)d_\phi^{1/2}\pl, \pl \forall x\in \M
 \end{equation}
For $1\le p\le \infty$, the weighted $L_p$-space $L_p(\M,\phi)$ is the completion of $\M$ under the norm
\[\norm{x}{p,\phi}=\norm{d_\phi^{1/2p} x d_\phi^{1/2p}}{p}\pl,\]
where $\norm{y}{p}=\tau(|y|^p)^{1/p}$ is the tracial $p$-norm.
For $p=2$, $L_p(\M,\phi)$ is a Hilbert space with KMS-inner product $\norm{x}{2,\phi}^2=\lan \Delta^{\frac{1}{4}}x\eta_\phi ,\Delta^{\frac{1}{4}}x\eta_\phi\ran$. By equation \eqref{eq:pre}, $\Phi$ is also a contraction on $L_1(\M,\phi)$, and hence a contraction on $L_p(\M,\phi)$ for all $1\le p\le \infty$ by complex interpolation.

The lemma below is an analog of Proposition \ref{prop:cd}.

\begin{prop}\label{prop:cd2}Let $\Phi:\M\to\M$ be a GNS-$\phi$-symmetric quantum Markov map for a normal faithful state $\phi$. Denote $\N$
as the multiplicative domain of $\Phi$. Then
 \begin{enumerate}
 \item[i)] $\N$ is invariant under $\al_t^\phi$. Hence there exists a $\phi$-preserving normal conditional expectation $E:\M\to \N$.
 \item[ii)] $\Phi|_\N$ is an involutive $*$-automorphism satisfying
\begin{align}\label{eq:cd} \Phi^2\circ E=E\circ \Phi^2=E\pl, \pl   E\circ \Phi=\Phi\circ E\pl . \end{align}
Moreover, $\Phi^2$ is a $\N$-bimodule map satisfying $\Phi^2(axb)=a\Phi^2(x)b$ for any $a,b\in\cN$ and $x\in\cM$.
\item[iii)] $\Phi$ is an isometry on $L_2(\N,\phi)$. If in addition, \[{\norm{\Phi(\id-E):L_2(\M,\phi)\to L_2(\M,\phi)}{2}<1}\pl,\] then $E=\lim_{n}\Phi^{2n}$ as a map from $L_2(\M,\phi)$ to $L_2(\M,\phi)$ .
    \end{enumerate}
 \end{prop}
\begin{proof} It suffices to explain i). The rest follows similar as Proposition \ref{prop:cd} (see also \cite[Lemma 2.5]{gao2022complete} for the finite dimensional case). Indeed, since $\Phi$ commutes with $\al_t^\phi$, for $x\in \N$
\begin{align*}
\Phi\big(\al_t^\phi(x)y\big)=&\Phi\big(\al_t^\phi(x \al_{t}^{\phi}\circ\al_{-t}^\phi(y) )\big)=\al_t^\phi\circ\Phi\big(x\al_{t}^{\phi}\circ \al_{-t}^\phi(y) \big)\\ =&\al_t^\phi\Big(\Phi(x)\Phi(\al_{t}^{\phi}\circ\al_{-t}^\phi(y)) \Big)=\al_t^\phi\circ\Phi(x) \al_t^\phi\circ \Phi\circ \al_{-t}^\phi(y)=\Phi(\al_t^\phi(x))\Phi(y)\pl.
\end{align*}
The other side multiplicativity is similar and this implies $ \al_t^\phi(x)\in\N$. By Takesaki's theorem \cite{takesaki1972conditional},
there exists $\phi$-preserving conditional expectation satisfying the defining property
\[\phi(xy)=\phi(xE(y)) \pl \forall \pl x\in \N, y\in \M\pl,\]
from which the GNS-$\phi$-symmetricness follows.
\end{proof}


\subsection{Haagerup's reduction}
A von Neumann algebra $\M$ is called type {III} if it does not admit a non-trivial semifinite trace.
We briefly review the basics of Haagerup's construction and refer to \cite{haagerup2010reduction} for more details. The key idea is to consider the additive subgroup $G=\bigcup_{n\in \bN} 2^{-n}\bZ \subset \mathbb{R}$ of the automorphism group.
Let $\M\subset B(H)$ be a von Neumann algebra and $\phi$ be a normal faithful state.
One can define the crossed product by the action $\al^\phi:G\curvearrowright \M$
\[ \hat{\M}\lel \M\rtimes_{\al^{\phi}}G \pl , \]
$\hat{\M}$ can be considered as the von Neumann subalgebra  $ \hat{\M}=\{\pi(\M),\la(G)\}'' \subset \M\overline\ten B(\ell_2(G))$ generated by the embeddings
\begin{align}\label{eq:cross}
  \pi:\M\to \M\rtimes_{\al^{\phi}}G\pl, \pl &\pi(a)\lel \sum_{g} \al_{g^{-1}}(a) \ten \ketbra{g}\nonumber\\
  \lambda:G\to \M\rtimes_{\al^{\phi}}G,\pl  &\la(g)(\ket{x}\ten \ket{h})=\ket{x}\ten \ket{gh}\pl ,  \pl \forall \pl \ket{x}\in H\pl, \ket{h}\in \ell_2(G) \pl.
  \end{align}
 Basically, $\pi$ is the transference homomorphism $\M\to \ell_\infty(G,\M)$, and $\lambda$ is the left regular representation on $\ell_2(G)$.
 The set of finite sums
 $\{\sum_g a_g \la(g)\pl | \pl  a_g \in \M \}\subset \hat{\M}$ forms a dense $w^*$-subalgebra of $\hat{\M}$.
 In the following, we identify $\M$ with $\pi(\M)\subset \hat{M}$ (resp. $a$ with $\pi(a)$) and view $\M\subset \hat{\M}$ as a subalgebra.
 The state $\phi$ admits a natural extension as a normal faithful state on $\hat{\M}$
 \[ \hat{\phi}(\sum_g a_g \la(g)) \lel \phi(a_0) \pl.\]
 Moreover, \[E_\M:\hat{\M}\to \M\pl, \pl E_{\cM}(\sum_g a_g\la(g))=a_0\]is the canonical normal conditional expectation such that $\phi\circ E_\M\lel \hat{\phi}$.

The main object in Haagerup's construction is an increasing family of subalgebras
 \[ \M_n \lel \hat{\M}_{\psi_n}:=\{x\in \hat{\M} \pl |\pl    \al_t^{\psi_n}(x)=x \pl, \pl\forall\pl  t\in \bR\}\pl,\]
given by the centralizer algebra $\hat{\M}_{\psi_n}$ for a suitable family of states $\psi_n$ so that $\bigcup_n \M_n$ is $w^*$-dense in $\hat{\M}$. The state $\psi_n$ is defined via a Radon-Nikodym density w.r.t to $\hat{\phi}$
 \[ \psi_n(x) \lel \hat{\phi}(e^{-a_n}x) \pl ,\pl
 a_n \lel -i2^n\Log (\la(2^{-n})) \pl .\]
Here $\Log$ is the principal branch of the logarithmic function with $0\le \Log(z)<2\pi$.
 Each subalgebra $\M_n$ contains $\la(G)$ and there exists normal conditional expectation $E_{\M_n}:\hat{\M}\to \M_n$. Indeed, by the definition of $\psi_n$, the modular group $\al_t^{\psi_n}$ is $2^{-n}$ periodic. The explicit form (see \cite[Lemma 2.3]{haagerup2010reduction}) is given by
\[E_{\M_n}=2^{n}\int_0^{2^{-n}}\al_t^{\psi_n} dt\pl.\]
The normalized state $\tau_n=\frac{\psi_n}{\psi_n(1)}$ is a normalized trace on $\M_n$.
The key properties of $\M_n$ are summarized in \cite[Theorem 2.1 \& Lemma 2.7]{haagerup2010reduction}, which we state below.
\begin{theorem}\label{thm:hr}
With above notations, $\M_n$ is an increasing family of von Neumann subalgebras satisfying the following properties
\begin{enumerate}
\item Each $(\M_n,\tau_n)$ is a finite von Neumann algebra.
\item $\bigcup_{n\ge 1} \M_n$ is weak$^*$-dense in $\hat{\M}$.
\item There exists a $\hat{\phi}$-preserving normal faithful conditional expectation $E_{\M_n}:\hat{\M}\to \M_n$ such that
\[ \hat{\phi}\circ E_{\M_n}=\hat{\phi}\pl, \pl \al_t^{\hat{\phi}}\circ E_{\M_n}= E_{\M_n}\circ \al_t^{\hat{\phi}}\pl. \]
Moreover, $E_{\M_n}(x)\to x$ in $\sigma$-strong topology for any $x\in \hat{\M}$.
\end{enumerate}
\end{theorem}

We now look at the Haagerup reduction on the states. For a state $\rho\in S(\M)$,  $\hat{\rho}=\rho\circ E_\M$ is the canonical extension on $\hat{\M}$. We denote $\rho_n:=\hat{\rho}|_{\M_n}\in \M_{n,*}$ as the restriction state of $\hat{\rho}$ on the subalgebra $\M_n \subset\hat{\M}$. Note that the predual $\M_{n,*}$ can be viewed as a subspace of $\hat{\M}_*$ via the embedding
\[ \iota_{n,*}: \hat{\M}_{n, *}\to \hat{\M}_*\pl,  \iota_{n,*}(\omega)=\omega\circ E_{\M_n}\pl. \]
 Via this identification, $\rho_n=\hat{\rho}|_{\M_n}\circ E_{\M_n}=\hat{\rho}\circ E_{\M_n}=E_{\M_n,*}(\hat{\rho})\in \hat{\M}_*$. Moreover, by the weak$^*$-density of the family $\M_n$, $\rho_n\to \hat{\rho}$ converges in the weak topology. An immediate consequence is the following approximation of relative entropy.



\begin{lemma}
\label{lemma:ultra}
Let $\rho$ and $\sigma$ be two normal states of $\M$. Then
 \[ D(\rho||\si) \lel D(\hat{\rho}||\hat{\si})=\lim_{n\to \infty} D(\rho_n||\sigma_n) \pl .\]
\end{lemma}
\begin{proof} Let $\iota:\M\subset \hat{\M}$ be the inclusion map. Because $\hat{\rho}=\rho\circ E_\M$ is an extension of $\rho$ , $\iota_*(\hat{\rho})=\hat{\rho}|_\M=\rho$ and similarly for $\sigma$. Both $\iota: \M \to \hat{\M}$ and $E_\M:\hat{\M}\to \M$ are quantum Markov maps. Then by the data processing inequality,
\[D(\rho||\si)=D(\iota_*(\hat{\rho})||\iota_*(\hat{\sigma}))\le D(\hat{\rho}||\hat{\sigma})=D(E_{\M,*}(\rho)||E_{\M,*}(\sigma))\le D(\rho||\si). \]
Thus $D(\rho||\si)=D(\hat{\rho}||\hat{\sigma})$.
As for the limit, we have
\begin{align*}
D(\hat{\rho}\|\hat{\sigma})&\leq \liminf_{n}D(\rho_{n}\|\sigma_{n})\\
&=\liminf_{n} D(E_{\M_n,*}(\rho_n)||E_{\M_n,*}(\sigma_{n,*}))\\
&\le D(\hat{\rho}||\hat{\sigma}) \pl,
\end{align*}
where the equality follows from the lower semi-continuity of relative entropy (see e.g. \cite[Theorem 2.7]{hiai2021quantum}). The second inequality is another use the data processing inequality.
\qd

We shall also apply the Haagerup's reduction on GNS-symmetric maps. Let $\Phi:\M\to \M$ be a GNS-$\phi$-symmetric quantum Markov map. Its canonical extension map
\[\hat{\Phi}:\hat{\M}\to \hat{\M}\pl, \pl  \hat{\Phi}(\sum_g a_g\la(g))
 \lel \sum_g \Phi(a_g)\la(g) \pl \]
 is also a GNS-$\hat{\phi}$-symmetric quantum Markov map. Indeed, $\hat{\Phi}=\Phi\ten \id_{B(\ell_2(G))}|_{\hat{\M}}$ is the restriction of $\Phi\ten \id_{B(\ell_2(G))}$ on $\hat{\M}\subset \M\overline{\ten} B(\ell_2(G)$. It is clear that $\hat{\Phi}$ has the multiplicative domain \begin{align}\label{eq:Ncross}\hat{\N}:=\N\rtimes_{\al^{\phi}} G\pl.\end{align}
 where $\N$ is the multiplicative domain of $\Phi$. In particular, this crossed product is well-defined because $\al^\phi_t(\N)=\N$. Moreover, the $\hat{\phi}$-preserving conditional expectation $\hat{E}:\hat{\M}\to\hat{\N}$ is nothing but the canonical extension of $E:\M\to \N$.

Recall that we write $E_\M$ and $E_n$ as the normal conditional expectations from $\hat{\M}$ onto $\M$ and $\M_n$ respectively. The next lemma shows that the extension $\hat{\Phi}$ is well compatible with the approximation family $\M_n$.
\begin{lemma}\label{lemma:commute}Let $\Phi:\M\to \M$ be a GNS-$\phi$-symmetric quantum Markov map.
With the notations above,
\begin{enumerate}
\item[i)] $\hat{\Phi}$ commutes with $E_\M, \hat{E}$ and $E_{\M_n}$. In particular, $\hat{\Phi}(\M_n)\subset \M_n$.
\item[ii)] The restriction $\Phi_n=\hat{\Phi}|_{\M_n}$ is a normal unital completely positive map symmetric with respect to the tracial state $\tau_n$.
\item[iii)] Let $\N_n\subset \M_n$ be the multiplicative domain for $\Phi_n$. Then the restriction map $E_n:=\hat{E}|_{\M_n}:\M_n\to \N_n$ is the $\tau_n$-preserving conditional expectation.
\end{enumerate}
\end{lemma}

\begin{proof} The relation $\hat{\Phi}\circ E_\M=E_\M\circ \hat{\Phi}$ is clear from the definition of $\hat{\Phi}$, and $\hat{\Phi}\circ \hat{E}=\hat{E}\circ \hat{\Phi}$ follows from Lemma \ref{prop:cd2}.
Recall that $\psi_n(x)=\hat{\phi}(e^{-a_n}x)$ with density operator $e^{-a_n}\in \la(G)''$ and $\la(G)$ is in the centralizer of $\hat{\phi}$ \cite[Lemma 2.3]{haagerup2010reduction}. Then
\[ \al_t^{\psi_n} \lel u(t)^*\al_t^{\hat{\phi}}u(t)=\text{ad}_{u(t)}\al_t^{\hat{\phi}} \]
for the unitary $u(t)=e^{-ita_n}$. Note that $\hat{\Phi}$ commutes with $\al_t^{\hat{\phi}}$ by GNS-$\hat{\phi}$-symmetry, and also commutes with $\text{ad}_{u(t)}$ because $u(t)\in \la(G)''$ is in $\hat{\Phi}$'s multiplicative domain. Thus $\hat{\Phi}$ commutes with $\al_t^{\psi_n}$ and hence the conditional expectation $E_{\M_n}=2^{-n}\int_0^{2^{-n}}\al_t^{\psi_n}$. This proves i).

For ii), we note that for $x,y\in \M_n$
\[\psi_n(x\Phi_n(y))=\hat{\phi}(e^{-a_n}x\hat{\Phi}(y))=\hat{\phi}(\hat{\Phi}(e^{-a_n}x)y)=\hat{\phi}(e^{-a_n}\hat{\Phi}(x)y)=\psi_n(\Phi_n(x)y)\pl,\]
where we use the fact that $\hat{\Phi}$ is GNS-$\hat{\phi}$-symmetric and $e^{-a_n}\in \la(G)''$ is in the fixed point subspace of $\hat{\Phi}$. Finally, iii) follows from applying i) and ii) to $\hat{E}$.  \qd

To summarize the lemma above, we have the following commuting diagrams.\\
\begin{figure}[h]
\begin{tikzpicture}
  \matrix (m) [matrix of math nodes,row sep=3em,column sep=4em,minimum width=2em]
  {
    {\color{white}\hat{l}}\M & \hat{\M} & \M_n{\color{white}\hat{l}} \\
     {\color{white}\hat{l}}\M & \hat{\M} & \M_n{\color{white}\hat{l}} \\};
  \path[-stealth]
    (m-1-1) edge node [left] {$\Phi$} (m-2-1)
    (m-1-2) edge node [above] {$E_\M$} (m-1-1)
    edge node [right] {$\hat{\Phi}$} (m-2-2)
     edge node [above] {$E_{\M_n}$} (m-1-3)
    (m-1-3) edge node [right] {$\Phi_n$} (m-2-3)
    (m-2-2) edge node [below] {$E_\M$} (m-2-1)
   edge node [below] {$E_{\M_n}$} (m-2-3);
\end{tikzpicture}\hspace{1cm}\begin{tikzpicture}
  \matrix (m) [matrix of math nodes,row sep=3em,column sep=4em,minimum width=2em]
  {
    {\color{white}\hat{l}}\M & \hat{\M} & \M_n{\color{white}\hat{l}} \\
     {\color{white}\hat{l}}\N & \hat{\N} & \N_n{\color{white}\hat{l}} \\};
  \path[-stealth]
    (m-1-1) edge node [left] {$E$} (m-2-1)
    (m-1-2) edge node [above] {$E_\M$} (m-1-1)
    edge node [right] {$\hat{E}$} (m-2-2)
     edge node [above] {$E_{\M_n}$} (m-1-3)
    (m-1-3) edge node [right] {$E_{n}$} (m-2-3)
    (m-2-2) edge node [below] {$E_\N$} (m-2-1)
   edge node [below] {$E_{\N_n}$} (m-2-3);
\end{tikzpicture}
\caption{Haagerup reduction of quantum Markov map and conditional expectation}
\label{figure}
\end{figure}\\
Basically, $\Phi_n$ is a family of trace symmetric channel which approximates $\hat{\Phi}$ as a natural extension of $\Phi$. The same picture holds for the conditional expectation $E_n,\hat{E}$ and $E$.

\subsection{Entropy contraction}\label{sec:GNSentropydecay}
We shall now discuss the entropy contraction of GNS-$\phi$-symmetric channels. The first step is to extend the entropy difference Lemma \ref{lemma:difference}. Define the state space that are bounded with respect to $\phi$,
\[S_b(\M,\phi)=\{\rho \in S(\M)\pl |\pl c^{-1}\phi\le  \rho\le c\phi\pl, \pl \text{for some } c>0 \}\pl.\]
For all $\rho\in S_b(\M,\phi)$, $D(\rho||\phi) <\infty$ is finite.  Such $S_b(\M,\phi)$ is a dense subset of $S(\M)$ because for any $\rho$ and $0<\eps<1$,
 $\rho_\eps=(1-\eps)\rho+\eps \phi\in S_b(\M)$.
For $\rho\in S_b(\M)$, we define the entropy difference for a GNS-$\phi$-symmetric quantum channel $\Phi_*$ as
\[D_{\Phi_*}(\rho):=D(\rho||\phi)-D(\Phi_*(\rho)||\phi)\pl.\]
By data processing inequality and $\Phi_*(\phi)=\phi$, $D_{\Phi_*}(\rho)\ge 0$. In the trace symmetric case,  $D_{\Phi_*}(\rho)=D(\rho||1)-D(\Phi_*(\rho)||1)=H(\rho)-H(\Phi_*(\rho))$ as in Section \ref{sec:unital}. Let $E$  be the conditional expectation onto the multiplicative domain of $\Phi$.
By the chain rule \cite[Theorem 2]{petz1991certain} that for any $E$
invariant state $\psi\circ E=\psi$,
\[ D(\rho||\psi)=D(\rho||E_*(\rho))+D(E_*(\rho)||\psi)\pl, \]
we have the alternative expressions $D_{\Phi_*}(\rho)=D(\rho||E_*(\rho))-D(\Phi_*(\rho)||\Phi_*E_*(\rho)),
$ where we used the property $\Phi_*E_*=E_*\Phi_*$  in Proposition \ref{prop:cd2}.

\begin{lemma}\label{lemma:difference2}Let $\Phi_*$ be a GNS-$\phi$-symmetric quantum channel.
 For any state $\rho,\om\in S_{b}(\M,\phi)$,
 \[ D(\rho||\Phi_*^2(\om))\kl D_{\Phi_*}(\rho)+D(\rho||\om) \pl .\]
\end{lemma}
\begin{proof}
Recall that we use $\rho_n=\hat{\rho}|_{\M_n}=E_{\M_n,*}(\hat{\rho})$ and $\om_n=\hat{\om}|_{\M_n}=E_{\M_n,*}(\hat{\om})$
as the restriction states on finite von Neumann algebra $\M_n\subset \hat{\M}$ obtained from the Haagerup reduction. By Lemma \ref{lemma:commute}, we know that $\Phi_n=\hat{\Phi}|_{\M_n}$ is a quantum Markov map symmetric with respect to the tracial state $\tau_n$. Thus by Lemma \ref{lemma:difference} in the tracial case,
 \[ D(\rho_n||\Phi_{n}^2(\om_n))
 \kl D_{\Phi_n}(\rho_n) +D(\rho_n||\om_n) \pl ,\]
where we identify $\Phi_{n}=\Phi_{n,*}$ by trace symmetry. Here, since $\Phi_n=\Phi|_{\M_n}$ is GNS-symmetric to $\phi_n=\phi|_{\M_n}$
\[D_{\Phi_n}(\rho_n)=D(\rho_n|| \tau_n)-D(\Phi_n(\rho_n)|| \tau_n)=D(\rho_n||\phi_n)-D(\Phi_n(\rho_n)|| \phi_n) .\]
By the definitions of $\Phi_n$ and $\rho_n$, and the hat "$\:\hat{}\:$" notation for states on $\hat{\M}$,
\begin{align*} &\Phi_{n}(\rho_n)=\hat{\rho}|_{\M_n}\circ \Phi_{n} =\hat{\rho}\circ \hat{\Phi}|_{\M_n}= \widehat{\Phi_*(\rho)}|_{\M_n}= \Phi_*(\rho)_n\pl ,\\
&\Phi_{n}^2(\rho_n)=\Phi_{n}(\Phi_*(\rho)_n)=\Phi_*^2(\rho)_n \pl .
\end{align*}
Then by Lemma \ref{lemma:ultra}, we can approximate every entropic term
\begin{align*}\lim_{n}D(\rho_n||(\Phi_{n})^2(\om_n))=&\lim_{n}D(\rho_n||\Phi_*^2(\om)_n)=D(\rho||\Phi_*^2(\om))\pl ,\\
\lim_{n} D_{\Phi_n}(\rho_n)=&\lim_{n}D(\rho_n|| \phi_n)-D(\Phi_n(\rho_n)|| \phi_n)\\=&\lim_{n}D(\rho_n||\phi_n)-D(\Phi_*(\rho)_n|| \phi_n)=D_{\Phi_*}(\rho)\pl ,\\
\lim_{n}D(\rho_n||\om_n)=&D(\rho||\om)\pl.  \qedhere
\end{align*}
 \qd

 The next lemma shows the CB-return time is also compatible with Haargerup reduction.
 \begin{lemma}\label{lemma:transfer} Let $\Psi:\M\to \M$ be a GNS-$\phi$-symmetric quantum Markov map and $E$ be the conditional expectation on its multiplicative domain. Suppose
 \begin{equation}\label{rett}
  (1-\eps) E \le_{cp} \Psi \le_{cp} (1+\eps)E \pl .
 \end{equation}
Then for all $\nen$,
 \[ (1-\eps)E_{n} \le_{cp} \Psi_n \le_{cp} (1+\eps)E_{n} \pl . \]
Moreover, if $0.9 E \le_{cp} \Psi \le_{cp} 1.1E$ and $\Psi\circ E=E$, then for any $\rho\in S_{B}(\M,\phi)$
 \[ \frac{1}{2}D(\rho||E_*(\rho))\le D(\rho||\Psi_*(\rho))\pl.\]
\end{lemma}

\begin{proof} The CP order inequality follows from the fact that
both maps $E_n$ and $\Psi_n$ are the restriction of $E\ten \id$ and $\Psi\ten \id$ on the subalgebra $\M_n\subset \hat{\M}\subset \M\overline{\ten} B(\ell_2(G))$. Then the entropy inequality can be obtained by the tracial case Lemma \ref{lemma:approximate} and approximation as in Lemma \ref{lemma:difference2}.
\end{proof}

 We then obtain the entropy contraction identical to the tracical case.
\begin{theorem}\label{thm:GNSsymmetric}
Let $\Phi:\M\to \M$ be a GNS-$\phi$-symmetric quantum Markov map and $E$ be the $\phi$-preserving conditional expectation onto its multiplicative domain $\N$. Define
 \[ k_{cb}(\Phi) \lel \inf\{k\in \bN^+ \pl |\pl  0.9 E\le_{cp}\Phi^{2k} \le_{cp} 1.1 E\pl \} \]
 Then, for any $\sigma$-finite von Neumann algebra $\cQ$, state $\rho\in S(\M\overline{\ten} \cQ)$,
 \[ D(\Phi_*\ten \id_{\cQ}(\rho)||(\Phi_*\circ E_*)\ten \id_{\cQ}(\rho)) \kl \Big (1-\frac{1}{2k_{cb}(\Phi)}\Big )D(\rho||E_*\ten \id_{\cQ}(\rho) )\pl .\]
\end{theorem}
\begin{proof}
For $\rho\in S_{b}(\M,\phi)$, the proof is same as the tracial case Theorem \ref{thm:unital} by using Lemma \ref{lemma:difference2} and Lemma \ref{lemma:transfer} above. The general case $\rho\in S(\M)$ can be approximated by $\rho_\eps= (1-\eps)\rho+ \eps\phi$.
\end{proof}
Recall that in finite dimensions the MLSI is defined as the supremum of $alpha$ such that
\[ 2\al D(\rho||E_*(\rho))\le I_L(\rho):=\tau(L_*(\rho)(\ln \rho-\ln\phi) )\pl.\]
The right hand side $I_L(\rho)$ is the entropy production, and the equivalence to entropy decay rely on the de Bruijn identity
\begin{align}\label{eq:debruijn} I_L(\rho)=-\frac{d}{dt}D(T_*(\rho)||E_*(\rho))|_{t=0}\pl.\end{align}
In infinite dimensions, the de Bruijn identity \eqref{eq:debruijn} is less justified even in $B(H)$ with $\dim(H)=+\infty$ (see discussions in \cite{huber2017geometric, konig2014entropy}). To avoid this issue, we define the MLSI on Type III von Neumann algebra as follows.
\begin{defi}For a GNS-$\phi$-symmetric quantum Markov semigroup $T_t=e^{-tL}:\M\to \M$, we define the modified log-Sobolev (MLSI) constant $\al_1(L)$ as the largest constant $\al$ such that
\begin{align}D(T_{t,*}(\rho)||E_*(\rho)) \kl e^{-2\al t}D(\rho||E_*(\rho))\pl, \pl \forall \rho\in S(\M)\pl,\label{eq:expdecay}\end{align}
where $E$ is the $\phi$-preserving conditional expectation onto the fixed point subalgebra $\N$.
The complete MLSI constant is then defined as $\al_c(L):=\sup_{\mathcal{\cQ}}\al(L\ten\id_{\cQ})$, where the supremum is over all $\sigma$-finite von Neumann algebra $\cQ$.
\end{defi}
This definition of MLSI also does not depend on any choice of reference state $\phi$ (see Lemma \ref{lemma:trace}).  With this definition, we obtain Theorem\ref{thm:main6}, which is restated below.\begin{theorem}\label{thm:GNSsemigroup} Let $T_t=e^{-tL}:\M\to \M$ be a GNS-$\phi$-symmetric quantum Markov semigroup.
Denote
$t_{cb} \lel \inf\{t>0 \pl|\pl  0.9 E\le_{cp}T_t \le_{cp} 1.1 E\} $.
Then
\[ \al_1\ge \al_c\ge  \frac{1}{2t_{cb}}.\]
Namely, for any $\sigma$-finite von Neumann algebra $\cQ$ and state $\rho\in S(\M\overline{\ten} \cQ)$, we have the exponential decay of relative entropy
 \[ D(T_{t,*}\ten \id_{\cQ}(\rho)||E_*\ten \id_{\cQ}(\rho)) \kl e^{-\frac{t}{t_{cb}}}D(\rho||E_*\ten \id_{\cQ}(\rho)) \pl,\pl t\ge 0\pl.\]
\end{theorem}
\begin{proof}This can be approximated using the tracial case Theorem \ref{thm:unital} as Lemma \ref{lemma:difference2} above.
\end{proof}

\begin{rem}{\rm In the above Haagerup's reduction, both $\hat{\Phi}$ and $\Phi_n$ are always  non-ergodic even given $\Phi$ is ergodic . From this point of view, our consideration for non-ergodic cases is essential even for ergodic $\Phi$. It also indicates that Haagerup's reduction does not work for LSI/hypercontractivity.}\end{rem}

As we have seen in Proposition \ref{prop:finitetcb} for the tracial case, a combination of heat kernel estimates and spectral gap allow us to bound CB return time. The same analysis remains valid in the GNS-$\phi$-symmetric case. For $1\le p\le \infty$, we define the $\phi$-weighted conditional $L_\infty^p(\N\subset \M, \phi)$ space as the completion of $\M$ under the norm
  \[ \norm{x}{L_\infty^p(\N\subset \M, \phi)}=\sup\{\norm{axb}{p,\phi}\pl |\pl a,b\in \N,\pl \norm{aa^*}{p,\phi}=\norm{b^*b}{p,\phi}=1\pl\}\pl.\]
For a GNS-symmetric $\N$-bimodule map $\Psi:\M\to \M$,  the equivalence in Proposition \ref{prop:equivalence} also holds,
\begin{align}(1-\eps) E\le_{cp} \Psi \le_{cp} (1+\eps) E\pl \Longleftrightarrow \pl  \norm{\Psi-E:L_\infty^1(\N\subset\M, \phi)\to L_\infty(\M)}{cb}\le \eps \label{eq:CPCB}\end{align}
Based on that, we have an analog of Proposition \ref{prop:finitetcb}.
\begin{prop}\label{prop:finitetcb3}Let $T_t:\M\to \M$ be a GNS-$\phi$-symmetric quantum Markov semigroup and $E:\M\to\N$ be the $\phi$-preserving conditional expectation onto the fixed point space. Suppose \begin{enumerate} \item[i)] the $\la$-Poincar\'e inequality that $\norm{T_{t}-E:  L_2(\M,\phi)\to L_2(\M,\phi)}{}\le e^{-\lambda t}\pl ,\pl \forall t\ge 0$;
\item[ii)] There exists $t_0$ such that $\norm{T_{t_0}:  L_\infty^1(\N\subset \M,\phi)\to L_\infty(\M)}{cb}\le C_0$.
  \end{enumerate}
  Then
  $ t_{cb}\kl \frac{1}{\lambda}\ln (10 C_0)+t_0$.
  In particular, if $C_{cb}(E)<\infty$, $ t_{cb}\kl \frac{1}{\lambda}\ln (10 C_{cb}(E))$.
\end{prop}

\begin{proof}The argument is similar to the tracial cases by using the property of $L_\infty^p(\N\subset \M,\phi)$ for general von Neumann algebra established in \cite{junge2010mixed}. See also \cite[Section 5]{bardet2021entropy} for argument in finite dimensional GNS-symmetric cases.
\end{proof}

\subsection{Applications to finite quantum Markov chains}
\label{sec:symmetricQMS}
Let $T_t=e^{-Lt}:\bM_d\to \bM_d$ be a quantum Markov semigroup on matrix algebra $\bM_d$. Its generator $L$ admits the following Lindbladian form (\cite{gorini1976completely,lindblad1976generators})
\[L(x)= i[h,x]+\sum_{j} \gamma_j(V_j^*[x,V_j]+[ V_j^*,x]V_j)\pl , \]
where $h,V_j\in \bM_d$ and $h=h^*$ is Hermitian. When $T_t$ is GNS-symmetric, one have the following simplified form \cite{alicki1976detailed,kossakowski1977quantum} that
 \[L(x)=\sum_{j}e^{-w_j/2}\Big(V_j^*[x,V_j]+[ V_j^*,x]V_j\Big)\pl , \]
where $\{V_j\}=\{V_j\}^*$ is an orthonormal set with respect to trace inner product and the eigenvector of modular group $\al_t^{\phi}(V_j)=e^{-iw_j t}V_j\pl.$ In finite dimensions, the completely Pimsner-Popa index $C_{cb}(E)$ is always finite.
Combining Theorem \ref{thm:GNSsemigroup} and Proposition \ref{prop:finitetcb3}, we obtain Corollary \ref{cor:main7},
\begin{align}\al_{c}\ge \frac{\lambda}{2\ln  (10 C_{cb}(E))}\label{eq:introfdQMS1}.\end{align}
which improves the bound $\al_c \ge \frac{\la}{2C_{cb}(E)}$ in the previous work of Gao and Rouz\'e \cite{gao2022complete}.
\begin{rem}{\rm
In the ergodic case $\N=\mathbb{C}1$, the conditional expectation $E_\phi(x)=\phi(x)1$ has index
\[C(E_\phi)=\norm{\phi^{-1}}{\infty}\pl, C_{cb}(E_\phi)\le \norm{\phi^{-1}}{\infty}^2\pl.\]
The above bound \eqref{eq:introfdQMS1}  gives
\[ \al_1\ge \al_{c}\ge \frac{\lambda}{2\ln10+4
\ln\norm{\phi^{-1}}{\infty}}\pl. \]
This can be compared to the bound
\begin{align} \al_1\ge \al_2\ge \frac{2(1-\frac{2}{\norm{\phi^{-1}}{\infty}})\lambda}{\ln(\norm{\phi^{-1}}{\infty}-1)} \label{eq:TK}\pl.\end{align}
This bound was proved Diaconis and Saloff-Coste \cite{diaconis1996logarithmic} for symmetric classical Markov semigroups. In the quantum case, it is only obtained for unital semigroups \cite{kastoryano2013quantum} and $d=2$ \cite{beigi2020quantum}.
For both classical and quantum depolarizing semigroup $L(x)=x-\phi(x)1$, this bound is known to be optimal for $\al_2$, which lower bounds $\al_1$.  Our results gives a general $\mathcal{O}(\frac{\lambda}{\norm{\phi^{-1}}{\infty}})$ lower bound for $\al_1$ for non-ergodic cases and also the complete constant $\al_c$.
}
\end{rem}

\begin{rem}{\rm The Corollary \ref{eq:loglog} shows that the CMLSI constant $\al_{c}$ for classical Markov semigroup is lower bounded by LSI constant $\al_2$ up to a $O(\log\log \norm{\mu^{-1}}{\infty})$ term. This argument does not work for Quantum Markov semigroup $T_t:\mathbb{M}_d\to \bM_d$ on matrix algebras, despite relation \eqref{eq:re}
remain valid for ergodic quantum Markov semigroup. The difference is that for matrix algebra, the bounded return time
\[t_b(e^{-2}):=\frac{1}{2}\inf\{t>0\pl |\pl \norm{T_t-E:L_1(\bM_d,\phi)\to L_\infty(\bM_d)}{}<1/e^2\}\]
and the CB return time of completely bounded norm
\begin{align*}
t_{cb}=\inf\{t>0 \pl | \pl \norm{T_t-E_\mu:L_1(\bM_d,\phi)\to L_\infty(\bM_d)}{cb}<1/10\pl.\}
\end{align*}
are quite different.
In the classical setting, we used the fact \[\norm{T:L_1(\Omega)\to L_\infty(\Omega)}{}=\norm{T:L_1(\Omega)\to L_\infty(\Omega)}{cb}\pl.\] So the $t_b(e^{-2})$ and $t_{cb}(0.1)$ are comparable by absolute constants.
In the noncommutative setting, we only have \[\norm{T_t-E_\mu:L_1(\bM_d,\phi)\to L_\infty(\bM_d)}{cb}\le d\norm{T_t-E_\mu:L_1(\bM_d)\to L_\infty(\bM_d)}{}.\] In the trace symmetric case, $\norm{\mu^{-1}}{\infty}=d$ and $t_{cb}(0.1)\le \frac{3}{2}t_b(e^{-2}) +\ln d$,
\[ \al_{c}\ge \frac{1}{2t_{cb}(0.1)}\ge  \frac{1}{3t_b(e^{-2})+ 2\ln d}\sim O(\frac{\al_2}{\ln d})\pl,\]
which is worse than the lower bound in the previous remark as $\al_2\le \lambda$.
}
\end{rem}

\subsection{Independence of invariant state}
The next lemma shows that the GNS-symmetry is also independent of the choice of invariant state $\phi$.
\begin{lemma}\label{lemma:trace}  Let $T:\M\to\M$ be a GNS-$\phi$-symmetric quantum Markov map for a normal faithful state $\phi$. Denote $E:\M\to\N$ as the $\phi$-preserving conditional expectation onto the multiplicative domain. Suppose $\psi$ is an another normal faithful state invariant under $E$, i.e. $\psi\circ E=\psi$. Then $T:\M\to\M$ is also  GNS-$\psi$-symmetric.
\end{lemma}

\begin{proof}
Without loss of generality we assume $\psi\le C\phi$ for some $C>0$. We first view them as the states on the subalgebra $\N$ by restriction. By \cite[Theorem 3.17]{takesaki2}, there exists $h\in \N$ such that
\[ \psi( x )=\phi(h^*xh ) \pl, \pl \forall x\in \N\pl. \]
This identity actually also holds for $y\in \M$. Indeed, because of $\phi\circ E=\phi$ and $\psi\circ E=\psi$,
\[ \psi( y )=\psi( E(y) )=\phi(h^*E(y)h )=\phi(E(h^*yh) )=\phi(h^*yh ) \pl, \pl \forall y\in \M\pl. \]
Moreover, one can replace $h$ by $T(h)$, because
\[ \psi( x )=\psi( T(x) )=\phi(h^*T(x)h )=\phi\circ T(T(h^*)xT(h) )=\phi(T(h^*)xT(h) )\]
where we use the fact that $T^2(h)=h$.
Thus the GNS-symmetry with respect to $\psi$ follows that for $x,y\in\M$,
\begin{align*} \psi( xT(y) )&=\phi(h^*xT(y)h )=\phi(h^*xT(yT(h))) =\phi(T(h^*x)yT(h))=\phi(T(h^*)T(x)yT(h))\\ =&\psi( T(x)y ) \pl,\end{align*}
where we used the multiplicative property of $T(axb)=T(a)T(x)T(b)$ for $a,b\in \N$.
The general case can be obtained via $\psi_{\eps}=(1-\eps)\psi+\eps \phi$. 
 \qd
We remark that if one has convergence $\lim_{n}\Phi^{2n}=E$ the in $L_2$-norm, the above $E$-invariant condition $\phi\circ E=\phi$ can be replaced by $\phi=\Phi^{2}\circ \phi$.

Note that the left hand side of \eqref{eq:CPCB} only relies on complete positivity. Indeed, the $L_\infty^1(\N\subset \M, \phi)$ norm at the right hand side is independent of the choice of the invariant state $\phi=\phi\circ E$.

\begin{lemma}\label{lemma:independent}
Let $\phi$ be a normal faithful state and $E:\M\to \N$ be a $\phi$-preserving conditional expectation. Suppose $\psi=\psi\circ E$ is another normal faithful state preserved by $E$. Then,
\[ \norm{x}{L_\infty^p(\N\subset\M, \phi)}=\norm{x}{L_\infty^p(\N\subset\M, \psi)}\pl, \pl \forall x\in \M\pl.\]
The identity extends to all $x\in L_\infty^p(\N\ssubset\M, \phi) $.
\end{lemma}
\begin{proof}
Note that if both $\phi$ and $\psi$ are $E$ invariant, then $d_{\psi}^{-\frac{1}{2p}}d_\phi^{\frac{1}{2p}}$ is affiliated to $\N$. Indeed, as argued in Lemma \ref{lemma:trace}, if $\psi\le C\phi$, then $d_\psi= hd_\phi h^*$ for some $h\in \N$, and the general case follows from approximation $\psi\le \frac{1}{\eps}((1-\eps)\phi+\eps \psi)$.
Then we have
\[\norm{aa^*}{\phi, p}=\norm{d_{\phi}^{\frac{1}{2p}}aa^*d_{\phi}^{\frac{1}{2p}}}{p}=\norm{d_{\psi}^{-\frac{1}{2p}}d_{\phi}^{\frac{1}{2p}}aa^*d_{\phi}^{\frac{1}{2p}}d_{\psi}^{-\frac{1}{2p}}}{\psi,2p} \]
Denote $a_1=d_{\psi}^{-\frac{1}{2p}}d_\phi^{\frac{1}{2p}}a$ and $b_1=bd_\phi^{\frac{1}{2p}}d_{\psi}^{-\frac{1}{2p}}$.
For $x\in \M$,
\begin{align*}
\norm{x}{L_\infty^p(\N\subset\M, \phi)}=&\sup_{\norm{\pl aa^*\pl}{\phi,2p}=\norm{\pl b^*b\pl}{\phi,2p}=1}\norm{axb}{\phi,p}
=\sup_{\norm{\pl aa^*\pl}{\phi,p}=\norm{\pl b_1b_1\pl }{\phi,p}=1}\norm{d_{\psi}^{-\frac{1}{2p}}d_\phi^{\frac{1}{2p}}axbd_{\phi}^{\frac{1}{2p}}d_{\psi}^{-\frac{1}{2p}}}{\psi,p}
\\=& \sup_{\norm{\pl a_1a_1^*\pl}{2p}=\norm{\pl b_1b_1^*\pl}{2p}=1}\norm{a_1xb_1}{\psi, p}=\norm{x}{L_\infty^p(\N\subset\M,\psi)}
\end{align*}
where the supremum are for $a,b\in \N$.
\end{proof}

\begin{rem}{\rm \label{rem:phitr} For finite $\M$, one particular invariant state of $E$ used in \cite{bardet2022hypercontractivity,bardet2021entropy} is
$\phi_{\tr}=E_*(1)$.
This state is convenient because $\phi_{\tr}|_\N$ is a trace. Then by Lemma \ref{lemma:independent}, we have
\[ \norm{x}{L_\infty^p(\N\subset \M, \phi)}=\norm{x}{L_\infty^p(\N\subset \M, \phi_\tr)}=\sup\{\norm{axb}{p,\phi_\tr}\pl |\pl  a,b\in \N,\pl \norm{a}{p,\phi_\tr}=\norm{b}{p,\phi_\tr}=1\pl\}\]
where we used the fact $L_p(\N,\phi_\tr)$ is a tracial $L_p$-space. We will use this point to simplify the discussion in Section \ref{sec:concentration}.
}
\end{rem}

\section{Applications and Examples}
\subsection{Entropy contraction coefficients}
In this section, we discuss the implications of our results on contraction coefficients studied in \cite{diaconis1996logarithmic, del2003contraction,muller2016entropy,gao2022complete}. These are analogs of functional inequalities for a single quantum channel.
\begin{defi}
Let $\Phi:\M\to \M$ be a quantum Markov map GNS-$\phi$-symmetric to a normal faithful state $\phi$ and $E:\M\to \N$ be the $\phi$-preserving conditional expectation onto the multiplicative domain of $\Phi$.
We define
\begin{itemize}
\item[i)] the $L_2$-contraction coefficient:
\begin{align}\lambda(\Phi):=\norm{\Phi(\id-E):L_2(\M,\phi)\to L_2(\M,\phi)}{}.\label{eq:PI}\end{align}
\item[ii)] the entropy contraction coefficient:
\[\al(\Phi):=\sup_{\rho} \frac{D(\Phi_*(\rho)||\Phi_*\circ E_*(\rho))}{D(\rho||E_*(\rho))}\pl.\]
\item[iii)] the complete entropy contraction coefficient $\al_c(\Phi):=\sup_{\cQ}\al(\id_\cQ\ten\Phi)$ where the supremum is over all $\sigma$-finite von Neumann algebras $\cQ$.
\end{itemize}
\end{defi}
The condition $\lambda(\Phi)<1$ can be viewed as a Poincar\'e inequality for a quantum channel $\Phi$, which implies the exponential convergence in $L_2$,
\[\norm{\Phi^n(X)-E(X)}{L_2(\M,\phi)}\le \lambda(\Phi)^n\norm{X-E(X)}{L_2(\M,\phi)}\to 0\pl.\]
Similarly, the entropy contraction coefficient gives the convergence in relative entropy
\[D(\Phi^n(\rho)||\Phi^n\circ E(\rho))\le \al(\Phi)^n D(\rho||E(\rho))\pl,\]
The complete constant $\al_c(\Phi)$ controls not only the entropy contraction of $\Phi$ but also $\id_\cQ\ten\Phi$ with any environment system $\mathcal{Q}$. This  leads to the tensorization property of $\al_c$ that for two GNS-symmetric quantum channels \cite{gao2022complete},
\begin{align}\al_c(\Phi_1\ten \Phi_2)=\max\{\al_c(\Phi_1),\al_c(\Phi_2)\}\label{eq:tensorizationchannel}\end{align}
For classical Markov maps, the tensorization property \eqref{eq:tensorizationchannel} is known to also hold for the non-complete constant $\al$. Nevertheless, for quantum Markov map (channel), this is not the case and $\al(\Phi)$ in general can be strictly less than $\al_c(\Phi)$ (see \cite[Section 4.4]{brannan2022complete}).

In finite dimensions, the existence of strictly contractive constant $\al_c(\Phi)<1$ was obtained in \cite[Theorem 4.1]{gao2022complete}. Our results gives an explicit estimate for $\al_c(\Phi)$.
\begin{cor}\label{cor:CSDPI} Let $\Phi$ be a GNS-symmetric quantum Markov map,
\[\lambda(\Phi)\le \al(\Phi)\le \al_c(\Phi)\le (1-\frac{1}{2k_{cb}(\Phi)})\le (1-\frac{-\ln \lambda(\Phi)}{\ln (10 C_{cb}(E))})\pl.\]
\end{cor}
\begin{proof}The estimate follows from Theorem \ref{thm:GNSsymmetric} and a discrete time analog of Proposition \ref{prop:finitetcb3}.
\end{proof}

\begin{rem}{\rm In the ergodic trace symmetric case $\N=\mathbb{C}1$ and $\M=\bM_d$,
  we have the trace map $E(x)=\tr(x)\frac{1}{d}$ and the CB-index $C_{cb}(E)= d^2$.
   The above estimate implies
   \begin{align}\label{eq:CSDPId} \lambda(\Phi)\le \al(\Phi)\le \al_c(\Phi)\le (1-\frac{-\ln \lambda(\Phi)}{\ln (10 d^2)})
   \end{align}

 This can be compared to
     \cite[Theorem 4.2]{muller2016entropy} and \cite[Corollary 27]{kastoryano2013quantum},
\begin{align}\label{eq:SDPId} \al(\Phi)\le 1-\frac{1}{2}\al_{2}(\id-\Phi^*\Phi)\le 1-\frac{(1-\la(\Phi)^2)^2(1-\frac{2}{d})}{\ln(d-1)}\pl,\end{align}
   where $\al_{2}(\id-\Phi^2)$ is the LSI constant of $\id-\Phi^2$ as a generator of quantum Markov semigroup. The two upper bounds in \eqref{eq:CSDPId} and \eqref{eq:SDPId} are comparable, as  both are asymptotically $\Theta(\frac{-\ln \lambda(\Phi)}{\ln d})$. The strength of our results is that \eqref{eq:CSDPId} also bounds the complete constant $\al_c(\Phi)$ which has tensorization property.
   }
   \end{rem}

\begin{rem}{\rm Our Lemma \ref{lemma:difference} implies
\[ 1-\al_1(\id-\Phi^2)\le \al(\Phi),\]
where $\al_1$ is MLSI constant of the semigroup generator $(\id-\Phi^*\Phi)$.
For classical Markov map, it was proved by Del Moral, Ledoux and Miclo \cite{del2003contraction} that there exists an universal constant $0<c<1$ such that
\begin{align}\label{eq:miclo} 1-\al_1(\id-\Phi^*\Phi)\le \al(\Phi)\le 1-c\al_1(\id-\Phi^*\Phi)\pl.\end{align}
To the best of our knowledge, the above upper bound in \eqref{eq:miclo} is open for quantum cases.
}
\end{rem}

\subsection{Graph random walks}
Let $G=(V,E)$ be a finite undirected graph with $|V|=d$ and the edge set $E\subset V\times V$. The discrete time random walk on $G$ is a finite Markov chain given by the stochastic matrix
\[K_G(u,v)=\begin{cases}
             \frac{1}{d(u)}, & \mbox{if } (u,v)\in E \\
             0, & \mbox{otherwise}.
           \end{cases}\]
           Here $d(u)$ is the degree of vertex $u\in V$. Then $K_G:l_\infty(V)\to l_\infty(V)$ is a Markov map.
If $G$ is not bipartite, $K_G$ admits a unique station distribution $\pi(u)=\frac{d(u)}{2m}$, where $|E|=m$. It is clear that $K_G$ is symmetric to the measure $\pi$, also called reversible. Hence $K_G$ is an ergodic  unital channel on $L_\infty(V,\pi)$ as $\pi(K_G(f))=\pi(f)$. The expectation map is $E_\pi(f)=\pi(f)1$ whose index is
\[ C_{cb}(E_\pi)=\norm{\pi^{-1}}{\infty}.\]
Both $K_G$ is connected and $E_\pi$ are symmetric operator on $L_2(V,\pi)$ and \[\la(K_G)=\norm{K_G-E_\pi:L_2(V,\pi)\to L_2(V,\pi)}{}<1\] if $K_G$ not bipartite (in the bipartite case $K_G$ has eigenvalue $-1$). Then our results implies
\begin{align}\al(K_G)\le \al_c(K_G)\le (1-\frac{1}{2k_{cb}(K_G)})\le (1-\frac{-\ln \lambda(K_G)}{\ln (10 \norm{\pi^{-1}}{\infty})})\pl. \label{eq:graph}\end{align}

\begin{exam}[Cyclic graphs]{\rm \label{ex:cyclic} Let us consider the cyclic graph $C_d=(V,E)$ with $d\ge 4$ where $V=\{1,\cdots,d\}$ and $E=\{(j,j+1)| j=1,\cdots, d \}$. Here the addition is understood in the sense of ``mod d".
 Then
\[ K_{C_d}(i,j)=\begin{cases}
             \frac{1}{2}, & \mbox{if } |i-j|=1 \\
             0, & \mbox{otherwise}.
           \end{cases}\]
           As $C_d$ is 2-regular, $K_{C_d}$ is symmetric to the uniform distribution $\pi(i)=1/d$. It is known that $K_{C_d}$ has spectrum \[\la_j=\cos(\frac{2\pi j}{d})\pl ,\pl j=0,\cdots, d-1\pl .\]
   The associated eigenvector is $e_j=\frac{1}{\sqrt{d}}(1,\omega^j,\omega^{2j},\cdots, \omega^{(d-1)j})$ where $\omega=\exp(\frac{2\pi i}{d})$.
When $d=2m+1$ is odd, $\pi$ is the unique stationary measure and $E_\pi$ is the projection onto the vector $e_0$. We have
\begin{align*} K_G^k-E_\pi= (K_G-E_\pi)^k=\sum_{j=1}^{2m}\la_j^k \ketbra{e_j}
\end{align*}
By triangle inequality, we have
\begin{align*} \norm{K_G^k-E_\pi:L_1(V,\pi)\to L_\infty(V,\pi)}{}&\le \sum_{j=1}^{2m}|\la_j|^k= 2\sum_{j=1}^{m}\cos(\frac{\pi j}{d})^k\\ &\le 2\frac{d}{\pi}\int_{0}^{\pi/2}\cos^k(x)dx
=2\frac{d}{\pi}W_k\le 2Cd\sqrt{\frac{1}{2k\pi }}
\end{align*}
where $C>0$ is some absolute constant by fact that the Wallis integrals $W_k=\int_{0}^{\pi/2}\cos^k(x)dx\sim \sqrt{\frac{\pi}{2k}}$. Thus
\[k_{cb}(K_{C_d})\le \frac{(10Cd)^2}{\pi}\sim \mathcal{O}(d^2), \]
and \eqref{eq:graph} implies
\[ \al(K_{C_d})\ge \al_c(K_{C_d})\ge 1-\mathcal{O}(d^{-2})\]
By Miclo's result \eqref{eq:miclo}, this is asymptotically tight because the MLSI constant $\al_1(I-K_G^2)\sim \mathcal{O}(d^{-2})$ (see Example \ref{ex:cyclic2} below for detials).
The similar asymptotic estimate also holds for even circle $d=2m$. 

}
\end{exam}

For the continuous time random walk, we consider $w:E\to (0,\infty)$ be a positive weighted function on the edge set $E$. The (weighted) graph Laplacian is given by the matrix
\begin{align*}
L_G(u,v)=\begin{cases}
           \sum_{e=(u,u')\in E}w_e, & \mbox{if } u=v \\
           -w_e, & \mbox{if } (u,v)\in E \\
           0, & \mbox{otherwise}.
         \end{cases}
\end{align*}
$L_G$ generates the continuous time random walk $T_t=e^{-L_G t}$ as a Markov semigroup and is  symmetric with to the uniform distribution $\pi$ on $V$. $T_t$ is ergodic  if and only if $G=(V,E)$ is connected.
The expectation map $E_\pi(f)=\pi(f){\bf 1}$ has index $C_{cb}(E_\pi)=d$. Then Corollary \ref{cor:main7},
\begin{align}\frac{ \la(L_G)}{2(\ln d+\ln 10)}\le \al_c(L_G)\le \al(L_G)\le \la(L_G)\pl. \label{eq:graph2}\end{align}
This lower bound of $\al_c(L_G)$ has better dependence on the dimension $d$ than \cite[Lemma 5.2]{li2020graph}.

\begin{exam}[Cyclic graphs]{\rm \label{ex:cyclic2} Let us again consider the cyclic graph $C_d$ with $d$ vertices. For the uniformly weighted case $w_e\equiv 1$, $L_{C_d}$ is a circulant matrix
\[L_{C_d}(i,j)=\begin{cases}
                 2, & \mbox{if } i=j \\
                 -1, & \mbox{if } |i-j|=1 \\
                 0, & \mbox{otherwise}.
               \end{cases}\]
Thus, $L_{C_d}=2(I-K_{C_d})$ where $K_{C_d}$ is the random walk kernel in Example \ref{ex:cyclic}, $ L_{C_d}$ has spectrum $\la_j=2(1-\cos\frac{2\pi j}{d})$. As discussed in \cite[Example 3.6]{diaconis1996logarithmic}, \[\norm{T_t-E:L_1(V,\pi)\to L_\infty(V,\pi) }{}\le 2\exp(-\frac{4t}{d^2})(\sqrt{1+d^2/4t})\pl\]
Choose $t_0=d^2$, we have
\[ \norm{T_t-E:L_1(V,\pi)\to L_\infty(V,\pi) }{}\le 2e^{-4}\sqrt{5/4}<\frac{1}{10}\pl.\]
Thus by Theorem \ref{thm:main1},
\[\frac{1}{2d^2}\le \al_c(L_{C_d})\le \al_1(L_{C_d})\le 2(1-\cos\frac{2\pi}{d})=\frac{8\pi^2}{d^2}+\mathcal{O}(\frac{1}{d^4})\pl. \]
This shows that our inverse of $t_{cb}$ bound for $\al_c$ is tight for this example up to absolute constant. Note that the LSI constant $\al_2(L_{C_d})$ is also $\Theta(\frac{1}{d^2})$.
}
\end{exam}
We refer to \cite{diaconis1996logarithmic,bobkov2006modified} more examples on spectral gap $\lambda$, Log-Sobolev constants $\al_2$,$\al_1$, and $L_\infty$ mixing time $t_b$ of finite Markov chains.

\subsection{A noncommutative Birth-Death process}\label{sec:birth}
Let us illustrate our estimate with a noncommutative Birth-Death process. This example is a generalization of graph Laplacians on matrix algebras (see \cite{li2020graph,junge2019stability} for similar constructions). 
To fix the notation, let  $G=(V, E)$ be an undirected graph with $n=|V|$ vertices and edge set $E$. For each edge $(r,s)\in E$, we introduce the edge Lindbladian on $\bM_n$,
\begin{align*} L_{rs}(x)=&e^{\beta_{rs}/2}L_{e_{rs}}(x) + e^{-\beta_{rs}}L_{e_{sr}}(x) \\
&= e^{\beta_{rs}/2}(e_{ss}x+xe_{ss}-2e_{sr}xe_{rs}) +e^{-\beta_{rs}/2}(e_{rr}x+xe_{rr}-2e_{rs}xe_{sr})\pl,
\end{align*}
where $e_{rs}\in \bM_n$ is the matrix unit with $1$ at the $(r,s)$ position. The total Lindbladian is a weighted sum over the edge set $E$,
\begin{align*}
L&\lel \sum_{(r,s)\in E} w(r,s) L_{rs}\\
&\lel 2\sum_{s\in V}\left(\sum_{(r,s)\in E}w(r,s)e^{\beta_{rs}/2}\right)(e_{ss}x+xe_{ss})-4\sum_{(r,s)\in E}w(r,s) e^{\beta_{rs}/2} e_{sr}xe_{rs}\pl,
\end{align*}
where we assume $\beta_{rs}=-\beta_{sr}$ and $w(r,s)=w(s,r)>0$ for the GNS-symmetry condition.
Note that for $j\neq k$,
\begin{align*} &L(e_{jk})=2(\sum_{(r,k)\in E}w(r,k)e^{\beta_{rk}/2}
+ \sum_{(r,j)\in E} w(r,j)e^{\beta_{r,j}/2})e_{jk} \pl,\\
&L(e_{jj}) = 4 \sum_{(r,j)\in E} w(r,j)(e^{\beta_{r,j}/2} e_{jj}-e^{-\beta_{r,j}/2}e_{rr}) \pl .
\end{align*}
Let us collect some relevant facts of such a Lindbladian $L$ as noncommutative extension of graph Laplacian.
 \begin{enumerate}
 \item[i)] Denote $\ell_{\infty}(V)\subset \bM_n$ as the diagonal subalgebra.
  $L(\ell_{\infty}(V))\subset \ell_{\infty}(V)$, and $L|_{\ell_{\infty}(V)}$ is a weighted graph Laplacian;
 \item[ii)] For $r\neq s$, the matrix unit $e_{rs}$ is an eigenvector of $L$
      \[ L(e_{rs}) \lel \gamma_{rs} e_{rs} ,\]
where $\gamma_{rs}=2(\sum_{(r,j)\in E}w(r,j)e^{\beta_{rj}/2}+\sum_{(k,s)\in E}w(k,s)e^{-\beta_{ks}/2})$.
     \item[iii)] $\ker(L)\subset \ell_{\infty}(V)$, and $\ker(L)=\mathbb{C}1$ if $\mathcal{G}=(V, E)$ is connected.
         \item[iv)] Let $\mu=(\mu_k)\in \ell_{\infty}(V)$ be a density operator in the diagonal subalgebra. Then $L$ is GNS-$\mu$-symmetric if $e^{\beta_{rs}}=\mu_s/\mu_r$ for any $s\neq r$.
\end{enumerate}
Assume $L=\sum_{(s,r)\in E}L_{sr}$ is an  ergodic graph Lindbladian satisfying GNS-$\mu$-symmetric condition for a diagonal density operator $\mu$. Denote $E_d$ as the projection onto diagonal subalgebra. We can decompose the semigroup $T_t=e^{-tL}$ on the diagonal part and off diagonal part.
\begin{align}
T_t=T_tE_d +T_t(\id-E_d):=T_t^{diag}+T_t^{off}\pl.\label{eq:bd-decomp}
\end{align}
It is clear from i) and ii) that $T_tE_d$ is a classical graph random walk and $T_t(\id-E_d)$ is a Schur multiplier on $\bM_n$. Using this decomposition, we consider the CB-return time of the semigroup
\[ t_{cb}(\eps):=\inf\{t>0\pl | \pl \norm{T_t-E_\mu: L_1(\bM_n,\mu)\to \bM_n}{cb}\leq \epsilon\}\]
satisfying
\[ t_{cb}(2\eps)\le t^{diag}_{cb}(\eps)+t^{off}_{cb}(\eps)\pl,\]
where $t^{diag}_{cb}$ and $t^{off}_{cb}$ are the CB-return time for the diagonal part $T_{t}E_{d}$ and off diagonal part $T_t(\id-E_d)$ respectively.
where \begin{align*}t_{cb}^{diag}(\epsilon)&=:\inf\{t>0\pl|\pl \|T_{t}E_{d}-E_{\mu}:L_{1}(v,\mu)\to L_{\infty}(V)\|_{cb}\leq \epsilon \}\\
t_{cb}^{off}(\epsilon)&=:\inf\{t>0\pl|\pl \|T_{t}(\id-E_{d}):L_{1}(\Mz_{n},\mu)\to \Mz_{n}\|_{cb}\leq \eps\}\end{align*}
For the diagonal part, $t_{cb}^{diag}(\epsilon)$ is a classical $L_\infty$ mixing time, i.e. the smallest $t$ such that
\[ \norm{T_tE_d-E_\mu: L_1(V,\mu)\to L_\infty(V)}{}\le \eps \pl.\]
For the off-diagonal term, we have by Effros-Ruan isomorphism that a Schur multiplier map
\begin{align*}\norm{T_t(\id-E_d): L_1(\bM_n,\mu)\to \bM_n}{cb}=&\norm{\sum_{r\neq s} \mu_{r}^{-1/2}e^{-\gamma_{rs}t}\mu_{s}^{-1/2}e_{rs}\ten e_{rs} }{\infty}\\ =&\norm{\sum_{r\neq s} \mu_{r}^{-1/2}e^{-\gamma_{rs}t}\mu_{s}^{-1/2}e_{rs} }{\infty}\pl.\end{align*}
Note that for each $t$,
\[A_t=\sum_{r\neq s} \mu_{r}^{-1/2}e^{-\gamma_{rs}t}\mu_{s}^{-1/2}e_{rs}\]
is a symmetric matrix with positive entry. A standard application of Schur's lemma for matrices with positive entries implies
\begin{align*}
\norm{A_t}{\infty}\le \sup_{r}\Big(\sum_{s}\mu_{r}^{-1/2}e^{-\gamma_{rs}t}\mu_{s}^{-1/2}\Big)\pl,
\end{align*}
which gives us an estimate for the off diagonal term $t_{cb}^{off}(\epsilon)$.

Now we consider the birth-death process on a finite state space $V=\{1,\cdots,n\}$, which we denote as $L^{BD}_n$. The corresponding edge $E$ set consists of only successive vertices $E=\{(j,j+1)| 1\le j\le n-1\}$. The simplest case chooses the uniform weight $w(r,s)=1$ for $(r,s)\in E$ and allows only one Bohr frequency $e^{-\beta}=\frac{\mu_{j}}{\mu_{j+1}}$, and the resulting stationary measure is the well-studied thermal state \[\mu=Z_{\beta}^{-1}(e^{-\beta j})_{j=1}^n\pl,\]
where $Z_{\beta}=\sum_{j=1}^{n} e^{-\beta j}$ is the normalization constant. In this case, $\gamma_{rs}=8(\cosh \beta) t$, and the off diagonal CB norm can be estimated by
\begin{align*} \norm{A_t}{\infty}\le & \sup_{r}\Big(\sum_{s=1}^n e^{\beta r/2}e^{\beta s/2}\Big) Z_\beta e^{-8(\cosh \beta) t}
\\ \le & e^{\beta \frac{n-2}{2}}\frac{1-e^{n\beta /2}}{1-e^{\beta /2}} \frac{ 1-e^{-n\beta }}{ 1-e^{-\beta}} e^{-8(\cosh \beta) t}\pl.
\end{align*}
Thus $t_{cb}^{off}(\eps)\le C_1(\beta)n$ for some constant $C_1(\beta)$ depending on $\beta$. For the classic part, We refer to \cite{Miclo99} and \cite{Chen} for the fact that the spectral gap is of order $O(1)$, i.e.
 \[ c(\beta) \le \la(L^{diag}_n)\le C_{2}(\beta) \]
for all $\nen$. For the commutative system on the diagonal part this implies (see also \cite{diaconis1996logarithmic})
 \[ t^{diag}_{cb}(\eps) \kl 2c(\beta)^{-1}(2+ |\log \mu_n|)\kl C_2(\beta)n \pl ,\]
 (for $\eps=e^{-2}$, but here the actual value of $\eps$ does not change the asymptotic estimate).
On the other hand, we have based on \cite{Miclo99} that \[t^{diag}_{cb}(0.1)\ge \al_1(L^{diag}_n)^{-1}\gl c(\beta)n\pl.\] Combining the diagonal and off diagonal part, we know  $t_{cb}(L^{BD}_n)\sim n$. It turns out CMLSI constant has asymptotic  $\al_c(L^{BD}_n)\sim \frac{1}{n}$, which indicates our  estimate $\al_{c}\ge \frac{1}{2t_{cb}}$ is asymptotically tight for this example.
\begin{theorem} \label{birthdeath} For $\beta > 0$, there exist constants $c(\beta), C(\beta)>0$ such that the CMLSI constant of noncommutative birth-death process $L_n^{BD}$ satisfies
 \[ \frac{c(\beta)}{n} \kl \al_c(L_n^{BD})\le \al_1(L^{BD}_n) \kl \frac{C(\beta)}{n} \pl .\]
The same $\Theta(\frac{1}{n})$ asymptotic holds for  $t_{cb}(L^{BD}_n)^{-1}$.
\end{theorem}

\begin{proof} It suffices to show that
\[\al_c(L^{BD}_{n})\le \al_1(L^{BD}_n)\le  \frac{C(\beta)}{n} \pl.\]
For this, we consider the function in the commutative system on the diagonal \[f(k)=\frac{Z(\beta)}{n}e^{\beta k} \text{ and } \sum_{k=1}^nf(k)\mu(k)=\sum_{k=1}^{n} \frac{Z(\beta)}{n}e^{\beta k}\frac{1}{Z(\beta)}e^{-\beta k}=1\]
so that $\rho:=f\mu$ represents a probability density. The relative entropy term satisfies
 \[ D(\rho||\mu)= D(f\mu||\mu)\lel \sum_k \frac{e^{-\beta k}}{Z(\beta)} f(k)(\beta k+\ln Z(\beta)-\ln n)\
 \lel \ln Z(\beta)-\ln n+ \beta \frac{n+1}{2} \pl .\]
Our density is $\rho\equiv(\frac{1}{n})$ and the reference density is $\mu(k)=\frac{e^{-\beta k}}{Z(\beta)}$.

Denote $a_k=|k\ran\lan k+1|$. On the diagonal, we have
\begin{align*}
  \frac{1}{2}L_*(f)&=
 \sum_k e^{\beta/2}(a_ka_k^*f-a_k^*fa_k)+ e^{-\beta/2}(a_k^*a_kf-a_kfa_k^*) \\
 &= \sum_k e^{\beta/2}(e_{k}f(k)-f(k)e_{k+1})
+ e^{-\beta/2}(f(k+1)e_{k+1}-f(k+1)e_k) \\
&= \frac{1}{Z(\beta)n} (e^{\beta/2}(e_0-e_{n})
 + e^{-\beta/2}(e_n-e_0)) \pl .
 \end{align*}
We have
\begin{align*}
L_{n,*}^{BD}(f)(k)=\begin{cases}
4 (e^{\beta/2}-e^{-\beta/2}), &\quad \text{if} \quad k=1;\\
0,&\quad \text{if} \quad k=2,n-1;\\
4 (e^{-\beta/2}-e^{\beta/2}), &\quad \text{if} \quad k=n.\\
\end{cases}
\end{align*}
Note that
  \[\ln \rho-\ln \mu \lel \ln f \lel \big(\beta k-\ln (Z(\beta)n)\big)_{k=1}^n \pl .\]
Then we have the entropy production
 \[ I_{L_{n}^{BD}}(\rho)=\tau(  L_{n,*}^{BD}(f) \ln f) \sim c(\beta) \pl \]
for some constant $c(\beta)$ only depending on $\beta$. This holds for $n\gl n_0$ large enough.
\qd

\begin{rem}
When $\beta > 0$, $\sum_{k=1}^n e^{-\beta k} = O(1)$ as a geometric series. In the case that $\beta = 0$,  the above birth-death process reduces to a ``broken" version of the cyclic graph (linear graph) as in Example \ref{ex:cyclic2} with $\al_c(L_n) \sim 1/n^2$.
\end{rem}

\subsection{Noncommutative concentration inequality}\label{sec:concentration}
In this section we show that $\text{CMLSI}$ of a GNS-$\phi$-symmetric semigroup implies concentration inequalities for the state $\phi$. The key quantity in the discussion is the Lipschitz semi-norm
 \[ \|x\|_{\Lip}^2 \lel \max\{ \norm{\Gamma_L(x,x)}{}\pl, \pl \norm{\Gamma_{L}(x^*,x^*)}{} \} \]
where the \emph{gradient form} (or Carr\'e du Champ operator) is
 \[ \Gamma_L(x,y) \lel \frac{1}{2}\Big(L(x^*)y+x^*L(y)-L(x^*y)\Big)\pl, \pl \forall x,y\in \text{dom}(L)\pl. \]
 Note that $\|\cdot\|_{\Lip}$ is a semi-norm (satisfying triangle inequality) because $\Gamma_L$ is completely positive bilinear form.
Our first lemma is to show that $\|x\|_{\Lip}$ can be approximated by Haagerup reduction.

\begin{lemma}\label{down} Let $x\in \M$. Then for all $\nen$,
 \[ \|E_{\M_n}(x)\|_{\Lip}\kl \|x\|_{\Lip} \pl .\]
\end{lemma}

\begin{proof} Recall the conditional expectation $E_{\M_n}:\hat{\M}\to \M_n$ is given by
\[ E_{\M_n}(x)=2^n\int_{0}^{2^{-n}} \al_t^{\psi_n}(x)d t \]
Note that $\al_t^{\psi_n}$ is an inner automorphism on $\M\rtimes_\al 2^{-n}\zz\cong L_\infty(\mathbb{T}, \M) $. We note that for a modular automorphism $\alpha_t$ such that $L\al_t=\al_tL$,
 \[ \Gamma_L(\al_t(x),\al_t(y))
 \lel \al_t(\Gamma_L(x,y)) \pl, \]
 which implies $\norm{x}{\Lip}=\norm{\alpha_t(x)}{\Lip}$.
Here both $\alpha_t^{\hat{\phi}}$ and $\alpha_t^{\psi_n}$ commutes with $\hat{L}=\id_{\mathbb{T}}\ten L$ by the GNS-symmetricness of $\hat{L}$.
Then by triangle inequality,
\begin{align*}
\norm{E_{\M_n}(x)}{\Lip} &= \Big \| 2^n\int_0^{2^{-n}} \al_{t}^{\psi_n}(x) dt \Big \|_{\Lip}\le 2^n\int_0^{2^{-n}} \norm{x }{\Lip}dt=\norm{x }{\Lip}\pl. \quad  \qedhere
\end{align*}
\qd
\begin{lemma}\label{lemma:commuting}Let $\M_0,\N\subset \M$ be two subalgebras and $\phi$ be a normal faithful state. Suppose $E_{0}:\M\to \M_0$ and $E:\M\to \N$ are $\phi$-preserving conditional expectations onto $\M_0$ and $\N$ respectively. Suppose $E\circ E_0=E_0\circ E$ satisfy the commuting square condition
\begin{center}
\begin{tikzpicture}
  \matrix (m) [matrix of math nodes,row sep=3em,column sep=4em,minimum width=2em]
  {
   \M & \M_0  \\
    \N & \N_0 \pl\pl, \\};
  \path[-stealth]
    (m-1-1) edge node [left] {$E$} (m-2-1)
    (m-1-1) edge node [above] {$E_0$} (m-1-2)
    (m-2-1) edge node [below] {$E_0$} (m-2-2)
    (m-1-2) edge node [right] {$E$} (m-2-2)
  ;
\end{tikzpicture} \end{center}
where $\N_{0}\subset \N$ is a subalgebra.
 Then for any $p\in [1,\infty]$ and any $x\in \M$,
\[ \norm{E_0(x)}{L_\infty^p(\N_0\subset \M_0,\phi)}=\norm{E_0(x)}{L_\infty^p(\N\subset \M,\phi)}\le \norm{x}{L_\infty^p(\N\subset \M,\phi)}\]
In other words, $L_\infty^p(\N_0\subset \M_0,\phi)\subset L_\infty^p(\N\subset \M,\phi)$ as a $1$-complemented subspace with projection $E_0$.
\end{lemma}
\begin{proof}We can assume $\phi=\phi_\tr$ in the Remark \ref{rem:phitr}. Using commuting square assumption, we know $E_0(a)\in \N_0$ for $a\in \N$.
By definition,
\begin{align*}
\norm{E_{0}(x)}{L_\infty^p(\N_0\subset \M_0, \phi)}= &\sup_{a,b\in\pl  \N_0}\norm{aE_{0}(x)b}{\phi,p}
\le \sup_{a,b\in\pl  \N_0}\norm{E_0(axb)}{\phi,p}\\
\le &\sup_{a,b\in\pl  \N_0}\norm{axb}{\phi,p} \le  \norm{x}{L_\infty^p(\N\subset \M, \phi)}
\end{align*}
where the supremum is for all $a,b$ in the corresponding subalgebra with $\norm{a}{\phi,p}=\norm{ b}{\phi,p}=1$. Now it suffices to show the other direction
$$\norm{x}{L_\infty^p(\N_0\subset \M_0,\phi)}\ge \norm{x}{L_\infty^p(\N\subset \M,\phi)}\pl,$$
for $x\in \M_0$. For that, we revoke that for $\frac{1}{p}+\frac{1}{q}=1$, $L_1^{p'}(\N\subset \M)\subset L_\infty^p(\N\subset \M, \phi)^*$ is as a weak$^*$-dense subspace \cite[Proposition 4.5]{junge2010mixed}. Here for $x\in \M$,
\[ \norm{y}{L_1^{q}(\N\subset \M)}=\inf_{y=azb}\norm{a}{2p,\phi}\norm{y}{q,\phi}\norm{b}{2p,\phi}\]
where the infimum is over all factorization $y=azb$ with $a,b\in \N, z\in \M$. The duality pairing is given by the KMS inner product,
\[\lan x,y \ran=\tau(x^*d_{\phi}^{1/2} y d_{\phi}^{1/2})=\lan x,y \ran_{\phi} \]
Indeed, it was proved in \cite[Corollary 3.13]{junge2010mixed} that
\[E_0: L_1^{q}(\N\subset \M)\to L_1^{q}(\N_0\subset \M_0)\]
is a contraction by the commuting square condition. Therefore, for $x\in \M_0$, by the  KMS-$\phi$-symmetry of $E_0$,
\begin{align*}\norm{x}{L_\infty^p(\N\subset \M, \phi)}=&\sup_{\norm{ \pl y\pl }{L_1^{q}(\N\subset \M)}=1}\lan x,y \ran_\phi \\ =& \sup_{\norm{ \pl y\pl }{L_1^{q}(\N\subset \M)}=1}\lan x,E_0(y) \ran_\phi
\\ \le& \sup_{\norm{ \pl z\pl }{L_1^{q}(\N_0\subset \M_0)}=1}\lan x,z \ran_\phi=\norm{x}{L_\infty^p(\N_0\subset \M_0, \phi)} \pl. \qedhere
\end{align*}
\end{proof}

\begin{lemma}\label{lemma:limit}For $x\in \M$,  $\lim_{n}\|E_{\M_n}(x)\|_{L_\infty^p(\N_n\subset \M_n, \psi_n)}= \|x\|_{L_\infty^p(\N\subset \M, \phi)}$
\end{lemma}
\begin{proof} Recall the commuting square condition $E_{\M_n}\circ \hat{E}=\hat{E}\circ E_{\M_n}$. By Lemma \ref{lemma:independent} \& \ref{lemma:commuting},
\begin{align*}\norm{E_{\M_n}(x)}{L_\infty^p(\N_n\subset \M_n, \psi_n)}=&\norm{E_{\M_n}(x)}{L_\infty^p(\N_n\subset \M_n,\hat{\phi})}
=\norm{E_{\M_n}(x)}{L_\infty^p(\N\subset \M,\hat{\phi})}\le \norm{x}{L_\infty^p(\N\subset \M,\hat{\phi})}\pl.
\end{align*}
The other direction follows from the weak$^*$-convergence $E_{\M_n}(x)
\to x$. Fix $\frac{1}{q}+\frac{1}{p}=1$.
 For any $\eps>0$, there exists $a_0,b_0\in \hat{\N}$ and $y_0\in \hat{\M}$ such that
\begin{align*}
&\norm{aa^*}{p,\hat{\phi}}=\norm{b^*b}{p,\hat{\phi}}=\norm{y}{\hat{\phi},q}=1\pl, \pl
&\hat{\tau}(d_{\hat{\phi}}^{1/2}axb d_{\hat{\phi}}^{1/2}y)\ge \norm{x}{L_\infty^p(\hat{\N}\subset \hat{\M}, \hat{\phi})}-\eps
\end{align*}
By the weak$^*$-density, we can choose $n_1,n_2,n_3$ and $n_4\ge \max\{n_1,n_2,n_3\}$ inductively such that
\begin{align*}
\tau(d_\phi^{1/2}E_{\M_{n_1}}(a)E_{\M_{n_4}}(x)E_{\M_{n_2}}(b)d_\phi^{1/2}E_{\M_{n_3}}(y))>\tau(d_\phi^{1/2}axbd_\phi^{1/2}y)-\eps >\norm{x}{L_\infty^p(\hat{\N}\subset \hat{\M}, \hat{\phi})}-2\eps\pl. \end{align*}
Since $E_{\M_{n}}(\hat{\N})=\N_n$ (see the commuting diagram after Lemma \ref{lemma:commute}), we have
\begin{align*}&\norm{E_{\M_{n_1}}(a)E_{\M_{n_1}}(a)^*}{\hat{\phi},p}\le \norm{E_{\M_{n_1}}(aa^*)}{\hat{\phi},p}\le \norm{aa^*}{\hat{\phi},p}= 1\\
&\norm{E_{\M_{n_3}}(b^*)E_{\M_{n_3}}(b)}{\hat{\phi},p}\le \norm{b^*b}{\hat{\phi},p}= 1\pl,
\\ & \norm{E_{\M_{n_4}}(y)}{\hat{\phi},q}\le \norm{y}{\hat{\phi},q}=1
\end{align*}
by the KMS-$\hat{\phi}$-symmetry of $E_{\M_n}$.
Then, for $n\ge n_4=\max\{n_1,n_2,n_3,n_4\}$,
\begin{align*}
\norm{E_{\M_{n}}(x)}{L_\infty^p(\hat{\N}_n\subset \hat{\M}_n, \hat{\phi})}&\ge \norm{E_{\M_{n_4}}(x)}{L_\infty^p(\hat{\N}_n\subset \hat{\M}_n, \hat{\phi})}\\ &\ge \tau(d_{\hat{\phi}}^{1/2}E_{\M_{n_1}}(a)E_{\M_{n_4}}(x)E_{\M_{n_2}}(b) d_{\hat{\phi}}^{1/2}E_{\M_{n_3}}(y))\\ &\ge  \norm{x}{L_\infty^p(\hat{\N}\subset \hat{\M}, \hat{\phi})}-2\eps.
\end{align*}
This proves
\[\lim_{n}\|E_{\M_n}(x)\|_{L_\infty^p(\N_n\subset \M_n, \psi_n)}= \|x\|_{L_\infty^p(\hat{\N}\subset \hat{\M}, \hat{\phi})}\pl.\]
Finally, the assertion follows from
\[ \|x\|_{L_\infty^p(\hat{\N}\subset \hat{\M}, \hat{\phi})}=\|x\|_{L_\infty^p(\N\subset \M, \phi)}\pl,\]
as a consequence of $E_0\circ \hat{E}=\hat{E}\circ E_0$ by Lemma \ref{lemma:commuting}.
\end{proof}

Now we restate and prove Theorem \ref{thm:main8}.
\begin{theorem}\label{thm:concentration} Let $\M$ be a $\sigma$-finite von Neumann algebra and $T_t=e^{-tL}$ be a GNS-$\phi$-symmetric quantum Markov semigroup. Suppose $T_t$ satisfies $\MLSI$ with parameter $\al>0$. There exists an universal constant $c$ such that for $2\le p <\infty$,
 \[ \al\|x-E(x)\|_{L_p(\M,\phi)}\le \al\|x-E(x)\|_{L_{\infty}^p(\N\ssubset \M,\phi)}\kl c\sqrt{p}\norm{x}{\emph{Lip}} \pl .\]
\end{theorem}

\begin{proof} We first show that if $T_t$ satisfies $\al$-MLSI, so does the approximation semigroup.
\[T_{n,t}=\hat{T}_t|_{\M_n}:\M_n\to \M_n\]
Indeed, as we see in the discussion above, $\M_n\subset \M\rtimes_{\al^{\phi}_t} 2^{-n}\mathbb{Z}\cong L_\infty(\mathbb{T}, \M)$ and the extension $\hat{T}_t=T_t\ten \id_{\mathbb{T}}$ has $\al$-MLSI (because $L_\infty(\mathbb{T})$ is a commutative space). Note that since $\M_n\subset \M\rtimes_{\al^{\phi}_t} 2^{-n}\mathbb{Z}\subset \M\rtimes_{\al^{\phi}_t} G$, the restriction
 $E_{\M_n}:\M\rtimes_{\al^{\phi}_t} 2^{-n}\mathbb{Z}\to \M_n$ is also a conditional expectation.
 Then for any $\rho,\sigma\in S(\M_n)$, we have
 \[D(E_{\M_n,*}\rho|| E_{\M_n,*}\sigma)\le D(\rho|| \sigma)=D(\rho|_{\M_n}|| \sigma|_{\M_n})\le D(E_{\M_n,*}\rho|| E_{\M_n,*}\sigma)\pl.\]
Using the commutation relation $T_{n,t}\circ E_{\M_n}=E_{\M_n}\circ \hat{T}_t$ and $E_{\M_n}\circ \hat{E}=E_n\circ E_{\M_n} $, we have for $\rho\in S(\M_n)$
\begin{align*}D(T_{t,n,*}\rho|| E_{n,*}\rho)=&D(E_{\M_n,*}T_{t,n,*}\rho|| E_{\M_n,*}E_{n,*}\rho)
 =D(\hat{T}_{t,*}E_{\M_n,*}\rho|| \hat{E}_{*}E_{\M_n,*}\rho)\\ \le& e^{-2\al t} D(E_{\M_n,*}\rho|| \hat{E}_{*} E_{\M_n,*}\rho)\\=& e^{-2\al t} D(E_{\M_n,*}\rho|| E_{\M_n,*}E_{n,*}\rho)= e^{-2\al t} D(\rho|| E_{n,*}\rho) \pl.
 \end{align*}
Thus $T_{n,t} $ has $\al$-$\MLSI$ on $\M_n$. Note that $T_{n,t}$ is both GNS-$\hat{\phi}$-symmetric for the extension state $\hat{\phi}$ and also symmetric for the trace $\psi_n$. Now, we may use   the tracial version of the concentration inequality \cite[Theorem 6.10]{gao2020fisher}  that for $x\in \M_n$
\[ \al\norm{E_{\M_n}(x)-E_nE_{\M_n}(x)}{L_\infty^p(\N_n\subset \M_n)}\le C\sqrt{p}\norm{E_{\M_n}(x)}{\Lip}\]
Now by the approximation of Lemma \ref{lemma:limit} and independence of $L_\infty^p(\N_n\subset \M_n)$ on the reference state, for $x\in \M$,
\begin{align*}
\norm{{x-E(x)}}{L_\infty^p(\N\subset \M,\phi)}=&\lim_{n}\norm{E_{\M_n}(x-E(x))}{L_\infty^p(\N_n\subset \M_n,\psi_n)}\\=&\lim_{n}\norm{E_{\M_n}(x)-E_{n}E_{\M_n}(x)}{L_\infty^p(\N_n\subset \M_n,\psi_n)}
\\ \le& C\sqrt{p}\norm{E_{\M_n}(x)}{\Lip}\le C\sqrt{p}\norm{x}{\Lip}\pl.
\end{align*}
The other inequality
\[ \norm{y}{L_\infty^p(\N\subset \M,\phi)}\ge \norm{y}{L_p(\M,\phi)}\]
 is clear from definition of $L_\infty^p(\N\subset \M,\phi)$.
\qd

For Gaussian type concentration property, we introduce the following definition. \begin{defi}For an operator $O$ we say that
 \[ \Prob_{\phi}(|O|>t) \kl \eps \]
if there exists a projection $e$ such that
 \[ \|eOe\|_{\infty} \le t \quad \mbox{and} \quad \phi(1-e)\le \eps \pl .\]
 \end{defi}
 The next lemma is a Chebyshev inequality for $\phi$-weighted $L_p$ norm.
\begin{lemma}\label{lemma:cheb}  Let $x\in L_p(\M,\phi)$ and $1<p<\infty$.  Then
\[ \Prob_\phi(|x|>t)  \kl  2\Big ( \frac{t}{4} \Big )^{-p}\|x\|^{p}_{p,\phi} \pl .\]
\end{lemma}
\begin{proof}We start with a positive element $x=y^2$ and assume $\|x\|_{p,\phi}=M$. Then  we have
 \[ M=\|x\|_{p,\phi} \lel \|d_{\phi}^{1/2p}xd_{\phi}^{1/2p}\|_p
 \lel \|yd_{\phi}^{1/2p}\|_{2p}^2 \pl .\]
Recall the asymmetric Kosaki $L_p$-space
\[ \norm{y}{L_{2p}^c(\M,\phi)}:=\|yd_{\phi}^{1/2p}\|_{2p} \pl, \]
and the complex interpolation relation \cite{junge2010mixed}
\[ L_{2p}^c(\M,\phi)=[\M,L_2^c(\M,\phi)]_{1/p} \pl, \]
and the relation between real and complex interpolation
 \[ L_{2p}^c(\M,\phi)=[\M,L_2^c(\M,\phi)]_{1/p}\subset [\M,L_2^c(\M,\phi)]_{1/p,\infty} \pl . \]
By the definition of real interpolation space, for every $s>0$ we have a decomposition $y \lel y_1+ y_2$ such that
 \[ \|y_1\|_{\infty}+ s\norm{y_2}{L_2^c(\M,\phi)} \kl s^{1/p}M^{1/2} \pl .\]
Then by Chebychev's inequality for the spectral projection $e=e_{[0,a]}(y_2^*y_2)$, we have
 \[ a \phi(1-e)\kl \phi(y_2^*y_2) \kl s^{2/p-2}M \quad \mbox{and} \quad  \|y_2e\|_{\infty}^2\lel \|ey_2^*y_2e\|_{\infty}\kl a \pl .\]
Choose $a=s^{2/p}M$  and deduce that
 \[ \|exe\|_{\infty}=\|ye\|_{\infty}^2 \kl (\|y_1e\|_{\infty}+ \|y_2e\|_{\infty})^2\kl  (s^{1/p}M^{1/2}+s^{1/p}M^{1/2})^{2}
  = 4a \pl .\]
Then for $t=4a$ and
 \[ \phi(1-e) \kl a^{-1} s^{2/p-2}M  \lel s^{-2} \lel (\frac{t}{4M})^{-p}=(\frac{t}{4})^{-p}M^p \pl.  \]
For an arbitrary $x$, we may write $x=x_1x_2$ such that
 \[ \|d_\phi^{1/2p}x_1\|_{2p}\lel \|x_2d_\phi^{1/2p}\|_{2p} \lel \norm{x}{p,\phi}=M \pl .\]
Then for each $s>0$, we have decomposition
 \[ x_1\lel x_{11}+x_{12} \pl ,\pl x_2=x_{21}+x_{22} \pl .\]
 with
\begin{align*}& \|x_{11}\|_{\infty}+ s\norm{x_{12}}{L_2^c(\M,\phi)}\kl s^{1/p}M^{1/2} \pl ,\pl 
\|x_{21}\|_{\infty}+ s\norm{x_{22}}{L_2^c(\M,\phi)} \kl s^{1/p}M^{1/2} \pl.\end{align*}
 We then use the Chebychev inequality for $e=e_{[0,a]}(x_{12}^{*}x_{12}+x_{22}^*x_{22})$,
 \[ a\phi(1-e)\le \phi(x_{12}^{*}x_{12} +x_{22}^*x_{22}) \kl 2s^{2/p-2}M \pl .\]
Take $a=s^{2/p}M$,
\begin{align*}
 \|exe\|_{\infty}
 &=  \|e(x_1x_2)e\|_{\infty} \lel \|e(x_{11}+x_{12})(x_{21}+x_{22})e\|_{\infty} \\
 &\kl \|x_{11}x_{22}\|_{\infty} + \|ex_{12}x_{21}\|_{\infty}+
 \|x_{11}x_{22}e\|_{\infty} + \|ex_{12}x_{22}e\|_{\infty} \\
 &\kl 4s^{2/p}M \pl .
 \end{align*}
Thus for $t=4s^{2/p}M$, by Chebychev's inequality for $e$,
 \begin{align*}
  \phi(1-e) &\leq \frac{1}{a}\phi(x_{12}^{*}x_{12}+x_{22}^{*}x_{22}) \le  a^{-1} 2s^{2/p-2}M \lel 2s^{-2}
  \lel 2 (\frac{t}{4M})^{-p}\pl . \qedhere
  \end{align*}
\qd

\begin{cor}\label{cor:concentration} Let $T_t=e^{-tL}$ be a GNS-$\phi$-symmetric quantum Markov semigroup. Suppose $T_t$ satisfies $\al$-MLSI. Then for any $x\in\M$ and $t>0$,
 \[ \Prob_{\phi}(|x-E_{fix}(x)|>t)
 \kl 2 \exp \Big ( -\frac{2}{e} \Big (\frac{\al t}{4c\norm{x}{\Lip}} \Big )^2 \Big )
\]
where $c$ is a universal constant as in Theorem \ref{thm:concentration}.
\end{cor}
\begin{proof} By Lemma \ref{lemma:cheb} and Theorem \ref{thm:concentration}, we have
 \[ \Prob_{\phi}(|x-E(x)|>t)\le 2 (t/4)^{-p}\norm{x-E(x)}{L_p(\M,\phi)}^p\le  2 \big(\frac{4c\norm{x}{\Lip}\sqrt{p}}{\al t}\big)^{p}\pl.\]
Minimizing over $p$ gives $p=\frac{1}{e}(\frac{\al t}{4c\norm{x}{\Lip}})^2$, which implies
\begin{align*}
&\Prob_{\phi}(|x-E(x)|>t)\le 2 \exp( -\frac{\al^2t^2}{16ec^2\norm{x}{\Lip}^2})\qedhere
\end{align*}
\qd


\begin{rem}{\rm  In the ergodic  case, the above results can be compared to \cite[Theorem 8]{rouze2019concentration}, which states that for self-adjoint $x=x*$
\begin{align*}
\phi(e_{\{|x-E(x)|>t\}})\le  \exp\Big( -\frac{\al t^2}{8\norm{d_\phi^{-1/2}xd_\phi^{1/2}}{\tilde{\Lip}}^2}\Big)
\end{align*}
with a different Lipschitz norm $\norm{\cdot}{\tilde{\Lip}}^2$. Our Corollary \ref{cor:concentration} here use a more natural definition of Lipschitz norm and apply to non-ergodic 
 cases. Nevertheless, the projection we have for
\[ \Prob_{\phi}(|x-E(x)|>t)\]
is not necessarily a spectral projection $e_{\{|x-E(x)|>t\}}$ and will depend on the state $\phi$. }
\end{rem}

\begin{rem}{\rm
In the operator valued setting, let $\cQ$ be any finite von Neumann algebra and $T_t\ten \id_{\cQ}$ be the amplification semigroup on $\cQ\overline{\ten} \cM$. The conditional expectation for $T_t\ten \id_{\cQ}$ is $E\ten \id_{\cQ}$. Note that by Lemma \ref{lemma:independent},
$T_t\ten \id_{\cQ}$ is GNS-symmetric to the product state $ \phi\ten \sigma$, for any state $\sigma\in S(\cQ)$ and any invariant state $\phi\in E_*(S(\M))$.
This means we obtain
 \[ \Prob_{\phi\ten \si}( |x-E_{fix}(x)|>t) \le 2 e^{-\frac{\al^2t^2}{C\norm{x}{\Lip}^2}} \]
for any product state $\phi\ten \si$ of this specific form. The projection of course depends on both $\phi$ and $\si$.
}\end{rem}

We illustrate our result with a special case as matrix concentration inequalities.
\begin{exam}[Matrix concentration inequality]\label{exam:randommatrix}{\rm Let $S_1,\cdots,S_n$ be an independent sequence of random $d\times d$-matrices $S_1,\cdots,S_n$ such that
\[ \norm{S_i-\mathbb{E}S_i}{\infty}\le M \pl , \pl a.e. \]
Tropp in \cite[Corollary 6.1.2]{tropp2015introduction} proved the following matrix Bernstein inequality that for the sum $Z=\sum_{k=1}^nS_k$,
\[ \bE\norm{Z-\mathbb{E}Z}{\infty}\le \sqrt{2v(Z)\log(2d)}+\frac{1}{3}M\log(2d)\]
and the matrix Chernoff bound
\[ P(|Z-\mathbb{E}Z|>t) \le 2d\exp\big(-\frac{t^2}{v(Z)+\frac{t}{3}M}\big)\]
where
\[v(Z)=\max \{ \norm{\bE((Z-\mathbb{E}Z)^*(Z-\mathbb{E}Z))}\pl, \norm{\bE((Z-\mathbb{E}Z)^*(Z-\mathbb{E}Z))}{}\}\pl.\]

Now to apply our results, we recall that the depolarizing semigroup with generator $L(f):=(I-E_\mu)(f)=f-\mu(f){\bf 1}_\Omega$ on any probability space $(\Omega,\mu)$
has $\al_c\ge \frac{1}{2}$ (a simple fact by convexity of relative entropy). For a random matrix $f:\Omega\to \bM_d$,
the Lipschitz norm is
\begin{align} \norm{f}{\Lip}^2= &\frac{1}{2}\max\{\norm{\hat{f}^*\hat{f}+E_\mu(\hat{f}^*\hat{f})}{\infty}\pl, \pl \norm{\hat{f}\hat{f}^*+E_\mu(\hat{f}\hat{f}^*)}{\infty}\}\nonumber\\ \le &\frac{1}{2}(\norm{f}{\infty}^2+v(f)), \label{eq:lip}\end{align}
where $\hat{f}=f-E_\mu(f)$ is the mean zero part.

Now we consider for each $k=1,\cdots, n$, $S_k:\Omega_k\to M_d$ as a random matrix on $(\Omega_k,\mu_k)$. Then on the product space $(\Omega,\mu)=(\Omega_1,\mu_1)\times\cdots  \times(\Omega_n,\mu_n)$, we have by Theorem \ref{thm:concentration} for $Z=\sum_{k}S_k$
\begin{align*}
\bE\norm{Z-\bE Z}{\infty}\le \Big ( \frac{1}{d}\bE\norm{Z-\bE Z}{p}^p \Big )^{1/p}\le d^{1/p}\norm{Z-\bE Z}{L_\infty(M_d, L_p(\Omega) )} \le 2cd^{1/p}\sqrt{p}\norm{Z}{\Lip} \pl,
\end{align*}
where $\norm{\cdot}{p}$ is the $p$-norm for the normalized trace ($\tr(1)=1$).
Applying \eqref{eq:lip} and optimizing $p$ gives
\[ \bE\norm{Z-\bE Z}{\infty}\le 2ce^{-1/2} \sqrt{ (v(Z)+M^2)\log d}\]
For the matrix Chernoff bound, we use Corollary \ref{cor:concentration}
\[ P(|Z-\mathbb{E}Z|>t)\le d\Prob_{\mu\ten \frac{\tr}{d}}(|Z-\mathbb{E}Z|>t)\le 2 d\exp\Big( -\frac{t^2}{64ec^2(v(Z)+M^2)}\Big)\pl.\]
}
\end{exam}

\section{Final Discussion}
{\bf \noindent 1. Positivity and complete positivity.} The central quantity in this work is the CB return time $t_{cb}$ and $k_{cb}$ defined via complete positivity. Alternatively, one can consider positive maps and positivite mixing time. Indeed, the entropy difference Lemma \ref{lemma:difference},
\[D(\rho\|\Phi^*\Phi(\omega))
 \le D_\Phi(\rho)  + D(\rho\|\omega)\]
holds for a positive unital trace preserving map $\Phi$. This is because, the operator concavity
\begin{align*}
\Phi(\ln x)\leq \ln \Phi(x),\quad \forall x\geq 0
\end{align*}
of the logarithmic function holds for any unital positive map $\Phi$ \cite{choi74}, and the monotonicity of relative entropy
  \begin{align*}
  D(\rho\|\sigma)\geq D(\Phi(\rho)\|\Phi(\sigma))
  \end{align*}
  was proved for any positive trace preserving map $\Phi$ in \cite{MR17} (see also \cite{frenkel2022integral}). Thus, both inequalities used in the proof of Lemma \ref{lemma:difference} hold for positive maps. Similarly, the conditions in Lemma  \ref{lemma:approximate} can be relaxed to
\begin{align} (1-\eps)E\le\Psi \le (1+\eps)E\pl, \label{eq:positive}\end{align}
where $\Phi\ge \Psi$ means $\Phi-\Psi$ is a positive map but not necessarily completely positive. Combining these two relaxed lemmas for positive maps, we have an analog of Theorem \ref{thm:main1}.
\begin{theorem}\label{thm:positive}
i) For a positive unital trace preserving map $\Phi:\M\to \M$,
\[ \al(\Phi)\le 1-\frac{1}{2k(\Phi)}\pl \text{ where } \pl  k(\Phi):=\inf\{k\in \mathbb{N}^+ \pl |\pl 0.9 E\le \Phi^{2k}\le 1.1E\}\pl\]
ii) For a trace symmetric positive unital semigroups $T_t=e^{-tL}:\M\to \M$,
\[ \al(L)\ge \frac{1}{2t(L)} \pl \text{ where } \pl t(L):=\inf\{t\in \mathbb{N}^+ \pl |\pl 0.9 E\le T_t\le 1.1E\}\]
\end{theorem}
Applying the above theorem to $\Phi\ten \id_{\mathcal{Q}}$ and $T_t\ten \id_{\mathcal{Q}}$ for any finite von Neumann algebra $\mathcal{Q}$ actually yields our main Theorem \ref{thm:main1} for trace symmetric cases. It remains open whether this observation holds for GNS-symmetric cases.

\begin{prob}Does Theorem \ref{thm:positive} with positivity conditions hold for GNS-symmetric cases?
\end{prob}

The obstruction is that in the Haagerup reduction, we need the complete positivity and CB return time $k_{cb}(\Phi)$ of $\Phi$  to imply positivity and positivity mixing time $k(\Phi)$ of the extension $\hat{\Phi}$, similar for the semigroup $T_t$. One possible approach is to avoid using Haagerup reduction, and prove Lemma \ref{lemma:difference2} directly.

The comparison between positivity and complete positivity has a deep root in the entanglement theory of  quantum physics (see \cite{chruscinski2017brief}).
From the mathematical point of view, although the positivity looks a more flexible condition, but it lacks of connection to CB norms as Proposition \ref{prop:equivalence}. Indeed, there is no non-complete version of Choi's theorem \cite{choi1975completely}
\[ C_T\in (\M\ten \M^{op})_+ \pl \Longleftrightarrow \pl  T(x)=\tau \ten \id (C_T(x\ten 1))\pl \text{ is } \pl  CP\]
Therefore, despite the estimate of $\al(\Phi)$ (resp. $\al_1(L)$) only requires $k(\Phi)$ (resp. $t(L)$), our kernel estimate Proposition \ref{prop:finitetcb} only applies to $k_{cb}(\Phi)$ (resp. $t_{cb}(L)$).\\

{\noindent\bf 2. GNS and KMS symmetry.} Both GNS-symmetry and KMS-symmetry are noncommutative generalizations for the detailed balance condition of classical Markov chains. As observed in \cite{carlen2017gradient}, GNS-symmetry is the strongest generalization of detail balance condition and KMS is the weakest, which means the assumption of GNS-symmetry is the most restrictive. It is natural to ask whether our main results (c.f. Theorem \ref{thm:main6} \& \ref{thm:GNSsymmetric}) can be obtained for KMS-symmetric channels or semigroups.

\begin{prob}Do entropy decay results Theorem \ref{thm:main6} \& \ref{thm:GNSsymmetric} or the entropy difference Lemma \ref{lemma:difference2} hold for KMS-symmetric maps?
\end{prob}

The key property of a GNS-symmetric map $\Phi$ is the commutation with modular group $\Phi\circ \al_t^\phi=\al_t^\phi \circ \Phi$. This has been used to ensure the compatibility of Haagerup reduction with channel and semigroups (see Lemma 4.5). It is interesting to see whether the same commuting diagram Figure \ref{figure} can be obtained for KMS Markov maps. That will allow us to use Haagerup reduction to obtain the entropy difference Lemma \eqref{lemma:difference2} for KMS-symmetric channels. Another approach is, again, to avoid using Haagerup reduction and prove the KMS-case directly. At the moment of writing, this is not unclear to us even on finite dimensional matrix algebras.\\

{\noindent\bf 3. MLSI and CMLSI constant.}
By the results of this work and also previous work \cite{li2020graph,brannan2022complete,gao2022complete2}, we now know the positivity of CMLSI constant $\al_c>0$ for many cases of classical Markov semigroups with the (non-complete) MLSI constant $\al>0$. That is, $\al\ge \al_c>0$ for
\begin{enumerate}
\item[i)] Finite Markov chains \cite{li2020graph, gao2022complete2};
\item[ii)] Heat semigroups on manifold with curvature lower bound \cite{brannan2022complete};
\item[iii)] Sub-Laplacians of H\"ormander system on a compact Riemannian manifold.
\end{enumerate}
It remains open that whether MLSI constant $\al$ and CMLSI constant $\al_c$ coincide for classical semigroups. This would be in the similar spirit that the bounded norm (resp. positivity) and the complete bounded norm (resp. complete positivity) coincide for a classical map on $L_\infty(\Omega,\mu)$

\begin{prob}Does $\al=\al_c$ for a classical symmetric Markov semigroup $T_t:L_\infty(\Omega,\mu)\to L_\infty(\Omega,\mu)$?
\end{prob}
For quantum Markov semigroup, the counterexample can be a qubit depolarizing semigroup
\[T_t:\bM_2\to \bM_2 \pl , \pl T_t(\rho)=e^{-t}\rho+(1-e^{-t})\frac{1}{2},\]
which has $\frac{1}{2}\le \al_c(T_t)<\al(T_t)=1$ because of entangled states \cite[Section 4.3]{brannan2022complete}. It nature to ask whether $\al_c<\al$ also holds for classical depolarizing channel.

Another interesting example is the heat semigroup on the unit torus $\mathbb{T}=\{z\in \mathbb{C}\pl | \pl  |z|=1\}$,
\[P_t: L_\infty(\mathbb{T})\to L_\infty(\mathbb{T})\pl, P_t(z^n)=e^{-n^2t}z^n\pl.\]
It was proved by \cite{weissler1980logarithmic} that $\al(P_t)=\lambda(P_t)=1$. The best known bound for CMLSI is  $\al_c(P_t)\ge \frac{1}{6}$. It is open that whether the gap can be closed.

\begin{prob}Does the heat semigroup $P_t$ on the torus $\mathbb{T}$ have $\al_c(P_t)=\al(P_t)=1$?
\end{prob} 

\bibliography{factorizable}
\bibliographystyle{plain}
\end{document}